\documentclass[11pt]{article}

\input{Qcircuit}
\input{ArticlePreamble.sty}

\usepackage{authblk}
\usepackage{cleveref}
\usepackage{mathrsfs}

\def\be{\begin{eqnarray}}

\def\ee{\end{eqnarray}}

\def\bee{\begin{eqnarray*}}

\def\eee{\end{eqnarray*}}

\newtheorem{thm}{Theorem}

\newtheorem{conj}{Conjecture}

\theoremstyle{remark}

\newtheorem*{alg_in}{Input}

\newtheorem*{alg_out}{Output}

\begin{document}

\title{The computational power of  normalizer circuits  \\ over black-box groups}
\author{Juan Bermejo-Vega, Cedric Yen-Yu Lin, Maarten Van den  Nest}

\maketitle

\begin{abstract}
This work presents a precise connection between Clifford circuits,  Shor's factoring algorithm and several other famous quantum algorithms with exponential quantum speed-ups for solving Abelian hidden subgroup problems. We show that all these different forms of quantum computation belong to a common new restricted model of quantum operations that we call \emph{black-box normalizer circuits}. To define these, we extend the previous model of normalizer circuits  \cite{VDNest_12_QFTs, BermejoVega_12_GKTheorem,BermejoLinVdN13_Infinite_Normalizers}, which are built of quantum Fourier transforms, group automorphism and quadratic phase gates associated with an Abelian group $G$. In \cite{VDNest_12_QFTs, BermejoVega_12_GKTheorem,BermejoLinVdN13_Infinite_Normalizers}, the group is always given in an explicitly decomposed form. In our model, we remove this assumption and allow $G$ to be a black-box group \cite{BabaiSzmeredi_Complexity_MatrixGroup_Problems_I}. While standard normalizer circuits were shown to be efficiently classically simulable \cite{VDNest_12_QFTs, BermejoVega_12_GKTheorem,BermejoLinVdN13_Infinite_Normalizers}, we find that normalizer circuits are powerful enough to factorize and solve classically-hard problems in the black-box setting. We further set upper limits to their computational power by showing that decomposing finite Abelian groups is complete for the associated complexity class. In particular, solving this problem renders black-box normalizer circuits efficiently classically simulable by exploiting the generalized stabilizer formalism  in \cite{VDNest_12_QFTs, BermejoVega_12_GKTheorem,BermejoLinVdN13_Infinite_Normalizers}. Lastly, we employ our connection to draw a few practical implications for quantum algorithm design: namely,  we give a no-go theorem for finding new quantum algorithms with black-box normalizer circuits, a universality result for low-depth normalizer circuits, and identify two other complete problems.
\end{abstract}

\tableofcontents

\section{Introduction}
 
\emph{What are the potential uses and limitations of quantum computers?} Arguably, the most attractive feature of quantum computers is their ability to efficiently solve problems for which no efficient classical solution is known. To date, a successful number of quantum algorithms have been discovered  \cite{Mosca09_Quantum_Algorithms_REVIEW,childs_vandam_10_qu_algorithms_algebraic_problems,Bacon10_Recent_Progres_Quantum_Algorithms,VanDamSasaki12_Q_algorithms_number_theory_REVIEW,MoscaSmith12_Algorithms_Quantum_Computers_REVIEW,Jordan_Quantum_Algorithm_Zoo}. Yet it remains  one of the greatest challenges of the field of quantum computing to understand for which precise problems quantum algorithms can be \emph{exponentially} (or \emph{super-polynomially}) faster than their classical counterparts. Some reasons why this question seems  to be hard to answer have been discussed by Shor \cite{Shor_04_Progress_Quantum_Algorithms}.

A fruitful approach to understand the emergence and the structure of exponential quantum speed-ups  is to study \emph{restricted} models of quantum computation. Ideally, the latter should exhibit interesting quantum features  and, at the same time, have less power than universal quantum computers (up to reasonable computational complexity assumptions). To date, several  models studied in the literature seem to have these desirable properties, including Clifford circuits
\cite{Gottesman_PhD_Thesis,Gottesman99_HeisenbergRepresentation_of_Q_Computers,Knill96non-binaryunitary,Gottesman98Fault_Tolerant_QC_HigherDimensions}, nearest-neighbor matchgates \cite{Valiant02_matchgates,Knill01_Fermionic_Linear_Optics,Terhal02_Simulation_noninteracting_fermion_circuits,Jozsa08_Matchgates_classical_simulation},  Gaussian operations  \cite{LloydBraunstein99_QC_over_CVs,Bartlett02Continuous-Variable-GK-Theorem,BartlettSanders02Simulations_Optical_QI_Circuits},  the one-clean qubit (DQC1) model \cite{KnillLaflamme98_DQC1}, and commuting circuits \cite{Shepherd10_PhD_thesis,ShepherdBremner09_Temporally_Unstructured_QC,BremnerJozsaShepherd08,Ni13Commuting_Circuits} (a more complete list is given at the end of this section). 

In this work, we introduce \emph{\textbf{black-box normalizer circuits}}, a \emph{new restricted family} of quantum operations, and characterize their computational power. Our circuit model extends the so-called model of  \emph{\textbf{normalizer circuits over Abelian groups}} \cite{VDNest_12_QFTs,BermejoVega_12_GKTheorem,BermejoLinVdN13_Infinite_Normalizers}, which, in turn, generalizes the  better-known model of  Clifford circuits \cite{Gottesman_PhD_Thesis,Gottesman99_HeisenbergRepresentation_of_Q_Computers,Knill96non-binaryunitary,Gottesman98Fault_Tolerant_QC_HigherDimensions}. In \cite{VDNest_12_QFTs,BermejoVega_12_GKTheorem,BermejoLinVdN13_Infinite_Normalizers}, normalizer circuits act in  high and infinite dimensional systems associated with an Abelian group $G$. The Hilbert space $\mathcal{H}$ has a standard basis $\{\ket{g}\}_{g\in G}$ labeled by the elements of $G$. The allowed operations, called \emph{normalizer gates}, can be of three types.
\begin{enumerate*}
\item Quantum Fourier transforms over subgroups of $G$.
\item Group automorphism gates: permutation-like gates that implement automorphisms of the group $G$.
\item Quadratic phase gates: diagonal gates which multiply standard basis states by quadratic phases.
\end{enumerate*}
In all previous work \cite{VDNest_12_QFTs, BermejoVega_12_GKTheorem,BermejoLinVdN13_Infinite_Normalizers}, the group $G$ is assumed to be given in a factorized form (\ref{eq:Elementary Group}), which endows the Hilbert space of the computation  with a tensor-product structure:
\begin{equation}\label{eq:Elementary Group}
G=\Z^a\times \T^b\times \DProd{N}{c}\quad \longleftrightarrow \quad \mathcal{H} =\mathcal{H}_{\Z}^{\otimes (a+b)}\otimes \mathcal{H}_{N_1}\otimes \cdots \otimes \mathcal{H}_{N_c}.
\end{equation} 
In (\ref{eq:Elementary Group}) $\Z$ is the group of \emph{integers}, $\Z_N$ the group of integers modulo $N$,  and $\T$ is the \emph{circle group}, consisting of angles from 0 to 1 (in  units of $2\uppi$) with the addition modulo 1. The Hilbert space $\mathcal{H}_\Z$ has a standard basis labeled by integers ($\Z$ basis) and a Fourier-basis labeled by angles ($\T$ basis).

A \emph{normalizer circuit over $G$}  \cite{VDNest_12_QFTs, BermejoVega_12_GKTheorem,BermejoLinVdN13_Infinite_Normalizers} is any quantum circuit composed of  normalizer gates. In particular, $n$-qubit Clifford circuits are normalizer circuits over the group $\Z_2^n$.

Despite containing arbitrary numbers of quantum Fourier transforms (which play an important role in Shor's algorithms \cite{Shor}) and entangling gates (automorphism, quadratic phase gates), it was shown in \cite{VDNest_12_QFTs, BermejoVega_12_GKTheorem,BermejoLinVdN13_Infinite_Normalizers} that normalizer circuits can be \emph{efficiently simulated} by classical computers. This result exploits an extended \textbf{stabilizer formalism} to track the evolution of normalizer circuits and generalizes the celebrated Gottesman-Knill theorem \cite{Gottesman_PhD_Thesis,Gottesman99_HeisenbergRepresentation_of_Q_Computers}.

The key new element in our work are normalizer circuits that can be associated with Abelian \emph{\textbf{black-box groups}} \cite{BabaiSzmeredi_Complexity_MatrixGroup_Problems_I}, which we may simply call ``black-box normalizer circuits''. An  group $\mathbf{B}$ (always Abelian in this work) is a black-box group if it is \emph{finite}, its elements are uniquely encoded by strings of some length $n$ and the group operation is performed by a black-box (the \emph{group oracle}) in one time-step. We define  \emph{black-box normalizer circuits} to be a normalizer circuits associated with groups of the form $G=G_\textrm{prev}\times \mathbf{B}$, where   $G_\textrm{prev}$ is of form (\ref{eq:Elementary Group}). 

The \textbf{\emph{key new feature}} in our work is that the black-box  group $\mathbf{B}$ is \emph{not} given to us in a factorized form. This is a subtle yet tremendously important difference: although such a  decomposition \emph{always} exists  for any finite Abelian group \cite{Humphrey96_Course_GroupTheory}, finding just one is regarded as a \emph{hard computational problem}; indeed, it is provably at least as hard as \emph{factoring}\footnote{\label{footnote:Hardness Group Decomposition}Knowing $\mathbf{B}\cong \DProd{d}{m}$ implies that the order of the group $|G|=d_1 d_2\cdots d_m$. Hardness results for computing orders \cite{BabaiSzmeredi_Complexity_MatrixGroup_Problems_I,Babai98apolynomial_time_theory_of_Black_Box_Groups} imply that the problem is provably hard for classical computers in the black-box setting. For groups  $\Z_N^\times$, computing $\varphi(N):=|\Z_N^\times|$ (the Euler totient function) is equivalent to factoring \cite{Shoup08_A_Computational_Introducttion_to_Number_Theory_and_Algebra}.}. Our \textbf{\emph{motivation}} to adopt the notion of black-box group is to study Abelian groups for which the group multiplication can be performed in classical polynomial-time while no efficient classical algorithm to decompose them is known. A key example is $\mathbb{Z}^{\times}_N$, the multiplicative group of integers modulo\cref{footnote:Hardness Group Decomposition}, which  plays an important role in Shor's factoring algorithm \cite{Shor}. With some abuse of notation, we  call any such group also a ``black-box group''\footnote{It will always be clear from context whether the group multiplication is performed by an oracle at unit cost or by some well-known polynomial-time classical algorithm. Most results will be stated in the black-box setting though.}.

\subsection*{Statement of results}

This work focuses on understanding the potential uses and limitations of black-box normalizer circuits. Our results (listed below) give a precise characterization of their \textbf{computational power}. On one hand, we show that several famous quantum algorithms, including shor's celebrated \emph{factoring algorithm}, can be implemented with black-box normalizer circuits. On the other hand, we apply  our former simulation results \cite{VDNest_12_QFTs, BermejoVega_12_GKTheorem,BermejoLinVdN13_Infinite_Normalizers}  to set upper limits to the class of problems that these circuits can solve, as well as to draw practical implications for quantum algorithm design.

Our main results are now summarized:

\begin{enumerate}
\item \textbf{Quantum algorithms.} We show that many of the best known quantum algorithms  are particular instances of normalizer circuits over black-box groups, including  Shor's celebrated factoring and discrete-log algorithms; it follows that black-box normalizer circuits can achieve \emph{\textbf{exponential quantum speed-ups}}. Namely, the following algorithms are examples of black-box normalizer circuits.
\begin{itemize}
\item \textbf{Discrete logarithm.} Shor's discrete-log quantum algorithm  \cite{Shor} is a normalizer circuit over $\Z_{p-1}^2\times\Z_p^\times$ (theorem \ref{thm:Discrete Log} \ref{sect:Discrete Log}).
\item \textbf{Factoring.} We show that  a hybrid infinite-finite dimensional version of Shor's factoring algorithm \cite{Shor}  can be implemented with normalizer circuit over $\Z\times\Z_N^\times$. We prove that there is a close relationship between \emph{Shor's original algorithm} and our version: Shor's  can be understood as a discretized qubit implementation of ours (theorems \ref{thm:Order Finding}, \ref{thm:Shor normalizer}).  We also discuss that the \emph{\textbf{infinite group}} $\Z$ plays a key role in our ``infinite Shor's algorithm'', by showing that it is impossible to implement Shor's modular-exponentiation gate efficiently, {even approximately}, with finite-dimensional normalizer circuits (theorem \ref{thm:ModExp requires Z}). Last,  we further \emph{conjecture} that only normalizer circuits over {infinite groups} can factorize (conjecture \ref{conj: Factoring not over finite groups}).
\item \textbf{Elliptic curves.} The generalized Shor's algorithm for computing discrete logarithms over an elliptic curve \cite{ProosZalka03_Shors_DiscreteLog_Elliptic_Curves,Kaye05_optimized_Quantum_Elliptic_Curve,CheungMaslovMathew08_Design_QuantumAttack_Elliptic_CC} can be implemented with black-box normalizer circuits (section \ref{sect:Elliptic Curve}); in this case, the black-box group is the group of integral points $E$ of the elliptic curve instead of $\Z_p^\times$.
\item \textbf{Group decomposition.}  Cheung-Mosca's algorithm for decomposing black-box finite Abelian groups \cite{mosca_phd, cheung_mosca_01_decomp_abelian_groups} is a combination of several types of black-box normalizer circuits. In fact, we discuss a new \emph{extended  Cheung-Mosca's algorithm} that finds even more information about the structure of the group and it is also based on normalizer circuits (section \ref{sect:Group Decomposition}).
\item \textbf{Hidden subgroup problem.} Deutsch's \cite{Deutsch85quantumtheory}, Simon's  \cite{Simon94onthe} and, in fact, all quantum algorithms that solve Abelian hidden subgroup problems \cite{Boneh95QCryptanalysis,Grigoriev97_testing_shift_equivalence_polynomials,kitaev_phase_estimation,Kitaev97_QCs:_algorithms_error_correction,Brassard_Hoyer97_Exact_Quantum_Algorithm_Simons_Problem,Hoyer99Conjugated_operators,MoscaEkert98_The_HSP_and_Eigenvalue_Estimation,Damgard_QIP_note_HSP_algorithm}, are normalizer circuits over groups of the form $G\times\mathcal{O}$, where $G$ is the group that contains the hidden subgroup $H$ and $\mathcal{O}$ is a group isomorphic to $G/H$ (section \ref{sect:Abelian HSPs}). The group $\mathcal{O}$, however, is not a black-box group due to a small technical difference between our  oracle model we use and the oracle setting in the HSP.
 
\item \textbf{Hidden kernel problem.}  The group $\mathcal{O} \cong G/H$ in the previous section becomes a black-box group if the oracle function in the HSP is a homomorphism between black-box groups: we call this subcase the {\emph{hidden kernel problem}} (HKP). The difference does not seem to be very significant, and can be eliminated by choosing different oracle models (section \ref{sect:Abelian HSPs}). However, we will never refer to Simon's or to general Abelian HSP algorithms  as ``black-box normalizer circuits'', in order to be consistent with our and pre-existing terminology. 
\end{itemize}
Note that it follows from the above that black-box normalizer circuits can render insecure widespread public-key cryptosystems, namely,  Diffie-Hellman key-exchange  \cite{DiffieHellman}, RSA \cite{RSA} and elliptic curve cryptography  \cite{Menezes96_cryptography_book,Buchmann00_cryptography_book}.

\item \textbf{Group decomposition is \emph{as hard as} simulating normalizer circuits.} Another main contribution of this work is to show that the group decomposition problem (suitably formalized) is, in fact, \textbf{\emph{complete}} for the complexity class \textbf{Black-Box Normalizer}, of problems efficiently solvable by probabilistic classical computers with oracular access to black-box normalizer circuits. Since normalizer circuits over decomposed groups are efficiently classically simulable \cite{VDNest_12_QFTs, BermejoVega_12_GKTheorem,BermejoLinVdN13_Infinite_Normalizers}, this result suggests that the computational power of normalizer circuits  originate \emph{precisely} in the classical hardness of learning the structure of a black-box group. 

We obtain this last  result by proving a significantly \textbf{\emph{stronger theorem}} (theorem \ref{thm:Simulation}), which states that any black-box normalizer circuit  can be efficiently simulated \emph{step by step} by a classical computer if  an efficient subroutine for decomposing finite Abelian groups is provided.

\item \textbf{A no-go theorem for new quantum algorithms.} In this work, we provide an negative answer to the question ``\emph{can new quantum algorithms based on normalizer circuits be found?}'': by applying the latter simulation result, we conclude that any new algorithm not in our list can be efficiently simulated step-by-step using the extended Cheung-Mosca algorithm and classical post-processing. This implies (theorem \ref{thm:No Go Theorem}) that new \emph{exponential} speed-ups cannot be found without changing our setting (we discuss how the setting might be changed in the discussion \ref{sect:Discussion}). This result says nothing about polynomial speed-ups.

\item \textbf{Universality of short normalizer circuits.} A practical consequence of our no-go theorem is that all problems in the class \textbf{Black Box Normalizer} can be solved using short normalizer circuits with a \emph{constant} number of normalizer gates. (We may still need polynomially many runs of such circuits, along with classical processing in between, but each individual normalizer circuit is short.) We find this observation interesting, in that it explains a very curious feature present in all the quantum algorithms that we study \cite{Shor,ProosZalka03_Shors_DiscreteLog_Elliptic_Curves,Kaye05_optimized_Quantum_Elliptic_Curve,CheungMaslovMathew08_Design_QuantumAttack_Elliptic_CC,mosca_phd, cheung_mosca_01_decomp_abelian_groups,Deutsch85quantumtheory,Simon94onthe,Boneh95QCryptanalysis,Grigoriev97_testing_shift_equivalence_polynomials,kitaev_phase_estimation,Kitaev97_QCs:_algorithms_error_correction,Brassard_Hoyer97_Exact_Quantum_Algorithm_Simons_Problem,Hoyer99Conjugated_operators,MoscaEkert98_The_HSP_and_Eigenvalue_Estimation,Damgard_QIP_note_HSP_algorithm} (section \ref{sect:Quantum Algorithms}): they all contain at most a constant number of \emph{quantum Fourier transforms} (actually at most two).
\item \textbf{Other complete problems.} As our last contribution in this series, we identify another two complete problems for the class \textbf{Black Box Normalizer} (section \ref{sect:Complete Problems}): these are the (afore-mentioned) Abelian \emph{hidden kernel problem}, and  the problem of finding a general-solution to a \emph{system of linear equations over black-box groups} (the latter are related to the systems of linear equations over groups  studied in \cite{BermejoVega_12_GKTheorem,BermejoLinVdN13_Infinite_Normalizers}.
\end{enumerate}

\subsection*{{A link between Clifford circuits and Shor's  algorithm}}

The results in this work together with those previously obtained in \cite{VDNest_12_QFTs, BermejoVega_12_GKTheorem,BermejoLinVdN13_Infinite_Normalizers} demonstrate the existence of a precise connection between Clifford circuits and  Shor's factoring algorithm. At first glance, it might be hard to digest that two types of quantum circuits that seem to be so far away from each other might related at all. Indeed, classically simulating Shor's algorithm is widely believed to be an intractable problem (at least as hard as factoring), while a zoo of classical techniques and efficient classical algorithms exist for simulating and computing properties of Clifford circuits \cite{Gottesman_PhD_Thesis,Gottesman99_HeisenbergRepresentation_of_Q_Computers,Knill96non-binaryunitary,Gottesman98Fault_Tolerant_QC_HigherDimensions,dehaene_demoor_coefficients,AaronsonGottesman04_Improved_Simul_stabilizer,dehaene_demoor_hostens,AndersBriegel06_Simulation_Stabilizer_GraphStates,VdNest10_Classical_Simulation_GKT_SlightlyBeyond,deBeaudrap12_linearised_stabiliser_formalism,JozsaVdNest14_Classical_Simulation_Extended_Clifford_Circuits}. However, from the point of view of this paper, both turn out to be \emph{intimately related} in that they  both are just different types of normalizer circuits. In other words, they are both \emph{members of a common family of quantum operations}.

Remarkably, this correspondence between Clifford and Shor, rather than being just a mere mathematical curiosity, has also some  sensible consequences for the theory of quantum computing. One that follows from theorem \ref{thm:Simulation}, our simulation result, is that all algorithms studied in this work (Shor's factoring and discrete-log algorithms,  Cheung-Mosca's, etc.) have a \emph{rich hidden structure} which enables simulating them classically  with a stabilizer picture approach ``à la Gottesman-Knill'' \cite{Gottesman_PhD_Thesis,Gottesman99_HeisenbergRepresentation_of_Q_Computers}. This structure let us track the evolution of the quantum state of the computation \emph{step by step} with a very special algorithm, which, despite being inefficient, exploits \emph{completely different} algorithmic principles than the naive brute-force approach: i.e.,\ writing down the coefficients of the initial quantum state and tracking their quantum mechanical evolution through the gates of the circuit\footnote{Note that throughout this manuscript we always work at a high-level of abstraction (algorithmically speaking), and that the ``steps'' in a normalizer-based quantum algorithm are always counted at the logic level of normalizer gates, disregarding smaller gates needed to implement them. In spite of this, we find the above  simulability property of black-box normalizer circuits to be truly fascinating. To get a better grasp of its significance, we may 	perform the following  thought experiment. Imagine, we would repeatedly concatenate black-box normalizer circuits in some intentionally complex geometric arrangement, in order to form a gargantuan, intricate ``Shor's algorithm'' of monstrous size. Even in this case, our simulation result states that if we can decompose Abelian groups (say, with an oracle), then we can efficiently simulate the evolution of the circuit, normalizer-gate after normalizer-gate, independently of the number of Fourier transforms, automorphism and quadratic-phase gates involved in the computation (the overhead of the classical simulation is always at most polynomial in the input-size).}. Although the stabilizer-picture simulation is \emph{inefficient}  when black-box groups are present (i.e., it does not yield an efficient classical algorithm for simulating Shor's algorithm), the mere existence of such an algorithm reveals how much mathematical structure these quantum algorithms have in common with Clifford and normalizer circuits.

In retrospect, and from an applied point of view, it is also rather satisfactory that one can gracefully exploit the above connection to draw practical implications for quantum algorithm design:  in our work, we have actively used our knowledge of the hidden ``Clifford-ish'' mathematical features of the Abelian hidden subgroup problem algorithms in deriving results 2, 3, 4 and 5 (in the list given in the previous section).
 
As a side remark, we regard it a memorable curiosity that replacing decomposed groups with black-box groups not only renders the simulation methods in \cite{VDNest_12_QFTs,BermejoVega_12_GKTheorem,BermejoLinVdN13_Infinite_Normalizers}   inefficient (this is, in fact, something to be expected, due to the existence of hard computational problems related to black-box groups), but it is also precisely this modification that suddenly bridges the gap between Clifford/normalizer circuits, Shor's algorithms, Simon's and so on. 

Finally, it is  mathematically elegant to note that all normalizer circuits we have studied are related through the so-called \textbf{Pontryagin-Van Kampen duality} \cite{Morris77_Pontryagin_Duality_and_LCA_groups,Stroppel06_Locally_Compact_Groups,Dikranjan11_IntroTopologicalGroups,rudin62_Fourier_Analysis_on_groups,HofmannMorris06The_Structure_of_Compact_Groups,Armacost81_Structure_LCA_Groups,Baez08LCA_groups_Blog_Post}, which states that all locally-compact Abelian (LCA) groups are dual to their characters groups. The role of this duality in the normalizer circuit model was discussed in our previous work \cite{BermejoLinVdN13_Infinite_Normalizers}.

\subsection*{Relationship to previous work}\label{sect:Relationship Previous Work}

Up to our best knowledge,  neither normalizer circuits over black-box groups, nor their relationship with Shor's algorithm or the Abelian hidden subgroup problem,  have been previously investigated. Normalizer circuits over explicitly-decomposed finite groups $\DProd{N}{a}$ were studied in \cite{VDNest_12_QFTs, BermejoVega_12_GKTheorem}, by two of us. We recently extended the formalism in \cite{BermejoLinVdN13_Infinite_Normalizers} to infinite groups of the form $\Z^a\times \T^b \times \DProd{N}{a}$.

Clifford circuits over qubits and qudits (which can be understood as normalizer circuits over groups of the form $\Z_2^m$ and $\Z_d^m$) have been extensively investigated in the literature \cite{Gottesman_PhD_Thesis,Gottesman99_HeisenbergRepresentation_of_Q_Computers,Knill96non-binaryunitary,Gottesman98Fault_Tolerant_QC_HigherDimensions,dehaene_demoor_coefficients,AaronsonGottesman04_Improved_Simul_stabilizer,dehaene_demoor_hostens,deBeaudrap12_linearised_stabiliser_formalism,VdNest10_Classical_Simulation_GKT_SlightlyBeyond,JozsaVdNest14_Classical_Simulation_Extended_Clifford_Circuits}. Certain generalizations of Clifford circuits that are not normalizer circuits have also been studied: \cite{AaronsonGottesman04_Improved_Simul_stabilizer,BravyiKitaev05MagicStateDistillation,Jozsa08_Matchgates_classical_simulation,VdNest10_Classical_Simulation_GKT_SlightlyBeyond,JozsaVdNest14_Classical_Simulation_Extended_Clifford_Circuits} consider Clifford circuits supplemented with some non-Clifford ingredients;  a different form of  Clifford circuits based
on projective normalizers of unitary groups
 were investigated in \cite{ClarkJozsaLinden08Generalized_Clifford_Groups}.
  
The hidden subgroup problem (HSP) has played a central role in the history of quantum algorithms and has been extensively studied before our work. The Abelian HSP, which is also a central subject of this work, is related to most of the best known quantum algorithms that were found in the early days of the field \cite{Deutsch85quantumtheory,Simon94onthe,Boneh95QCryptanalysis,Grigoriev97_testing_shift_equivalence_polynomials,kitaev_phase_estimation,Kitaev97_QCs:_algorithms_error_correction,Brassard_Hoyer97_Exact_Quantum_Algorithm_Simons_Problem,Hoyer99Conjugated_operators,MoscaEkert98_The_HSP_and_Eigenvalue_Estimation,Damgard_QIP_note_HSP_algorithm}. Its best-known generalization, the non-Abelian HSP, has also been heavily investigated due to its relationship to the graph isomorphism problem and certain shortest-vector-lattice problems \cite{EttingerHoyerKnill2004_Hidden_Subgroup,HallgrenRusselTaShma2003_Normal_Subgroup,Kuperberg2005_Dihedral_Hidden_Subgroup,Regev2004_Dihedral_Hidden_Subgroup,Kuperberg2013_Hidden_Subgroup,RoettelerBeth1998_Hidden_Subgroup,IvanyosMagniezSantha2001_Hidden_Subgroup,MooreRockmoreRussellSchulman2004,InuiLeGall2007_Hidden_Subgroup,BaconChildsVDam2005_Hidden_Subgroup,ChiKimLee2006_Hidden_Subgroup,IvanyosSanselmeSantha2007_Hidden_Subgroup,MagnoCosmePortugal2007_Hidden_Subgroup,IvanyosSanselmeSantha2007_Nil2_Groups,FriedlIvanyosMagniezSanthaSen2003_Hidden_Translation,Gavinsky2004_Hidden_Subgroup,ChildsVDam2007_Hidden_Shift,DenneyMooreRussel2010_Conjugate_Stabilizer_Subgroups,Wallach2013_Hidden_Subgroup} (see also the reviews \cite{lomont_HSP_review,childs_lecture_8,VanDamSasaki12_Q_algorithms_number_theory_REVIEW} and references therein).
  
The notion of black-box group, which is a key concept in our setting, was first considered by Babai and Szméredi in  \cite{BabaiSzmeredi_Complexity_MatrixGroup_Problems_I} and have since been extensively studied in classical complexity theory  \cite{Arvind97solvableblack-box,Babai1991_Vertex_Transive_Graphs_Random_Generation_Finite_Groups,Babai1992_Bounded_Round_Interactive_Proofs_Finite_Groups,Babai97_Randomization_group_algorithms,Babai98apolynomial_time_theory_of_Black_Box_Groups}. In general, black-box groups may not be Abelian and do not need to have uniquely represented elements \cite{BabaiSzmeredi_Complexity_MatrixGroup_Problems_I}; in the present work, we only consider Abelian uniquely-encoded black-box groups. 

In quantum computing, black-box groups were previously investigated in the context of quantum algorithms, both in the Abelian \cite{mosca_phd,cheung_mosca_01_decomp_abelian_groups,Zhang11Decomposing} and the non-Abelian group setting \cite{watrous00_quantumAlgorithms_solvableGroups,IvanyosMagniezSantha2001_Hidden_Subgroup,FriedlIvanyosMagniezSanthaSen2003_Hidden_Translation,MagniezNayak2005_Group_Commutativity,Fenner05_QAlg_Group_Theoretic_Problems,IvanyosSanselmeSantha2007_Nil2_Groups,LeGall2010_Group_Isomorphism,Zatloukal2013_Equivalent_Group_Extensions}. Except for a few exceptions (cf.\ \cite{watrous00_quantumAlgorithms_solvableGroups,Zhang11Decomposing}) most quantum results have been obtained for uniquely-encoded black-box groups. 

Aside from generalizations of Clifford circuits \cite{VDNest_12_QFTs, BermejoVega_12_GKTheorem,BermejoLinVdN13_Infinite_Normalizers,Gottesman_PhD_Thesis,Gottesman99_HeisenbergRepresentation_of_Q_Computers,Knill96non-binaryunitary,Gottesman98Fault_Tolerant_QC_HigherDimensions,dehaene_demoor_coefficients,AaronsonGottesman04_Improved_Simul_stabilizer,dehaene_demoor_hostens,AndersBriegel06_Simulation_Stabilizer_GraphStates,VdNest10_Classical_Simulation_GKT_SlightlyBeyond,deBeaudrap12_linearised_stabiliser_formalism,JozsaVdNest14_Classical_Simulation_Extended_Clifford_Circuits} (which includes normalizer circuits), many other classes of restricted quantum circuits have been studied in the literature. Some examples (by no means meant to be an exhaustive list) are nearest-neighbor matchgate circuits \cite{Valiant02_matchgates,Knill01_Fermionic_Linear_Optics,Terhal02_Simulation_noninteracting_fermion_circuits,Jozsa08_Matchgates_classical_simulation,Bravyi05_Lagrangian_Rep_Fermionic_Linear_Optics,JozsaKrausMiyakeWatrous10Matchgates,VdNest11_Matchgates,BravyiKoenig12_Simulation_Dissipative_Fermionic_Linear_Optics,deMeloCwiklinskiTerhal13_Noisy_Fermionic_Quantum_Computation}, the one-clean qubit model \cite{AmbainisShulmanVazirani06_Computing_highly_mixed_states,Pouline03_Integrability_DQC1,Poulin04_Fidelity_Decay_DQC1,Shepherd06_DQC1,ShorJordan08_jones_polynomial_Complete_DQC1,JordanWocjan09_DQC1_Jones_Homfly_polynomials,jordan2014approximating_Turaev_Viro_DQC1,MorimaeFujiiFitzsimons14_Hardness_Simulating_DQC1}, circuit models based on Gaussian or linear-optical operations
\cite{LloydBraunstein99_QC_over_CVs,Bartlett02Continuous-Variable-GK-Theorem,BartlettSanders02Simulations_Optical_QI_Circuits,Aaronson11_Computational_Complexity_Linear_Optics,Veitch12_Negative_QuasiProbability_Resource_QC,MariEisert12_Positive_Wigner_Functions_Quantum_Computation,VeitchWiebeFerrieEmerson13_Simulation_scheme_large_class_quantum_optics_experiments}, commuting circuits \cite{Shepherd10_PhD_thesis,ShepherdBremner09_Temporally_Unstructured_QC,BremnerJozsaShepherd08,Ni13Commuting_Circuits}, low-entangling\footnote{Here entanglement is measured with respect to the Schmidt-rank measure ( low-entangling circuits with respect to continuous entanglement measures are universal for quantum computation \cite{VDNest2012_Little_Entanglement}).} circuits \cite{Jozsa03_Role_Entanglement_Quantum_Computational_SpeedUP,Vidal_03_Efficient_Simulation_Sligtly_Entangled} ,
low-depth circuits \cite{TerhalDiVincenzo02Adaptive_ConstantDepth_Quantum_Comp,MarkovShi08_Simulating_QuantumComp_TensorNetwork},
tree-like circuits \cite{MarkovShi08_Simulating_QuantumComp_TensorNetwork,aharonov_AQFT,yoran_short_QFT,browne_QFT,Yoran08_Contractable_circults_little_entanglement}, low-interference circuits \cite{nest_weak_simulations,Stahlke14_Interference_resource_speedup}  and a few others \cite{Jordan10_Permutational_Quantum_Computing,schwarz2013simulating}.

\subsection*{Discussion and outlook}\label{sect:Discussion}

We finish our introduction by discussing a few potential avenues for finding new quantum algorithms as well as some open questions suggested by our work.

In this work, we provide a strict no-go theorem for finding new quantum algorithms with black-box normalizer circuits, as we define them. There are, however, a few possible ways to  modify our setting leading to scenarios where one could bypass these results and, indeed, find new interesting quantum algorithms. We now discuss some.

One enticing possibility would be to study possible extensions of the normalizer circuit framework to non-Abelian groups, in connection with non-Abelian hidden subgroup problems \cite{EttingerHoyerKnill2004_Hidden_Subgroup,HallgrenRusselTaShma2003_Normal_Subgroup,Kuperberg2005_Dihedral_Hidden_Subgroup,Regev2004_Dihedral_Hidden_Subgroup,Kuperberg2013_Hidden_Subgroup,RoettelerBeth1998_Hidden_Subgroup,IvanyosMagniezSantha2001_Hidden_Subgroup,MooreRockmoreRussellSchulman2004,InuiLeGall2007_Hidden_Subgroup,BaconChildsVDam2005_Hidden_Subgroup,ChiKimLee2006_Hidden_Subgroup,IvanyosSanselmeSantha2007_Hidden_Subgroup,MagnoCosmePortugal2007_Hidden_Subgroup,IvanyosSanselmeSantha2007_Nil2_Groups,FriedlIvanyosMagniezSanthaSen2003_Hidden_Translation,Gavinsky2004_Hidden_Subgroup,ChildsVDam2007_Hidden_Shift,DenneyMooreRussel2010_Conjugate_Stabilizer_Subgroups,Wallach2013_Hidden_Subgroup}. We have not addressed this question in the present work. In this direction, the classical simulability of non-Abelian quantum Fourier transforms was studied in \cite{bermejo2011classical} by one of us.

A second possibility would be to consider more general types of normalizer circuits than ours, by \textbf{\emph{extending the class of Abelian groups}} they can be associated with. However, looking at more general  \emph{decomposed} groups does not look particularly promising: we believe that the methods here and in \cite{BermejoLinVdN13_Infinite_Normalizers} can be extended, e.g., to simulate normalizer circuits over groups of the form $\R^a\times \Z^b \times \T^c \times \DProd{N}{d} \times \mathbf{B}$, with additional $\R$ factors (cf.\ our discussion in \cite{BermejoLinVdN13_Infinite_Normalizers}). On the other hand, allowing more general types of groups to act as \emph{black-boxes} looks rather promising to us: one may, for instance, attempt to extend the notion of normalizer circuits to act on Hilbert spaces associated with multi-dimensional infrastructures \cite{Sarvepalli14_1D_infrastructures,FonteinWocjan11_Q_Alg_Period_Lattice_Infrastructure}, which may, informally, be understood as ``infinite black-box groups''\footnote{An $n$-dimensional infrastructure $\mathcal{I}$ provides a classical presentation for an $n$-dimensional hypertorus group $\R^n/\Lambda\cong \T^n$, where $\Lambda$ is an (unknown) period lattice $\Lambda$. The elements of this continuous group are represented with some classical structures known as \emph{$f$-representations}, which are  endowed with an operation that allows us to compute within the torus. Although one must deal carefully with non-trivial technical aspects of infinite groups in order to properly define and compute with $f$-representations (cf.\ \cite{Sarvepalli14_1D_infrastructures,FonteinWocjan11_Q_Alg_Period_Lattice_Infrastructure} and references therein), one may intuitively understand infrastructures as ``generalized black-box hypertoruses''. We stress, though, that it is not standard terminology to call ``black-box group'' to an infinite group.} We expect, in fact, that  known quantum algorithms for finding hidden periods and hidden lattices within real vector spaces  \cite{Hallgren07_Pells_equation,Jozsa03_Hallgrens_Algorithm,Schmidt05_Q_Algorithm_Computation_Unit_Group,Hallgren2005_Unit_Group_Class_Group} and/or or infrastructures \cite{Sarvepalli14_1D_infrastructures,FonteinWocjan11_Q_Alg_Period_Lattice_Infrastructure} (e.g.,\ Hallgren's algorithm for solving Pell's equation \cite{Hallgren07_Pells_equation,Jozsa03_Hallgrens_Algorithm}) could be at least partially interpreted as generalized normalizer circuits in this sense  . Addressing this question would  require a careful treatment of precision errors that appear in such algorithms due to the presence of transcendental numbers, which play no role in the present paper\footnote{No such treatment is needed in this work, since we study quantum algorithms for finding hidden structures in \emph{discrete} groups.}. Some  open questions in this quantum algorithm subfield have been discussed in \cite{FonteinWocjan11_Q_Alg_Period_Lattice_Infrastructure}.

A third possible direction to investigate would be whether different models of normalizer circuits could be constructed over \textbf{\emph{algebraic structures that are not groups}}. One could, for instance, consider sets with \emph{less algebraic structure}, like semi-groups. In this regard, we highlight that a quantum algorithm for finding discrete logarithms over finite semigroups was recently given in   \cite{Childs14_Discrete_Log_Semigroups}. Alternatively, one could study also \emph{sets} with more structure than groups, such as \emph{fields}, whose study is relevant to Van Dam-Seroussi's  quantum algorithm for estimating Gauss sums  \cite{VanDamSeroussi02_Gauss_Sums_QALG}. 

Lastly, we mention some open questions suggested by our work.

In this work, we have not investigated the computational complexity of black-box normalizer circuits \emph{without} classical post-processing. There are two facts which suggest  that power of black-box normalizer circuits alone might, in fact, be significantly smaller. The first is the fact that the complexity class of problems solvable  by Clifford circuits alone is $\oplus \mathbf{L}$ \cite{AaronsonGottesman04_Improved_Simul_stabilizer}, believed to be a strict subclass of $\mathbf{P}$. The second is that normalizer circuits  seem to be incapable of implementing most classical functions coherently  even with constant accuracy (this has been rigorously  shown in finite dimensions \cite{VDNest_12_QFTs,BermejoVega_12_GKTheorem}).

Finally, one may study whether considering more general types of inputs,  measurements or adaptive operations might change the power of black-box normalizer circuits. Allowing, for instance, input product states has potential to increase the power of these circuits, since this already occurs for standard Clifford circuits \cite{BravyiKitaev05MagicStateDistillation,JozsaVdNest14_Classical_Simulation_Extended_Clifford_Circuits}. Concerning measurements, the authors believe that allowing, e.g.\, adaptive Pauli operator measurements (in the sense of \cite{BermejoVega_12_GKTheorem}) is unlikely to give any additional computational power to black-box normalizer circuits: in the best scenario, this could only happen  in infinite dimensions, since adaptive normalizer circuits over finite Abelian groups are also efficiently classically simulable with stabilizer techniques \cite{BermejoVega_12_GKTheorem}. With more general types of measurements, it should be possible to recover full quantum universality, given that qubit cluster-states (which can be generated by Clifford circuits) are a universal resource for measurement-based quantum computation \cite{raussen_briegel_01_Cluster_State,raussen_briegel_onewayQC}. The possibility of obtaining  intermediate hardness results if non-adaptive yet also non-Pauli measurements are allowed (in the lines of \cite{Aaronson11_Computational_Complexity_Linear_Optics}  or \cite[theorem 7]{JozsaVdNest14_Classical_Simulation_Extended_Clifford_Circuits}) remains also open.

\section{Abelian groups}\label{sect:Groups}

The most general groups we will consider in this work are abelian groups of the form
\be\label{group_hilbert_space} G= \Z^a \times \T^b \times \DProd{N}{c} \times \mathbf{B}, \label{general_group} \ee
where $a$, $b$, $N_1,\cdots,N_c$ are arbitary integers and $\mathbf{B}$ is a finite abelian \emph{black box group}, to be defined more precisely later.

We will discuss each of the constitutent groups in turn.

\subsection{$\Z$: the group of integers}
$\Z$ simply refers to the group of integers under addition; it is infinite, but finitely generated (by the element $1$).

\subsection{$\T$: the torus group}
$\T$ refers to the group of real numbers in the interval $[0,1)$ under addition modulo $1$. Unlike all the other components we will consider, it is both infinite and not finitely generated. The introduction of $\T$ is necessary to allow the use of quantum Fourier transforms over $\Z$, as we'll see in the next section.

\subsection{Finite Abelian groups}

Let us start by stating a very important theorem we will use for our results:
\begin{thm}[\textbf{Fundamental Theorem of Finite Abelian Groups} \cite{Humphrey96_Course_GroupTheory}]\label{thm:Fundamental Theorem FAGroups}
Any finite Abelian group $G$ has a decomposition into a direct product of cyclic groups, i.e. 
\be\label{eq:Group Decomposition}
G \cong \Z_{d_1} \times \Z_{d_2} \times \cdots \times \Z_{d_k}
\ee
for some positive integers $d_1,\cdots,d_k$.
\end{thm}
Actually finding such a decomposition for a group $G$ may be difficult in practice. For example, consider the set of integers modulo $N$ that are also relatively prime to $N$; this set forms a group under multiplication. (This group is known as the \emph{multiplicative group of integers modulo $N$}, or $\Z_N^\times$.) It is not known classically how to decompose $\Z_N^\times$ into its cyclic subgroups. For example, if $N=pq$ for $p$, $q$ prime then $\Z_{pq}^\times \cong \Z_{p-1} \times \Z_{q-1}$, and hence decomposing $\Z_{pq}^\times$ is at least as hard as factoring $pq$ or, equivalently,  breaking RSA \cite{RSA}. More generally, decomposing $\Z_N^\times$ is known to be polynomial time equivalent to factoring \cite{Shoup08_A_Computational_Introducttion_to_Number_Theory_and_Algebra}. In the quantum case, however, Cheung and Mosca gave an algorithm \cite{mosca_phd,cheung_mosca_01_decomp_abelian_groups} to decompose any finite abelian group.

In equation (\ref{general_group}), the factors $\DProd{N}{c}$ represent an arbitrary finite Abelian group for which \emph{the group decomposition is known}. The case where the decomposition is unknown will be covered by the black box group $\mathbf{B}$.

\subsection{Black box groups}\label{sect:Black Box Groups}

In this work, we define a \emph{black-box group} $\mathbf{B}$ \cite{BabaiSzmeredi_Complexity_MatrixGroup_Problems_I}  to be a finite  group whose elements are uniquely encoded by binary strings of a certain size $n$, which is the length of the encoding. The elements of the black-box group can be multiplied and inverted  at unit cost by querying a black-bock, or \emph{group oracle}, which computes these operations for us.  The order of a  black-box group with encoding length $n$ is bounded above by $2^n$: the precise order   $|\mathbf{B}|$ may not be given to us, but it is assumed that the group oracle can identify which  strings in the group encoding correspond to elements of the group. When we say that a particular black-box group (or subgroup) is given (as the input to some algorithm), it is meant that a \emph{list of generators} of the group or subgroup is explicitly provided.

From now on, all  black-box groups in this work will be assumed to be \emph{Abelian}. Although we only consider finite Abelian black-box groups, we stress now, that  there is a subtle but crucial  difference between these groups and the explicitly decomposed finite Abelian groups in \cite{VDNest_12_QFTs,BermejoVega_12_GKTheorem}: although, mathematically, all Abelian black-box groups have a decomposition (\ref{eq:Group Decomposition}), it is  \emph{computationally hard} to find one and we assume no knowledge of it. In fact, our motivation to introduce black-box groups in our setting is precisely to model those Abelian groups that cannot be efficiently decomposed with  known classical algorithms that have, nevertheless, efficiently classically computable group operations. With some abuse of notation, we shall call all  such groups also ``black-box groups'', even if no oracle is needed to define them; in such cases, oracle calls will be replaced by \ppoly{n}-size classical circuits for computing group multiplications and inversions.

As an  example, let us consider again the group $\Z_N^\times$. This group can be naturally modeled as a black-box group in the above sense: on one hand, for any $x,y \in \Z_N^\times$, $xy$ and $x^{-1}$ can be efficiently  computed using Euclid's algorithm \cite{brent_zimmerman10CompArithmetic}; on the other hand, decomposing $\Z_N^\times$ is as hard as factoring \cite{Shoup08_A_Computational_Introducttion_to_Number_Theory_and_Algebra}. Note, in addition, that a generating set of $\Z_N^\times$ can found by taking a linear number of samples\footnote{Sampling  $\Z_N^\times$ can be done by sampling  $\{0,\cdots,N-1\}$ uniformly and then rejecting samples that are not relatively prime to $N$; this takes $O(\log\log N)$ trials to succeed with high probability. A similar approach works, in general, for sampling generating-sets of uniquely-encoded finite Abelian groups \cite{Fontein2014_Probability_Generating_Lattice}.}
 of $\Z_N^\times$.

\section{The Hilbert space of an Abelian group}\label{sect:Hilbert space}

In this section we introduce Hilbert spaces associated with Abelian groups of the form \be\label{group_hilbert_space} G=\Z^a\times \DProd{N}{b}\times \mathbf{B}, \ee where $\mathbb{Z}_N$ is the additive group of integers modulo $N$,  $\Z$ is the additive group of integers and $B$ is a black-box group. Apart from the short discussion on Hilbert space associated with black-box groups, we follow the definitions given in our previous works \cite{VDNest_12_QFTs,BermejoVega_12_GKTheorem,BermejoLinVdN13_Infinite_Normalizers}. 

\subsection{Finite Abelian groups}

First we consider $\mathbb{Z}_N$. With this group, we associate an $N$-dimensional Hilbert space ${\cal H}_{N}$ having a basis $\{|x\rangle: x\in\mathbb{Z}_N\}$, henceforth called the standard basis of ${\cal H}_N$. A state in ${\cal H}_{N}$ is (as usual) a unit vector $|\psi\rangle = \sum \psi_x|x\rangle$ with $\sum |\psi_x|^2=1$, where $\psi_x\in\mathbb{C}$ and where the sums are over all $x\in\mathbb{Z}_N$.

Second, we consider a black box group $\mathbf{B}$. With such a group we associate a $|\mathbf{B}|$-dimensional Hilbert space ${\cal H_B}$ with standard basis states $|b\rangle$ where $b$ ranges over all elements of $\mathbf{B}$.

\subsection{The integers $\Z$}\label{sect_Hilbert_space_Z}

An analogous construction is considered for the group $\mathbb{Z}$, although an important distinction with the former cases is that $\mathbb{Z}$ is infinite. We consider the infinite-dimensional Hilbert space ${\cal H}_{\Z}= \ell_2(\mathbb{Z})$ with standard basis states $|z\rangle$ where $|z\rangle\in \mathbb{Z}$. A state in ${\cal H}$ has the form \be\label{psi_Z} |\psi\rangle = \sum \psi_x|x\rangle\ee with $\sum |\psi_x|^2=1$, where the infinite sum is over all $x\in\mathbb{Z}$. More generally, any sequence $(\psi_x: x\in\mathbb{Z})$ with $\sum |\psi_x|^2<\infty$ can be associated with a quantum state in ${\cal H}_{\Z}$ via normalization. Sequences whose sums are not finite give rise to ``unnormalizable'' states. These states do not belong to ${\cal H}_{\Z}$ and are hence unphysical; however it is often convenient to consider such states nonetheless. Examples are the ``plane wave states'' \begin{equation}\label{eq:Fourier basis state of Z} |p\rangle:=\sum_{z\in \Z} 	\overline{\euler^{2\uppi i zp}}|z\rangle \quad p\in [0, 1).
\end{equation}
We denote  $\T:=[0, 1)$ as a group with addition modulo 1, called the (one-dimentional) torus group. Even though the  $|p\rangle$ themselves do not belong to ${\cal H}$, every state in the Hilbert space $\mathcal{H}_\Z$ can be written as a linear combination of them: \be\label{psi_T}|\psi\rangle = \int_{\mathbb{T}} \mbox{d}p \ \phi(p) |p\rangle\ee for some complex function $\phi:\T\to \mathbb{C}$,  where d$p$ denotes the Haar measure on $\T$. Thus the states $|p\rangle$ form an alternate basis\footnote{As discussed in \cite{BermejoLinVdN13_Infinite_Normalizers},  the states $|p\rangle$ are unnormalizable and, strictly speaking, do not form a ``basis'' in the usual sense (they are not  elements of  ${\cal H}_{\Z}$) but they can nevertheless used as a basis for our practical purposes.  Rigorously speaking, the $|p\rangle$ ``states'' can be  understood as Dirac-delta measures or Schwartz-Bruhat tempered distributions \cite{Bruhat61_Schwatz-Bruhat-functions,Osborne75_Schwartz_Bruhat}.  Although we will not do it here, our results can be stated more formally through the theory of rigged Hilbert spaces \cite{delaMadrid05_roleofthe_riggedHilbert, Antoine98_QM_beyond_Hilber_space, Gadella02_unified_Dirac_formalism, gadella12_Riggings_LCA_Groups}, which is often applied to study observables with continuous spectra.} of ${\cal H}$ which is infinite, parametrized by a continuous set. We call this basis the \emph{\textbf{Fourier basis.}} The Fourier basis is orthonormal in the sense that \be \langle p|p'\rangle =  \delta(p-p'),\ee
where $\delta(\cdot)$ is the Dirac delta.

In the following, when dealing with the Hilbert space ${\cal H}_{\Z}$, we will use both the standard and Fourier basis. More precisely, in our computational model (cf. section \ref{sect:Normalizer circuits over blackbox groups}), there are \emph{\textbf{two ``standard'' bases}} of the Hilbert space ${\cal H}_{\Z}$, parametrized by the groups $\Z$ and $\T$ (we reserve the term standard basis for the first, the $\Z$ basis).  This is different from the finite space ${\cal H}_N$ where we only use the standard basis labeled by $\Z_N$. 

\subsection{Total Hilbert space}

In general we will consider consider groups of the form (\ref{group_hilbert_space}). The associated Hilbert space has the tensor product form
\begin{equation}\label{total_H} {\cal H}={\cal H}_{\Z}^{\otimes a}\otimes {\cal H}_{{N_1}}\otimes \dots\otimes {\cal H}_{{N_c}}\otimes {\cal H}_\mathbf{B}.
\end{equation}
In this work, we treat   $\mathcal{H}$ as the underlying Hilbert space of a  quantum computation with  $m:=a+c+1$ computational registers. The first $a$ registers ${\cal H}_{\Z}$ are infinite dimensional.  The latter $c+1$ registers are finite dimensional, and the last one, $\mathcal{H}_\mathbf{B}$, is associated with some Abelian black-box group $\mathbf{B}$.

\subsubsection*{The group-element bases of $\mathcal{H}$}

As afore-mentioned, a normalizer computation  works within multiple ``standard bases'' of $\mathcal{H}$, which can change during the computation. Input states and final measurements in any of these bases are allowed. This is a central feature of the computational model we describe in the next section.  

The allowed bases of a normalizer computation are what-we-call the \emph{group-element bases} of $\mathcal{H}$, which we now describe. A group element basis $\mathcal{B}_{G}$ of $\mathcal{H}$ is a basis of group element states $\{\ket{g},g\in G\}$ parametrized by an  Abelian group $G$ of the form
\begin{equation}\label{group_labels_basis} G= G_1\times\cdots \times G_{a+b}\times \mathbb{Z}_{N_1}\otimes \dots\mathbb{Z}_{N_c}\times \mathbf{B} \quad\textnormal{where } G_i \in\{\mathbb{Z}, \mathbb{T}\},\end{equation}
 \begin{equation}\label{basis} {\cal B}_{G}:=\left\lbrace |g\rangle:=|g(1)\rangle\otimes \dots|g(m)\rangle, \quad g=(g(1), \dots, g(m))\in G\right\rbrace.
\end{equation}
The notation $G_i = \T$ indicated that the state $\ket{g(i)}$ (locally) is a Fourier state of $\Z$ (\ref{eq:Fourier basis state of Z}). Note that the states $\ket{g}$ are product-states with respect to the tensor-product structure (\ref{total_H}) of $\mathcal{H}$. By construction, there are $2^a$ possible choices of groups  for the same Hilbert space, so that the number of group-element bases is $2^a$. As discussed in \cite{BermejoLinVdN13_Infinite_Normalizers} all groups (\ref{group_labels_basis}) are related by the Pontryagin duality\footnote{Groups (\ref{group_labels_basis}) form a family (and a category) of groups generated by replacing the factors $G_i$ of the original group $\Z^{a}\times \DProd{N}{c}\times \mathbf{B}$ with their character groups, and then identifying isomorphic groups \cite{BermejoLinVdN13_Infinite_Normalizers}. The so-called Pontryagin duality \cite{Morris77_Pontryagin_Duality_and_LCA_groups,Stroppel06_Locally_Compact_Groups,Dikranjan11_IntroTopologicalGroups,rudin62_Fourier_Analysis_on_groups,HofmannMorris06The_Structure_of_Compact_Groups} determines the number $2^a$ of groups in the family. The number of group-element bases $2^a$ is larger than $1$ iff \emph{infinite groups} are involved. Due to, there are not multiple ``standard-bases'' in finite-dimensional normalizer circuits   (nor in Clifford circuits) \cite{VDNest_12_QFTs,BermejoVega_12_GKTheorem}.}.

We now consider some examples. First, note that, by construction, the \emph{standard basis} of $\mathcal{H}$ is the group-element basis $\mathcal{B}_G$ with $G$ of the form $\Z^{a}\times \DProd{N}{c}\times \mathbf{B}$:
\be |x(1)\rangle\otimes \cdots\otimes |x(a)\rangle\otimes |y(1)\rangle\otimes\dots\otimes |y(c)\rangle\otimes |\mathbf{b}\rangle,\quad x(i)\in \Z,\: y(j)\in\Z_{N_j},\:\mathbf{b}\in\mathbf{B},\nonumber \ee
and, clearly, $(x,y,\mathbf{b})$ is an element of $\Z^{a}\times \DProd{N}{b}\times \mathbf{B}$. Alternatively, by choosing the basis of some (say the $a$th) $H_\Z$ register to be  the Fourier basis, we get
\be
\left(|x(1)\rangle\otimes \cdots\otimes |x(a-1)\rangle\otimes |p\rangle\right)\otimes |y(1)\rangle\otimes \dots\otimes|y(c)\rangle\otimes |\mathbf{b}\rangle, \nonumber\ee
where $\ket{p}$ is now a Fourier basis state (\ref{eq:Fourier basis state of Z}). Now $p\in \T$ and the  basis is parametrized by the elements of  $(\Z^{a-1}\times \T)\times \DProd{N}{c}\times \mathbf{B}$.

\section{Black box normalizer circuits}\label{sect:Normalizer circuits over blackbox groups}

In this section we define black-box normalizer circuits acting on Hilbert spaces of the form (\ref{total_H}). We will split the discussion in two parts: in section \ref{sect_def_BB_normalizer_finite} we discuss black box normalizer circuits for finite-dimensional spaces, i.e. spaces where  ${\cal H}_{\Z}$ does not occur in the decomposition. Restricting to the finite-dimensional case will allow us to introduce black box normalizer circuits without many technical complications.

 In a second step, in section \ref{sect_def_BB_normalizer_infinite} we allow for general spaces  of the form (\ref{total_H}). The definition of black box normalizer circuits will be technically more involved owing to the fact that in ${\cal H}_{\Z}$ both the standard basis and the Fourier basis need to be considered; this technical element is however not essential to understand the basic idea behind black box normalizer groups and the reader may skip section \ref{sect_def_BB_normalizer_infinite} in a first reading. However, the general definition of black box normalizer circuits is necessary to make a rigorous connection with e.g. Shor's factoring algorithm.

\subsection{Finite groups}\label{sect_def_BB_normalizer_finite}

Let $G$ be a finite Abelian group of the form  $G=\DProd{N}{c} \times \mathbf{B}$. 
The associated Hilbert space
\be
{\cal H}={\cal H}_{{N_1}}\otimes \dots\otimes {\cal H}_{{N_c}}\otimes {\cal H}_\mathbf{B}
\ee 
 has standard basis vectors $|g\rangle$ where $g$ ranges over all elements of $G$ (cf.\ section \ref{sect:Hilbert space}). A normalizer gate over $G$ is either an automorphism gate, quadratic phase gate or quantum Fourier transform, as defined next:

\

\noindent \textbf{Automorphism gates.} Recall that a group automorphism is an invertible map $\alpha:G\rightarrow G$ satisfying $\alpha(g+h)  = \alpha(g) + \alpha(h)$ for every $g, h\in G$. An automorphism gate over $G$ is an operation $U_\alpha:\ket{h}\rightarrow\ket{\alpha(h)}$ where $\alpha$ is an automorphism; we consider automorphism gates to be available as black-box quantum gates (oracles). Note that each $U_{\alpha}$ acts as a permutation on the standard basis and is hence a unitary operation.

\

\noindent \textbf{Quadratic phase gates.}  A function $\chi:G\rightarrow U(1)$  (from the group $G$ into the complex numbers of unit modulus) is called a character  if $\chi(g+h) =\chi(g)\chi(h)$ for every $g, h\in G$. A function $B:G\times G\to U(1)$ is said to be a bicharacter if it is a character in both arguments. A function $\xi:G\rightarrow U(1)$   is called \emph{quadratic} if 
\begin{equation}
\xi(g+h)=\xi(g)\xi(h)B(g,h),\quad \text{for every $g$, $h\in G$}
\end{equation}
for some bicharacter $B(g,h)$. A quadratic phase gate is any diagonal unitary operation acting on the standard basis as $D_{\xi}: \ket{h}\rightarrow \xi(h)\ket{h}$, where $\xi$ is a quadratic function of $G$. Similar to automorphism gates, we consider quadratic phase gates to be available as black-box quantum gates.

\

\noindent\textbf{Quantum Fourier transform.} In contrast to both automorphism gates and quadratic phase gates, which act on the entire system ${\cal H}$, we will only consider settings where quantum Fourier transforms never act on the black box portion ${\cal H}_\mathbf{B}$ of the total system. This a natural restriction; in particular, (to our knowledge) in all existing quantum  algorithms that do use QFTs, these QFTs act on systems of the form ${\cal H}_{{N_1}}\otimes \dots\otimes {\cal H}_{{N_c}}$. This is precisely the case we consider here. To define the QFT, consider the Hilbert space ${\cal H}_N$ with standard basis vectors $|x\rangle$ with $x\in \Z_N$. The QFT ${\cal F}_{N}$ over $\Z_N$ is a unitary operation which acts on $|\psi\rangle = \sum \psi_x|x\rangle$ in ${\cal H}_N$ as
\begin{equation}\label{eq:QFT over Z_N}
 {\cal F}_N|\psi\rangle= \sum_{y\in\mathbb{Z}_N} \hat\psi(y) |y\rangle \quad\mbox{ with } \hat\psi(y):= \frac{1}{\sqrt{N}} \sum_{x\in\mathbb{Z}_N} e^{2\pi i xy}\psi(x).
\end{equation}
In this work we will consider the quantum Fourier transform ${\cal F}_{N_i}$ to act on ${\cal H}_{N_i}$, for any system $i=1, \dots, a$.

\

\noindent{\textbf{A normalizer circuit over $G$}} \cite{VDNest_12_QFTs,BermejoVega_12_GKTheorem} is any unitary circuit composed of normalizer gates. As input state, we will consider any standard basis state $|g\rangle$ with $g\in G$. After all gates in the circuit are applied, a measurement in the standard basis is performed. We do not consider intermediate measurements. We also recall, as mentioned above, that all automorphism and quadratic phase gates are to be given as black-box operations.

\subsection{Infinite groups}\label{sect_def_BB_normalizer_infinite}

Here we extend the definition of black box normalizer circuits to general Hilbert spaces of the form (\ref{total_H}). This will be technically somewhat more involved than the finite-dimensional case. The main technical complication is not related the black-box portion ${\cal H}_\textbf{B}$ of the Hilbert space, but rather to the infinite-dimensional space ${\cal H}_{\Z}$ and, in particular, to the fact that we work with \emph{two different bases} in this space, labeled by \emph{two completely different groups} $\Z$ and $\T$. We have already introduce these bases in section \ref{sect:Hilbert space}.  The role of these bases in the normalizer circuit model has been discussed in detail in our previous work \cite{BermejoLinVdN13_Infinite_Normalizers}, but we give here a self-contained summarized account. 

\

\noindent\textbf{Designated bases $\mathcal{B}_G$.} Consider a Hilbert space of the form (\ref{total_H}). Fix any member $G$ of the family of $2^a$ groups defined  in (\ref{group_labels_basis}). Note that $G$ thus generally contains factors $\Z_{N_i}$, $\Z$,  $\T$ and $\mathbf{B}$. We consider the corresponding group-element basis ${\cal B}_G= \{|g\rangle: g\in G\}$ as defined in (\ref{basis}). We set ${\cal B}_G$ to be (what we call) the \emph{designated basis} of the computation. The group $G$ that labels  $\mathcal{B}_G$ choice of basis will determine what normalizer gates (read below). $\mathcal{B}_G$ will be the basis in which measurements are performed.

\

\noindent \textbf{Automorphism gates.} The definition of automorphism gates is similar to above, but now the action of these gates is defined relative to the \emph{designated basis}. That is, we consider a group automorphism $\alpha:G\rightarrow G$ and define the corresponding automorphism gate $U_{\alpha}$ by its action on the  basis ${\cal B}_G$, as follows: $U_\alpha:\ket{h}\rightarrow\ket{\alpha(h)}$. As above, automorphism gates are given as oracles. Furthermore, in the infinite case we impose that the group automorphism must be a \emph{continuous} map.

\

\noindent \textbf{Quadratic phase gates.}  The action of quadratic gates is also defined relative to the designated basis. A function $\chi:G\rightarrow U(1)$  is called a character if its continuous and if $\chi(g+h) =\chi(g)\chi(h)$ for every $g, h\in G$. A function $B:G\times G\to U(1)$ is said to be a bicharacter if it is a character in both arguments. A function $\xi:G\rightarrow U(1)$   is called \emph{quadratic} if it is continuous and if
\begin{equation}
\xi(g+h)=\xi(g)\xi(h)B(g,h),\quad \text{for every $g$, $h\in G$}
\end{equation}
for some bicharacter $B(g,h)$. A quadratic phase gate is any diagonal unitary operation acting on ${\cal B}_G$ as $D_{\xi}: \ket{h}\rightarrow \xi(h)\ket{h}$, where $\xi$ is a quadratic function of $G$. As above, quadratic gates are given as oracles. 

\

\noindent\textbf{Quantum Fourier transforms.} In contrast to both automorphism gates and quadratic phase gates, which leave the designated basis unchanged, the role of the quantum Fourier transform (QFT)  is \emph{precisely} to change the designated basis ${\cal B}_G$ (at a given time) into another group-element basis ${\cal B}_{G'}$. This transformation  follows certain rules \cite{BermejoLinVdN13_Infinite_Normalizers}, described next.

Roughly speaking, the QFT of $\mathcal{H}_\Z$ is a basis change between the standard and the Fourier basis (\ref{eq:Fourier basis state of Z}). Note that the QFT over $\Z_N$, which we introduced as a gate  (\ref{eq:QFT over Z_N}), could be defined also in this way (as a change of basis). However, the Fourier transform has now more exotic features than their finite-dimensional counterparts. First, strictly speaking, the QFT over $\mathcal{H}_\Z$ is not a \emph{quantum gate}: despite being a change of basis, it does not define a  unitary rotation. Second, there are actually \emph{two} inequivalent Fourier transforms of $\mathcal{H}_\Z$. These technicalities deserve further discussion.

\textbf{QFTs over $\mathcal{H}_\Z$ are not gates}. In the case of ${\cal H}_N$, both the standard and the Fourier basis have the same cardinality (since  $\Z_N$ is isomorphic to its character group) and such a change of basis can be actively performed by means of a unitary rotation, which defines the QFT over $\Z_N$. In the case of ${\cal H}_{\Z}$, the standard basis $\{|x\rangle: x\in \Z\}$ and Fourier basis $\{|p\rangle: p\in \T\}$ have different cardinality (recall section \ref{sect_Hilbert_space_Z}) and cannot be ``rotated'' into each other. Therefore, the QFT, while corresponding to a change of basis, will as such not be a unitary gate in the usual sense\footnote{Mathematically, this Fourier transform is a unitary transformation between two different functional spaces, $L^2(\Z)$ and $L^2(\T)$. The latter two define one quantum mechanical system with two possible bases (of Dirac-delta measures) labeled by $\Z$ and $\T$. In the finite dimensional case, the picture is simpler because the QFT is a unitary transformation of $L^2(\Z_N)$ onto itself. (These facts are consequences of the Plancherel theorem for locally compact Abelian groups \cite{rudin62_Fourier_Analysis_on_groups,HofmannMorris06The_Structure_of_Compact_Groups}.)}. Because of this asymmetry, there are \textbf{two QFTs}, defined as follows.
\begin{itemize}
\item \textbf{QFT over $\Z$}. If the standard basis ($\Z$ basis) is the designated basis of ${\cal H}_{\Z}$, states are represented as \be|\psi\rangle= \sum_{x\in\mathbb{Z}} \psi(x) |x\rangle.\ee 
Gates are defined according to this (integer) basis, which is also our measurement basis. When we say that the \emph{QFT over $\Z$} is applied to $|\psi\rangle$, we mean that the designated basis is changed from the standard basis to the Fourier basis. The state does not actually change (no gate is physically applied\footnote{We choose this notation to be consistent with our previously existing terminology \cite{VDNest_12_QFTs,BermejoVega_12_GKTheorem}.}), but the normalizer gates acting after the QFT  will be associated with $\T$ (\emph{not} $\Z$), and measurements will be performed in the $\T$ basis. We therefore ought to write the wavefunction of the state $\ket{\psi}$ in the Fourier basis:
    \begin{equation}\label{eq:Fourier transform over Z}
    |\psi\rangle= \int_{\mathbb{T}}\mbox{d}p\  \hat\psi(p) |p\rangle \quad\mbox{ with } \hat\psi(p):=  \sum_{x\in\mathbb{Z}} e^{2\pi i px} \psi(x).
    \end{equation}
\item \textbf{QFT over $\boldsymbol{\mathbb{T}}$.} In the opposite case, the designated basis of ${\cal H}_{\Z}$ is the Fourier basis, in which a state looks like
\be|\psi\rangle= \int_{\mathbb{T}}\mbox{d}p\  \psi(p) |p\rangle.
\ee 
When we say that the \emph{QFT over $\T$} is applied to $|\psi\rangle$, we mean that the designated basis is changed from the Fourier basis to the standard basis. Like in the previous case, we must re-express the state $|\psi\rangle$ in the new designated basis:
    \begin{equation}\label{eq:Fourier transform over T}
    |\psi\rangle= \sum_{x\in\mathbb{Z}} \hat\psi(x) |x\rangle \quad\mbox{ with } \hat\psi(x):=  \int_{\mathbb{T}}\mbox{d}p\ e^{2\pi i px} \psi(p).
    \end{equation}
\end{itemize}
Note that, by definition, the QFT over $\Z$ may only be applied if the designated basis is the standard basis and, conversely, the QFT over $\T$ may only be applied of the designated basis is the Fourier basis. \\

\noindent\textbf{Full and partial QFTs over $G$.} 
We last consider the total Hilbert space ${\cal H}$ with designated basis ${\cal B}_G$. We allow for the application of a QFT on any of the individual spaces ${\cal H}_{N_i}$ or ${\cal H}_{\Z}$ in the tensor product decomposition of ${\cal H}$ (we call this a \emph{partial QFT}). The designated basis is changed according to the rules described above on all subsystems ${\cal H}_{\Z}$ where a QFT is applied. The full \emph{QFT over $G$}  is the combination of all partial QFTs acting on the smaller  registers.

\subsection{The black-box normalizer circuit model} \label{sect:circuit model}

We are now ready to introduce normalizer circuits in precise terms. Roughly speaking, a \emph{\textbf{black-box normalizer circuit}} of size $T$ is a quantum circuit ${\cal C}=U_T\cdots U_1$ composed of $T$ \emph{normalizer gates} $U_i$, which we have introduced in the previous sections. More precisely, a   \emph{normalizer circuit over $G=\Z^{a+b}\times\DProd{N}{c}\times\mathbf{B}$} is a quantum circuit that acts on a Hilbert space ${\cal H}$  associated with the group $G$.
\begin{equation*}
\mathcal{H}=\mathcal{H}_\Z^{a}\otimes\mathcal{H}_\Z^{b}\otimes \left(\mathcal{H}_{N_1}\otimes \cdots \otimes\mathcal{H}_{N_c}\right)\otimes \mathcal{H}_\mathbf{B}.
\end{equation*}
In this decomposition, the parameters $a$, $b$, $c$, $N_i$ and the Abelian black-box group $\mathbf{B}$ can be chosen arbitrarily.  

To define a complete circuit model, we specify next the allowed inputs, gates and measurements of the computation.
\begin{itemize}
\item \textbf{\emph{Designated basis.}} In a normalizer computation there is no fixed ``standard basis'', in the usual sense of the word that comes from  the standard model of quantum circuits  \cite{nielsen_chuang}. Instead, there is a \emph{\textbf{designated basis}} ${\cal B}_{G_t}$ at every time step $t$ of the circuit, that is subject to change along the computation. ${\cal B}_{G_t}$ is restricted to be group-element basis, as in equation (\ref{basis}).

\item \emph{\textbf{Input states}}. The input states of a normalizer computation are  elements of some designated group basis $\mathcal{B}_{G_0}$ at time zero\footnote{The results in this paper would still hold if the allowed inputs are periodic states, coset-states and, in general, stabilizer states  \cite{VDNest_12_QFTs,BermejoVega_12_GKTheorem,BermejoLinVdN13_Infinite_Normalizers} \emph{if} the stabilizer group of the input state is given as an input.}. Without loss of generality, we assume that the registers  $\mathcal{H}_\Z^a$ and $\mathcal{H}_\Z^b$ are fed, respectively, with standard-basis  $\ket{n}$, $n\in \Z$ and Fourier-basis states  $\ket{p}$, $p\in \T$. In our notation, this is equivalent to choosing the basis $\mathcal{B}_{G_0}$ with $G_0=\Z^a \times \T^b \times \DProd{N}{c}\times \mathbf{B}$.

\item \textbf{\emph{Structure of the circuit}}:
\begin{itemize}
\item[$\circ$] At time $t=1$, the gate $U_1$ is applied, which is either an automorphism gate, quadratic phase gate over $G_0$ or a QFT. Recall that automorphism gates and quadratic phase gates are given as black boxes. The designated basis is changed from ${\cal B}_{G_0}$ to ${\cal B}_{G_1}$, for some group $G_1$ in the family (\ref{group_labels_basis}), which is only different from $G_0$ if a QFT is applied (recall the update rules from the previous section).

\item[$\circ$] At time $t=2$, the gate $U_1$ is applied, which is, again, either an automorphism gate, quadratic phase gate or a QFT over $G_1$. The designated basis is changed from ${\cal B}_{G_1}$ to ${\cal B}_{G_2}$, for some group $G_2$.

\item[$\circ$] The gates $U_3, \dots, U_t$ are considered similarly. We denote by ${\cal B}_{G_t}$ the designated basis after application of $U_t$ (for some group $G_t$ in the family (\ref{group_labels_basis})), for all $t=3, \dots, T$. Thus, after all gates have been applied, the designated basis is $G_T$.

\item[$\circ$] After the circuit, a measurement in the designated basis $G_T$ is performed.
\end{itemize}
\end{itemize}

\subsubsection*{Precision requirements}

In the model of quantum circuits above, input states and final measurements in the Fourier-basis $\{\ket{p}, p\in\T\}$ of $\mathcal{H}_\Z$  can never be implemented with perfect accuracy, a limitation that stems from the fact that the $\ket{p}$ states are \emph{unphysical}. This can be quickly seen in two ways: first, in the $\Z$ basis, these states are infinitely-spread plane-waves  $\ket{p}=\sum \overline{\euler^{2\pii zp}} \ket{z}$; second, in the $\T$
basis, they are infinitely-localized Dirac-delta pulses. Physically, preparing Fourier-basis states or measuring in this basis \emph{perfectly} would require infinite energy and lead to infinite precision issues in our computational model.

In the algorithms we study in this work (namely, the order-finding algorithm in theorem \ref{thm:Order Finding}), Fourier states over $\Z$ can be substituted with \emph{\textbf{realistic physical approximations}}. The degree of \emph{precision} used in the process of  Fourier state preparation is treated as a \emph{computational resource}. We model the precision used in a computation as follows.

Since  our goal is  to use the Fourier basis $\ket{p}$, $p\in \T$, to represent information in a computation, we require the ability to store and retrieve information in this  continuous-variable basis. Our assumption is that for any set $X$ with cardinality $d=|X|$, we can divide the continuous circle-group $\T$ spectrum into $d$ equally sized sectors of length $1/d$ and use them to represent the elements of $X$. More precisely, to each element of $X$ we assign a number in $\Z_d$. The element $x_i\in X$ with index $i\in\Z_d$ is then represented by any state of the subspace $V_{i,d}=\mathrm{span}\{\ket{\tfrac{i}{d}+\Delta} \textnormal{ with }|\Delta|  <\tfrac{1}{2d}\}$. We call the latter states  \emph{$d$-approximate Fourier states} and refer to $d$ as the \emph{precision level}  of the computation. We assume that these states can be prepared and distinguished to any desired precision $d$ in the following way:
\begin{itemize}
\item[1.] \textbf{State preparation assumption.} Inputs $\ket{\psi_i}$ with at least $\tfrac{2}{3}$ fidelity to some element of  $V_{i,d}$ can be prepared for any $i\in \Z_d$.
\item[2.] \textbf{Distinguishability assumption.}  The subspaces $V_{i,d}$ can be reliably distinguished.
\end{itemize}
Note that $d$ determines how much information is stored in the Fourier basis.

\begin{definition}[\textbf{Efficient use of precision}\footnote{Note that this definition is not necessary to define normalizer circuits but to discuss the physicality of the model. We point out that there might be better ways to model precision than ours (which may, e.g., lead to tighter bounds or more efficient algorithms), but our simple model is enough to derive our main results. We advance that, even if these precision requirements turned out to be high in practice, there exist efficient discretized \emph{qubit} implementations  of all the infinite-dimensional quantum algorithms that we study later in the paper (cf.\ theorem  \ref{thm:Shor normalizer}).}] A \emph{quantum algorithm} that uses $d$-approximate Fourier states to solve  a computational problem with input size $n$ is said to use an \emph{efficient} amount of precision if and  only  $	\log{d}$ is upper bounded by some polynomial of $n$. Analogously,  an algorithm that stores information  in the standard basis $\{\ket{m},\,m\in\Z\}$ is said to be \emph{efficient} if the states with $m$ larger than some threshold $\log{( m_\textrm{max})}\in O(\poly{n})$ do not play a role in the computation.
\end{definition}

\subsubsection*{Classical encodings for normalizer gates}

We finish this section discussing how the  normalizer gates of a given normalizer circuit are presented in a classical encoding. Since quantum Fourier transforms can be specified by mere bit strings storing the circuit locations where they act, we focus on automorphism and quadratic phase gates. We again let $G=\Z^a \times \T^b\times \DProd{N}{c}\times \mathbf{B}$ be the group that defines the designated basis $\mathcal{B}_G$ in a normalizer circuit and define $m=a+b+c$.

From now on, we restrict ourselves to studying group automorphisms and quadratic functions which are \emph{efficiently computable rational functions}. This limits the class of \emph{classical functions} that we consider.
\begin{enumerate}
\item \textbf{Rational.\footnote{We expect this assumption not to be essential, but it simplifies our proofs by allowing us to use exact arithmetic operations. Our stabilizer formalism in \cite{BermejoLinVdN13_Infinite_Normalizers}  can still be applied if the functions $\alpha$, $\xi$ are not rational, and we expect some version of the  simulation result (theorem \ref{thm:Simulation}) to hold even when trascendental numbers are involved (taking carefully into account precision errors). It is an good question to explore whether an exact simulation result may hold for algebraic numbers \cite{Cohen:1995Course_Computational_Algebraic_Number_Theory}.}} An  automorphism (or an arbitrary function) $\alpha:G\rightarrow G$ is rational if it returns rational outputs for all rational inputs. A quadratic function $\xi$ is rational if it can be written in the form $\xi(g)=\exp\left(2\pii\, q(g) \right)$ where $q$ is a rational function from $G$ into $\R$ modulo $2\Z$.
\item\textbf{Efficiently computable.}  $\alpha$ and  $q$ can be computed by  polynomial-time uniform family of classical circuits $\{\alpha_i\}$, $\{q_i\}$. All   $\alpha_i$, $q_i$ are $\ppoly{m,i}$ size classical circuits that  query the \emph{black-box group oracle} at most $\ppoly{m,i}$ times: their inputs are strings of rational numbers whose numerators and denominators are represented by $i$ classical bits (their size is $O(2^i)$). For any rational element $g\in G$ that can be represented with so many bits (if $G$ contains factors of the form $\T$ these are approximated by fractions), it holds that $\alpha_i(g)=\alpha(g)$
 and $q_i(g)=q(g)$. 
 
In certain cases (see section \ref{sect:Quantum Algorithms}) we will consider groups like $\Z_N^\times$  which, strictly speaking, are not black-box groups (because polynomial time algorithms for group multiplication for them are available  and there is no need to introduce oracles). In those cases, the queries to the black-box group oracle (in the above model) are substitute by some efficient subroutine.
\end{enumerate}
We add a third restriction to the above.
\begin{itemize}
\item[3.] \textbf{Precision bound.} For any  $q$ or $\alpha$ that acts on an \emph{infinite} group  a bound $n_\textrm{out}$ is given so that for every $i$, the number of bits needed to specify the numerators and denominators in the output of $q_i$ or $\alpha_i$ exactly is at most $i+n_{\mathrm{out}}$. The bound $n_\textrm{out}$ is independent of $i$ and indicates how much the input of each function may grow or shrink along the computation of the output\footnote{For infinite groups there is no fundamental limit to how much the output of $\alpha$ or $q$ may grow/shrink with respect to the input (this follows from the normal forms in  \cite{BermejoLinVdN13_Infinite_Normalizers}). The number $n_\textrm{out}$ parametrizes the precision needed to compute the function. Similarly to 2., this assumption might me weakened if a treatment for precision errors is incorporated in the model.}. This bound is used to correctly store the output of  maps $\alpha:\Z^a\rightarrow \Z^a$, $\alpha':\Z^a\rightarrow \T^a$ and to detect whether the output  of a function $\alpha'':\T^b \rightarrow \T^b$ might get truncated modulo 1.
\end{itemize}
The allowed automorphism gates $U_\alpha$ and quadratic phase gates $D_\xi$ are those associated with efficiently computable rational functions $\alpha$, $\xi$. We  ask these  unitaries to be efficiently implementable as well\footnote{Recall that, in finite dimensions, the gate cost of implementing a classical function $\alpha$ as a quantum gate is at most the classical cost \cite{nielsen_chuang} and that computing $q$ efficiently is enough to implement $\xi$ using  phase kick-back tricks \cite{Kaye05_optimized_Quantum_Elliptic_Curve}. We expect these results to extend to  infinite dimensional systems of the form $\mathcal{H}_\Z$.}, by \ppoly{m,i,n_\textrm{out}}-size quantum circuits comprising at most \ppoly{m,i,n_\textrm{out}} quantum queries of the group oracle. The variable $i$ denotes the bit size used to store the labels $g$ of the inputs $\ket{g}$ and bounds the precision level $d$ of the normalizer computation, which we set to fulfill $\log d  \in O(i+n_\textrm{out})$. The complexity of a normalizer gate is measured by the number of gates and (quantum) oracle queries needed to implement them.

In the next section \ref{sect:Quantum Algorithms}, we will see particular examples of efficiently computable normalizer gates. We will repeatedly make use of automorphism gates of the form \begin{equation}\nonumber
U_\alpha \ket{k_1,\ldots,k_m,x}\longrightarrow \ket{k_1,\ldots,k_m,b_1^{k_1}\cdots  b_m^{k_m}x}
\end{equation} where $k_i$ are integers and $b_j$, $x$ are elements of some black-box group $\mathbf{B}$. These gates are allowed in our model, since there exist well-known efficient classical circuits for modular exponentiation given access to a group multiplication oracle \cite{brent_zimmerman10CompArithmetic}. In this case, a precision bound can be easily computed: since the infinite elements $k_i$ do not change in size and all the elements of $\mathbf{B}$ are specified with strings of the same size, the  output of $\alpha$  can be represented with as many bits as the input and we can simply take $n_\textrm{out}= 0$ (no extra bits are needed).

Many examples of efficiently computable normalizer gates were given in \cite{VDNest_12_QFTs,BermejoVega_12_GKTheorem}, for decomposed finite group $\DProd{N}{c}$. It was also shown in  \cite{VDNest_12_QFTs} that all normalizer gates over such groups can be efficiently implemented.

\section{Quantum algorithms}\label{sect:Quantum Algorithms}

\subsection{The discrete logarithm problem over $\Z_p^\times$}\label{sect:Discrete Log}

In this section we consider the discrete-logarithm problem studied by Shor \cite{Shor}. For any prime number $p$, let $\Z_p^\times$  be the multiplicative group of non-zero integers modulo $p$. An instance of the discrete-log  problem over $\Z_p^\times$ is determined by two  elements $a$, $b\in \Z_p^\times$, such that $a$ generates the group $\Z_p^\times$. Our task is to find the smallest non-negative integer $s$ that is a solution to the equation $a^s=b \bmod{p}$; the number is called the discrete logarithm $s=\log_a b$.

We now review Shor's algorithm \cite{Shor,childs_vandam_10_qu_algorithms_algebraic_problems} for this problem and prove our first result.
\begin{theorem}
[\textbf{Discrete logarithm}]\label{thm:Discrete Log} Shor's quantum algorithm for the discrete logarithm problem over $\Z_p^\times$ is a black-box normalizer circuit over the group $\Z_{p-1}^2\times \Z_p^\times$.
\end{theorem}
Theorem \ref{thm:Discrete Log} shows that black box normalizer circuits over \emph{finite} Abelian groups can efficiently solve a problem for which no efficient classical algorithm is known. In addition, it tells us that black-box normalizer circuits can render widespread public-key cryptosystems vulnerable: namely, they break the Diffie-Helman key-exchange protocol \cite{DiffieHellman}, whose security  relies in the assumed classical intractability of the discrete-log problem.
\begin{proof}
Let us first recall the main steps in Shor's discrete log algorithm.

\begin{algorithm}[Shor's algorithm for the discrete logarithm]\label{alg:Discrete Log}

\begin{alg_in} Positive integers $a$, $b$, where $\Z_p^\times = \langle a \rangle $.

\end{alg_in}
\begin{alg_out} The least nonnegative integer $s$ such that $a^s\equiv b \pmod{p}$.
\end{alg_out}

We will use three registers indexed by integers, the first two modulo $p-1$ and the last modulo $p$. The first two registers will correspond to the additive group $\Z_{p-1}$, while the third register will correspond to the multiplicative group $\Z_p^\times$. Two important ingredients of the algorithm will be the unitary gates $U_a:\ket{s}\rightarrow\ket{sa}$ and $U_b:\ket{s}\rightarrow\ket{sb}$.

\begin{enumerate}

\item \textbf{Initialization:} Start in the state $\ket{0}\ket{0}\ket{1}$.

\item Create the superposition state $\frac{1}{p-1} \sum_{x,y=0}^{p-1} \ket{x} \ket{y} \ket{1} $, by applying the standard quantum Fourier transform on the first two registers.

\item Apply the unitary $U$ defined by $U\ket{x}\ket{y} \ket{z} = \ket{x} \ket{y} \ket{za^xb^y}$, to obtain the state
\[ 
\frac{1}{p-1}\sum_{x,y=0}^{p-1} \ket{x} \ket{y} \ket{a^xb^y}
\]
This is equivalent to applying the controlled-$U_a^x$ gate between the first and third registers, and the controlled-$U_b^y$ between the second and third registers.

\item Measure and discard the third register. This step generates a so-called coset  state
\[ 
\frac{1}{\sqrt{p-1}}\sum_{k=0}^{p-1} \ket{\gamma  +ks,-k},
\]
where $\gamma$ is some uniformly random element of $\Z_{p-1}$ and $s$ is the discrete logarithm.

\item Apply the quantum Fourier transform over $\Z_{p-1}$ to the first two registers, to obtain
\[ 
\frac{1}{\sqrt{p-1}}\sum_{k'=0}^{p-1} \euler^{2\pii \frac{k'\gamma}{p-1}}\ket{k',k's},
\]
\item Measure the system in the standard basis to obtain a pair of the form $(k',k's)\bmod{p}$ uniformly at random.

\item Classical post-processing. By repeating the above process $n$ times, one can extract the discrete logarithm $s$ from these pairs with exponentially high probability (at least $1- 2^{-n}$), in classical polynomial time.
\end{enumerate}
\end{algorithm}
Note that the Hilbert space of the third register precisely corresponds to  $\mathcal{H}_{\mathbf{B}}$ if we choose the black-box group to be $\mathbf{B}=\Z_p^\times$. It is now easy to realize that Shor's algorithm for discrete log is a normalizer circuit over $\Z_{p-1}\times\Z_{p-1}\times\Z_{p}^\times$: steps 2 and 4 correspond to applying partial QFTs over $\Z_{p-1}$, and the gate $U$ applied in state 3 is a group automorphism over $\Z_{p-1}\times\Z_{p-1}\times\Z_{p}^\times$. 
\end{proof}
We stress that, in the proof above, there is no known efficient classical algorithm for solving the group decomposition problem for the group $\Z_p^\times$  (as we define it in section \ref{sect:Group Decomposition}): although, by assumption, we know that $\Z_p^\times=\langle a \rangle \cong \Z_{p-1}$, this information does not allow us to convert elements from one representation to the other, since this requires solving the discrete-logarithm problem itself. In other words, we are unable to compute classically the \emph{group isomorphism} $\Z_p^\times \cong \Z_{p-1}$. In our version of the group decomposition problem, we require the ability to \emph{compute} this group isomorphism. For this reason, we treat the group  $\Z_p^\times$ as a \emph{black-box group}.

\subsection{Shor's factoring algorithm	}\label{sect:Factoring}

In this section we will show  that normalizer circuits can efficiently compute the order of elements of (suitably encoded) Abelian groups. Specifically, we show how to efficiently solve the {order finding problem}  for every  (finite) Abelian black-box  group $\mathbf{B}$  \cite{BabaiSzmeredi_Complexity_MatrixGroup_Problems_I} with normalizer circuits. Due to the well-known classical reduction of the factoring problem to the problem of computing orders of elements of the group $\Z_N^\times$, our result implies that black-box normalizer circuits can efficiently factorize large composite numbers, and thus break the widely used RSA public-key cryptosystem \cite{RSA}.

We briefly introduce the \textbf{order finding problem} over a  black-box group $\mathbf{B}$, that we always assume to be finite and Abelian. In addition, we assume that the elements of the black-box group can be uniquely encoded with  $n$-bit strings, for some known $n$. The task we consider is the following: given an  element $a$ of $\mathbf{B}$, we want to compute the order $|a|$ of $a$ (the  smallest positive integer $r$ with the property\footnote{Since $\mathbf{B}$ is finite, the order $|a|$ is  a well-defined number.} $a^r=1$).  Our next theorem states that this version of the order finding problem can be efficiently solved by a quantum computer based on normalizer circuits.
\begin{theorem}[\textbf{Order finding over $\mathbf{B}$}]\label{thm:Order Finding}  Let $\mathbf{B}$ be a finite Abelian black-box group with an $n$-qubit encoding and  $\mathcal{H}_\mathbf{B}$ be the Hilbert space associated with this group. Let $V_a$ be the unitary that performs the group multiplication operation on $\mathcal{H}_\mathbf{B}$: $V_a\ket{x}=\ket{ax}$. We denote by $c\mbox{\,-}V_a$  the unitary that performs $V_a$ on $\mathcal{H}_\mathbf{B}$ controlled on the value of an ancillary register $\mathcal{H}_\Z $:
 \begin{equation}\notag
 \ket{m,x}\quad\xrightarrow{\quad c-U_{a}\quad}\quad\ket{m,a^mx},\qquad\textnormal{for any $m$ in } \Z.
 \end{equation}
Assume that we can query an oracle that implements $c\mbox{\,-}V_a$ in one time step for any  $a\in \mathbf{B}$. Then, there exists a hybrid version of Shor's order-finding algorithm, which can compute the order $|a|$ of any $a\in\mathbf{B}$ efficiently, using  normalizer circuits over the group $\Z\times \mathbf{B}$ and classical post-processing. The algorithm runs in polynomial-time, uses an efficient amount of precision and succeeds with high probability.
\end{theorem}
In theorem \ref{thm:Order Finding},  by ``efficient amount of precision'' we mean that instead of preparing Fourier basis states of $\mathcal{H}_{\Z}$ or measuring on this (unphysical) basis, it is enough to use realistic physical approximations of these states (cf. section \ref{sect:Normalizer circuits over blackbox groups}).
\begin{proof}
For simplicity, we assume that a generating set of $\mathbf{B}$ with $O(n)$ elements is given (otherwise we could generate one efficiently probabilistically by sampling elements of $\mathbf{B}$). 

We divide the proof into two steps. In the first part, we give an infinite-precision  quantum algorithm to randomly sample elements from the set $\mathrm{Out}_a=\{\frac{k}{|a|}: k\in \Z\}$ that uses normalizer circuits over the group $\Z\times\mathbf{B}$ in polynomially many steps. In this first algorithm, we assume that Fourier basis states of $\mathcal{H}_\Z$ can be prepared perfectly and that there are no physical limits in measurement precision; the outcomes $k/|a|$ will be stored with floating point arithmetic and with finite precision. The algorithm allows one to extract the period $|a|$ efficiently by sampling fractions $k/|a|$ (quantumly) and then using a continued fraction expansion (classically).

 In the second part of the proof, we will remove the infinite precision assumption.

Our first algorithm is essentially a variation of Shor's algorithm for order finding \cite{Shor} with one key modification: whereas Shor's algorithm uses a large $n$-qubit register $\mathcal{H}_2^n$ to estimate the eigenvalues of the unitary $V_a$, we will replace this multiqubit register with a single \emph{infinite} dimensional Hilbert space $\mathcal{H}_{\Z}$. The algorithm is \emph{hybrid} in the sense that it involves both continuous- and discrete-variable registers. The key feature of this algorithm is that, at every time step, the implemented gates are \emph{normalizer gates}, associated with the groups $\Z\times \Z_N^\times$ and  $\T\times \Z_N^\times$ (which are, themselves, related via the partial Fourier transforms $\mathcal{F}_\Z$ and $\mathcal{F}_\T$). The algorithm succeeds with constant probability.
\begin{algorithm}[\textbf{Hybrid order finding with infinite precision}]\label{alg:Order Finding infinite precision}
 $\,$
 \begin{alg_in}
A black box (finite abelian) group $\mathbf{B}$, and an element $a \in \mathbf{B}$.
 \end{alg_in}
 \begin{alg_out}
 The order $s$ of $a$ in $\mathbf{B}$, i.e. the least positive integer $s$ such that $a^s = 1$.
 \end{alg_out}
 We will use multiplicative notation for the black box group $\mathbf{B}$, and additive notation for all other subgroups.
\begin{enumerate}
\item \textbf{Initialization:} Initialize $\mathcal{H}_\Z$ on the Fourier basis state $\ket{0}$ with $0\in \T$, and $\mathcal{H}_{\mathbf{B}}$ on the state $\ket{1}$, with $1\in \mathbf{B}$. In our formalism, we will regard $\ket{0,1}$ as a standard-basis state of the basis labeled by $\T\times\mathbf{B}$.
\item  Apply the Fourier transform $\mathcal{F_\T}$ to the register $\mathcal{H}_\Z$.  This changes the designated basis of this register to be  the one labeled by the group $\Z$. The state $\ket{0}$ in the new basis is an infinitely-spread comb of the form $\sum_{m\in\Z}\ket{m}$.
\item Let the oracle $V_a$ act jointly on $\mathcal{H_\Z}\times \mathcal{H}_\mathbf{B}$; then the state is mapped in the following manner:
\begin{equation}
\sum_{m\in\Z}\ket{m}\ket{1}\quad\xrightarrow{\quad c\mbox{\,-}V_a\quad }\quad\sum_{m\in \Z} \ket{m, a^m}.
\end{equation}
Note that, in our formalism,  the oracle $c\mbox{\,-}V_a$  can be regarded as an automorphism gate $U_\alpha$. Indeed, the gate implements a classical invertible function on the group $\alpha(m,x)=(m,a^mx)$. The function is, in addition, a continuous\footnote{This is vacuously true: since the group $G:=\Z\times\mathbf{B}$ is discrete, \emph{any} functtion $f:G\rightarrow G$ is continuous.} group automorphism, since
\begin{align}	\notag
\alpha\left((m,x)(n,y)\right)=\alpha(m+n,xy)&=(m+n,(a^{m+n})(xy))\\&=(m+n,(a^{m}x)(a^{n}y))=(m,a^{m}x)(n,a^{n}y)\label{eq:ModExp is Automorphism}\\
&=\alpha(m,x)\alpha(n,y).\notag
\end{align}
\item Measure and discard the register $\mathcal{H}_{\mathbf{B}}$. Say we obtain $a^s$ as the measurement outcome. Note that the function $a^m$ is periodic with period $r=|a|$, the order of the element. Due to periodicity, the state after measuring $a^s$ will be of the form
\begin{equation}\label{eq:Periodic State}
\left(\sum_{j\in \Z} \ket{s+jr}\right) \ket{a^s}.
\end{equation}	
After dicarding $\mathcal{H}_{\mathbf{B}}$ we end up in a periodic state $\sum \ket{s+jr}$ which encodes $r=|a|$.
\item  Apply  the Fourier transform $\mathcal{F_\Z}$ to the  register $\mathcal{H}_\Z$.  We work again in the Fourier basis of $\mathcal{H}_\Z$, which is labelled by the circle group $\T$. The periodic state $\sum\ket{s+jr}$ in the dual $\T$ basis reads \cite{oppenheim_Signals_and_Systems}
\begin{equation}\label{eq:Periodic state after Fourier Transform}
\sum_{k=0}^{r-1} \euler^{2\pii \frac{sk}{r}} \ket{\tfrac{k}{r}}  
\end{equation}
\item Measure $\mathcal{H}_{\Z}$ in the Fourier basis (the basis labeled by $\T$). Since we that the initial state of the computation is as close to $\ket{0}$ as we wish, the wavefunction of the final state  (\ref{eq:Periodic state after Fourier Transform}) is \emph{sharply peaked} around values $p\in\T$ of the form $k/r$. As a result, a high resolution measurement will let us sample these numbers (within some floating-point precision window $\Delta$) nearly uniformly at random.
\item \textbf{Classical postprocessing:} Repeat steps 1-7 a few times and use a (classical) continued-fraction expansion algorithm \cite{nielsen_chuang,KLM_QC_07}  to extract the order $r$ from the randomly sampled multiples $\{k_i/r\}_i$. This can be done, for instance, with an  algorithm from \cite{Knill95onshors} that obtains $r$ with \emph{constant} probability after sampling two numbers $\tfrac{k_1}{r}$, $\tfrac{k_2}{r}$, if the measurement resolution is high enough:  $\Delta \leq 1/2r^2$ is enough for our purposes.
\end{enumerate}\vspace*{-10pt}
\end{algorithm}
Manifestly,  there is a strong similarity between algorithm \ref{alg:Order Finding infinite precision} and Shor's factoring algorithm:  the quantum Fourier transforms $\mathcal{F}_\T$ in our algorithm $\mathcal{F}_\Z$ play the role of the discrete Fourier transorm $\mathcal{F}_{2^n}$, and $c\mbox{\,-}V_a$  acts as the modular exponentation gate \cite{Shor}.  In fact, one can regard algorithm \ref{alg:Order Finding infinite precision} as a ``hybrid'' version of Shor's algorithm combining both  continuous and discrete variable registers. The remarkable feature of this version of Shor's algorithm is that the quantum part of the algorithm 1-6 is a normalizer computation.

Algorithm \ref{alg:Order Finding infinite precision} is efficient if we just look at the number of gates it uses. However, the algorithm is \emph{inefficient} in that it uses infinitely-spread Fourier states $\ket{p}=\sum_{m\in\Z}\euler^{-2\pii pm}\ket{m}$ (which are unphysical and cannot be prepared with finite computational resources) and arbitrarily precise measurements. We finish the proof of theorem \ref{thm:Order Finding} by giving an improved algorithm that does not rely on unphysical requirements.

\newpage

\begin{algorithm}[\textbf{Hybrid order finding with finite precision}] \label{alg:Order Finding Finite precision}$\,$
\begin{enumerate}
 \item[1-2] \textbf{Initialization:} Initialize $\mathcal{H}_{\mathbf{B}}$ to $\ket{1}$. The register  $\mathcal{H}_\Z$ will begin in an \emph{approximate} Fourier basis state $\ket{\widetilde{0}}= \tfrac{1}{\sqrt{2M+1}}\sum_{-M}^{+M}\ket{m}$,  i.e.\ a square pulse of length $2M+1$ in the integer basis, centered at 0. This step simulates steps 1-2 in algorithm \ref{alg:Order Finding infinite precision}.
 
 \item[3-4] Repeat steps 3-4 of algorithm \ref{alg:Order Finding infinite precision}. The state after obtaining the measurement outcome $a^s$ is now different due to the finite ``length'' of the comb $\sum_{m=0}^{M} \ket{m}$; we obtain
 \begin{equation}\label{eq:Periodic state finite}
|\psi\rangle =\frac{1}{\sqrt{L}}\sum_{-L_a}^{L_b}\ket{s+jr},
 \end{equation}
 where $L=L_a+L_b+1$ and $s$ is obtained nearly uniformly at random from $\{ 0,\ldots, r-1\}$. The values $L_a$, $L_b$ are positive integers of of the form  $\lfloor M/r\rfloor - \epsilon $ with $-2\leq \epsilon\leq 0$  (the particular value of $\epsilon$ depends on $s$, but it is irrelevant in our analysis). Consequently,  we have  $L=2\lfloor M/r\rfloor - (\epsilon_a +\epsilon_b)$.

\item[5]  Apply  the Fourier transform $\mathcal{F_\Z}$ to the  register $\mathcal{H}_\Z$ . The wavefunction of the final state  $\hat{\psi}$  is the Fourier transform of  the wavefunction $\psi$ of  (\ref{eq:Periodic state finite}). We  compute $\hat{\psi}$ using formula (\ref{eq:Fourier transform over Z}):
\begin{align}
 \hat{\psi}(p)&=\sum_{x\in\Z} \euler^{2\pii p x}\psi(x)=\frac{1}{\sqrt{L}}\sum_{-L_a}^{L_b} \euler^{2\pii p (s+jr)}=\frac{1}{\sqrt{L}} \left(\euler^{2\pii p s}\right)\frac{\euler^{2\pii p r (L_b+1)}-\euler^{-2\pii p r L_a}}{\euler^{2\pii p r }-1} \notag\\
 &= \frac{\euler^{2\pii p\left( s+\tfrac{L_b-L_a}{2}\right)}}{\sqrt{L}}  \frac{ \sin{ \left( \pi L pr  \right)} } { \sin{\left( \pi pr  \right)}} = \frac{\euler^{2\pii p\left(s +\tfrac{L_b-L_a}{2}\right) }}{\sqrt{L}} D_{L,r}({p}) 	\label{eq:Dirichlet State}
\end{align}
(to derive the equation, we apply the summation formula of the geometric series and  re-express the result in terms of the Dirichlet kernel \cite{rudin62_Fourier_Analysis_on_groups}
\begin{equation}\label{eq:Dirichlet Kernel}
D_{L,r}(p)= \frac{\sin{ \left( \pi L pr \right)}}{\sin{\left( \pi pr  \right)}}.
\end{equation} 
\item[6] \textbf{Measure $\mathcal{H}_{\Z}$ } in the Fourier basis. We show now that, if  the resolution is high enough, then the probability distribution of measurement outcomes will be ``polynomially close'' to the one obtained in the infinite precision case (\ref{eq:Periodic state after Fourier Transform}). Intuitively, this is a consequence of the fact that in the limit $M\rightarrow \infty$  (when the initial state becomes an infinitely-spread comb), we have also $L\rightarrow\infty$ and that the function $D_{L},r(p)$  converges  to a train  $\sum_{k=0}^{r-1} \delta_{k/r}(p)$ of Dirac measures \cite{rudin62_Fourier_Analysis_on_groups}. In addition, for a high finite value of $M$, we  find that the probability of obtaining some outcome  $p$ within a $\Delta=\tfrac{1}{Lr}$ window of a fraction  $\tfrac{k}{r}$ is also high.
\begin{equation}\label{eq:Probability in Shor Infinite}
\mathrm{Pr}(|p-\tfrac{k}{r}|\leq \tfrac{\Delta}{2})=\frac{1}{L} \int_{-\tfrac{\Delta}{2}}^{+\tfrac{\Delta}{2}} \frac{\sin^2{ \left( \pi L pr \right)}}{\sin^2{\left( \pi pr  \right)}}  \,\mathrm{d}p  \geq \tfrac{\Delta}{L}\frac{\sin^2{ \left( \tfrac{\uppi}{2} \right)}}{\sin^2{\left( \tfrac{\uppi}{2L}   \right)}} \geq \frac{4}{\uppi^2r},
\end{equation}
where we use the mean value theorem  and the bound $\sin(x)^2\leq x^2$.  It follows that with  \emph{constant} probability (larger than $4/\pi^2\approx0.41$)  the measurement will output some outcome $\tfrac{\Delta}{2}$-close to a number of the form $k/r$. (A tighter lower bound of $2/3$ for the success probability can be obtained by evaluating the integral numerically.)

Lastly, note that although the derivation of (\ref{eq:Probability in Shor Infinite}) implicitly assumes that the finial measurement is infinitely precise, it is enough to implement measurements with resolution close to $\Delta$. Due to the peaked shape of the final distribution (\ref{eq:Probability in Shor Infinite}), it follows that $\Theta(\tfrac{1}{M})$ resolution is enough if our task is to sample $\tfrac{\Delta}{2}$-estimates of these fractions nearly uniformly at random; this scaling  is \emph{efficient} as a function of   $M$ (cf. section \ref{sect:Normalizer circuits over blackbox groups}).

\item[7] \textbf{Classical postprocessing:} We now set $M$ (the length of the initial comb state) to be large enough so that $\tfrac{\Delta}{2}=\tfrac{1}{2Lr}\leq \tfrac{1}{2r^2}$; taking $\log M = O(\poly n)$ is enough for our purposes. With such an $M$, the measurement step 6 will output a number $p$ that is $\tfrac{1}{2r^2}$ close to a $\tfrac{k}{r}$ with  high probability, which can be increased to be arbitrarily close to 1 with a few repetitions. We then proceed as in step 7 of algorithm \ref{alg:Order Finding infinite precision} to compute the order $r$.\qedhere
\end{enumerate}\vspace*{-10pt}
\end{algorithm}
\end{proof}

\subsubsection*{Shor's algorithm as a normalizer circuit}

Our discussion in the previous section reveals strong a resemblance between our hybrid normalizer quantum algorithm for order finding and Shor's original quantum algorithm for this problem \cite{Shor}: indeed, both quantum algorithms employ remarkably similar circuitry. In this section we show that this resemblance is actually more than a mere fortuitous analogy, and that, in fact, one can understand Shor's original order-finding algorithm as a discretized version of our finite-precision hybrid algorithm for order finding \ref{alg:Order Finding infinite precision}.
\begin{theorem}[\textbf{Shor's algorithm as a normalizer circuit}] \label{thm:Shor normalizer} Shor's order-finding algorithm \cite{Shor} provides an efficient discretized implementation of our hybrid normalizer algorithm \ref{alg:Order Finding Finite precision}. 
\end{theorem}
Note that the theorem does not imply that all possible quantum algorithms for order finding are normalizer circuits (or discretized versions of some normalizer circuit). What it shows is that the one first found by Shor in \cite{Shor} does exhibit such a structure.
\begin{proof}
Our approach will be to show explicitly that the evolution of the initial quantum state  in Shor's algorithm is analogous to that of the initial state in algorithm \ref{alg:Order Finding Finite precision} if we discretize the computation. Recall that Shor's algorithm implements a quantum phase estimation \cite{kitaev_phase_estimation} for the unitary $V_a$. Let $D$ be the dimension of the Hilbert space used to record such phase. We assume $D$ to be odd\footnote{This choice is not essential, neither in Shor's algorithm nor in algorithm \ref{alg:Order Finding Finite precision}, but it simplifies the proof.} and write $D=2M+1$. Then Shor's algorithm can be written as follows:
\begin{enumerate}
\item Initialize the state $\ket{0,1}$ on the Hilbert space $\mathcal{H}_D\times \mathcal{H}_{\Z_N^\times}$.
\item Apply the discrete Fourier transform $\mathcal{F}_{\Z_D}$ on $\mathcal{H}_D$ to obtain
\begin{equation}
\sum_{m=0}^{D-1} \ket{m}\ket{1}=\sum_{-M}^{M} \ket{m}\ket{1}.
\end{equation}
So far, we have simulated  step 1 in algorithm \ref{alg:Order Finding Finite precision}  by constructing the same periodic state. These first two steps are also clearly analogous to  steps 1-2 in algorithm \ref{alg:Order Finding infinite precision}.
\item[3-4] Apply the modular exponentiation gate $U_{\mathrm{me}}$, which is the following unitary  \cite{Shor}
\begin{equation}\label{eq:Modular Exponentiation Gate}
U_{\mathrm{me}}\ket{m,x}=\ket{m,a^mx },
\end{equation}
to the state. Measure the register $\mathcal{H}_{\Z_N^\times}$ in the standard basis. We obtain, again, a quantum state of the form (\ref{eq:Periodic state finite}), with $L\leq D$. 

\item[6] We apply the discrete Fourier transform $\mathcal{F}_{\Z_D}$ to the register $\mathcal{H}_{\Z_D}$ again. We claim now that the output state will be a discretized version of (\ref{eq:Dirichlet State}) due to a remarkable \textbf{mathematical correspondence} between  Fourier transforms. Note that any quantum state $\ket{\psi}$ of the infinite-dimensional Hilbert space $\mathcal{H}_{\Z}$ can be regarded as a quantum state of $\mathcal{H}_{D}$ given that the support of $\ket{\psi}$ is limited to the standard basis states $\ket{0}, \ket{\pm 1}, \ldots, \ket{\pm M}$. Let us denote the latter state $\ket{\psi_D}$ to distinguish both. Then, we observe a correspondence between letting $\mathcal{F}_\Z$ act on $\ket{\psi}$ and letting $\mathcal{F}_{\Z_D}$ act on $\ket{\psi_D}$.
\begin{equation}\label{eq:Fourier Correspondence}
 \displaystyle\hat{\psi}(p)= \sum_{x=-M}^{x=+M} \euler^{2\pii p x}\psi(x)\qquad \longleftrightarrow\qquad\displaystyle\hat{\psi}_D(k)= \sum_{x=-M}^{x=+M} \euler^{2\pii \tfrac{kx}{D}}\psi_D(x)
\end{equation}
The correspondence (equation \ref{eq:Fourier Correspondence}) tells us that, since we have $\psi(x)=\psi_D(x)$, it follows that the Fourier transformed function $\hat{\psi}_D(k)$ is precisely the function $\hat{\psi}(p)$ evaluated at points of the form $p=\tfrac{k}{D}$. The final state can be written as 
\begin{equation}\label{eq:Discretized Shor Output}
\sum_{k=0}^{D-1} \hat{\psi}\left(\tfrac{k}{D}\right)\ket{k}. 
\end{equation}
which is, indeed, a discretized version of  (\ref{eq:Dirichlet State}). 
\item[7-8] The last steps of Shor's algorithm are identical to 7-8 in algorithm \ref{alg:Order Finding Finite precision}, with the only difference being that the wavefunction (\ref{eq:Discretized Shor Output}) is now a discretization of (\ref{eq:Dirichlet State}). The probability of measuring a number $k$ such that $\tfrac{k}{D}$ is close to a multiple of the form $\tfrac{k'}{r}$ will again be high, due to the properties of the Dirichlet kernel (\ref{eq:Dirichlet Kernel}). Indeed, one can show (see, e.g.\, \cite{childs_vandam_10_qu_algorithms_algebraic_problems}) with an argument similar to (\ref{eq:Probability in Shor Infinite}) that, by setting $D=N^2$, the algorithm outputs with constant probability and almost uniformly a fraction $\tfrac{k}{D}$ among the two closest fraction to some value of the form $k/r$ (see e.g. \cite{Shor} for details). The period $r$ can be recovered, again, with a continued fraction expansion.\qedhere
\end{enumerate}
 \end{proof}

 \subsubsection*{Normalizer gates over $\infty$ groups are necessary to factorize}
 
 At this point, it is a natural question to ask whether it is necessary at all to replace the Hilbert space $\mathcal{H}_2^n$ with an infinite-dimensional space $\mathcal{H}_\Z$ with an integer basis in order to be able to factorize with normalizer circuits. We discuss in this section that, in the view of the authors, this is a \textbf{key indispensable ingredient} of our proof. 
 
 We begin our discussion by showing rigorously,  in the black-box set-up, that  no quantum algorithm for factoring based on \emph{modular exponentation} gates (controlled $V_a$ rotations) can be efficiently implemented with normalizer circuits over finite Abelian groups, in a strong sense.
 \begin{theorem} \label{thm:ModExp requires Z} Let $\mathcal{H}_M$ be the Hilbert space with basis $\{\ket{0},\ldots,\ket{M-1}\}$ and dimension $M$.  Let $\mathbf{B}$ be an Abelian black-box group with associated Hilbert space $\mathcal{H}_{\mathbf{B}}$. Consider the composite Hilbert space $\mathcal{H}=\mathcal{H}_M\times \mathcal{H}_{\mathbf{B}}$ and define  $U_\mathrm{me}$ to be the unitary gate on $\mathcal{H}$ defined as $U_{\mathrm{me}}\ket{m,x}=\ket{m,a^mx }$, where $a,x \in \mathbf{B}$ and  $m\in\Z_M$. Then, unless $M$ is a multiple of the order of $a$, there does not exist any normalizer circuit over $\mathcal{H}$ (even of exponential size) satisfying $\|\mathcal{C}-U_\mathrm{me} \|_{\mathrm{op}}\leq1-{2}^{-1/2}$.
 \end{theorem}
 We prove the theorem in appendix \ref{app:ModExp requires Z}. We highlight that a similar result was proven in  \cite[theorem 2]{VDNest_12_QFTs}: that normalizer circuits over groups of the form $\Z_{2^n}\times\Z_N$ also fail to approximate the modular exponentiation.  Also, we point out that it is easy to see that the converse of theorem \ref{thm:ModExp requires Z} is also true:  if $|a|$ divides $M$, then an argument similar to (\ref{inproof:Ume is no Clifford})  shows that $(m,x)\rightarrow(m,a^mx)$ is a group automorphism of $\Z_M\times \mathbf{B}$, and the gate  $U_\mathrm{mf}$ automatically becomes a normalizer automorphism gate.
 
The main implication of theorem \ref{thm:ModExp requires Z}  is that finite-group normalizer circuits \emph{cannot} implement nor approximate the quantum modular exponentiation gate between $\mathcal{H}_\mathbf{B}$, playing the role of the target system, and some ancillary control system, \emph{unless} a multiple $M=\lambda |a|$ of the order  of $a$ is known in advance. Yet the problem of finding  multiples of orders is \emph{at least as hard as factoring and order-finding}: for $\mathbf{B}=\Z_N^\times$, a subroutine to find multiples of orders  can be used to efficiently compute classically a multiple of the order of the group $\varphi(N)$, where $\varphi$ is the Euler totient function, and  it is known that factoring is  polynomial-time reducible to the problem of finding a single multiple of the form  $\lambda\varphi(N)$  \cite{Shoup08_A_Computational_Introducttion_to_Number_Theory_and_Algebra}. 

We arrive to the conclusion that,  unless we are in the trivial case where we know how to factorize in advanced, a factoring algorithm based on finite-group normalizer gates cannot comprise controlled-$V_a$ rotations. We further conjecture that any other approach based on finite-group normalizer gates cannot work either.
 \begin{conj} \label{conj: Factoring not over finite groups}
 Unless factoring is contained in \textnormal{BPP}, there is no efficient quantum algorithm to solve the factoring problem using only normalizer circuits over finite Abelian groups (even when these are allowed to be black-box groups) and classical pre- and  post-processing.
 \end{conj}
 We back up our conjecture with two facts. On one hand, Shor's algorithm for factoring \cite{Shor} (to our knowledge, the only quantum algorithm for factoring thus far) uses a modular exponentiation gate to estimate the phases of the unitary $V_a$, and these gates are hard to implement with finite-group normalizer circuits due to theorem \ref{thm:ModExp requires Z}. On the other hand,  the reason why this does works for the group $\Z$ seems to be, in the view of the authors,   intimately related to the fact that the order-finding problem can be naturally casted as an instance of the Abelian \textbf{\emph{hidden subgroup problem}} over $\Z$  (see also section \ref{sect:Abelian HSPs}).  Note that, although one can always cast the order-finding problem  as an HSP over any finite group $\Z_{\lambda \varphi(N)}$ for an integer $\lambda$, this formulation of the problem is unnatural in our setting, as it requires (again) the prior knowledge of a multiple of $\varphi(N)$, which we could use to factorize and find orders classically without the need of a quantum computer \cite{Shoup08_A_Computational_Introducttion_to_Number_Theory_and_Algebra}.

\subsection{Elliptic curves}\label{sect:Elliptic Curve}

In the previous sections we have seen that black-box normalizer circuits can compute discrete logarithm in $\Z_p^\times$ and break the Diffie-Hellman key exchange protocol. In the proof, we showed that Shor's algorithm for this problem decomposes naturally in terms of normalizer gates over $\Z_p^\times$, treated as a black-box group. 

It is known that Shor's algorithm can be adapted in order to compute discrete logarithms over arbitrary black-box groups. In particular, this can be done for the group of solutions $E$ of an elliptic curve \cite{ProosZalka03_Shors_DiscreteLog_Elliptic_Curves,Kaye05_optimized_Quantum_Elliptic_Curve,CheungMaslovMathew08_Design_QuantumAttack_Elliptic_CC}, thereby rendering elliptic curve cryptography (ECC) vulnerable. Efficient unique encodings and fast multiplication algorithms for these groups are known, so that they can be formally treated as black-box groups. In this section, we show that a quantum algorithm given by Proos and Zalka  \cite{ProosZalka03_Shors_DiscreteLog_Elliptic_Curves} to compute discrete logarithms over elliptic curves can be implemented with black-box normalizer circuits.

\subsubsection*{Basic notions}

To begin, we review some  rudiments of the theory of elliptic curves. For simplicity, our survey focuses only on the particular types of elliptic curves that were studied in \cite{ProosZalka03_Shors_DiscreteLog_Elliptic_Curves}, over fields with characteristic different than 2 and 3. Our discussion applies equally to the (more general) cases considered in  \cite{Kaye05_optimized_Quantum_Elliptic_Curve,CheungMaslovMathew08_Design_QuantumAttack_Elliptic_CC}, although the definition of the  elliptic curve group operation becomes more cumbersome in such settings\footnote{Correspondingly, the complexity of performing group multiplications in \cite{Kaye05_optimized_Quantum_Elliptic_Curve,CheungMaslovMathew08_Design_QuantumAttack_Elliptic_CC} is greater.}. For more details on the subject, in general, we refer the reader to \cite{childs_vandam_10_qu_algorithms_algebraic_problems, lorenzini1997invitation}.

Let $p> 3$ be prime and let $K$ be the field defined by endowing the set $\Z_p$ with the addition and multiplication operations modulo $p$.  An \emph{elliptic curve} $E$ over the field $K$ s a finite Abelian group formed by the solutions  $(x,y)\in K\times K$ to an equation
\begin{equation}\label{eq:Elliptic Curve}
C:y^2 = x^3 + \alpha x + \beta
\end{equation}
together with a special element $O$ called the ``point at infinity''; the coefficients $\alpha$, $\beta$ in this equation live in the field $K$. The discriminant  $\Delta:=-16(4\alpha^3 + 27 \beta^2)$ is nonzero, ensuring that  the curve is non-singular. The elements of $E$ are endowed with a commutative group operation. If $P\in E$ then $P+O= O + P = P$. The inverse element $-P$ of $P$ is obtained by the reflection of $P$ about the $x$ axis. Given two elements $P=(x_P,y_P)$ and $Q=(x_Q,y_Q)\in E$, the element $P+Q$ is defined via the following rule:
\begin{equation} \label{eq:Elliptic curve group operation}
P+Q= 
\begin{cases}
O & \textnormal{if $P=(x_P,y_P)=(x_Q,-y_Q)=-Q$,}  \\
 -R& \textnormal{otherwise (read below).} 
\end{cases}
\end{equation}
In the case $P\neq Q$, the point $R$ is computed as follows:

\begin{minipage}{0.4\textwidth}
\begin{align}
x_{R} &=\lambda^2 - x_P - x_Q  \notag \\
y_{R} &= y_P- \lambda(x_P-x_{R})  \notag 
\end{align}
\end{minipage}
\begin{minipage}{0.4\textwidth}
\begin{equation} 
\lambda:=\notag
\begin{cases}
\tfrac{y_Q-y_P}{x_Q-x_P} & \textnormal{if $P\neq Q$} \\
 \tfrac{3 x_P^2+\alpha}{2 y_P} & \textnormal{if $P= Q$ and $y_P \neq 0$}
\end{cases}
\end{equation}
\end{minipage}\\

\noindent  $R$ can also be defined, geometrically, to be the ``intersection between the elliptic curve and the line through $P$ and $Q$'' (with a minus sign) \cite{childs_vandam_10_qu_algorithms_algebraic_problems}. 

It is not hard to check form the definitions above that the elliptic-curve group $E$ is finite and Abelian; from a computational point of view, the elements of $E$ can be stored with $n\in O(\log |K|)$ bits and the group operation can be computed in $O(\poly{n})$ time. Therefore, the group $E$ can be treated as a \textbf{black box group}.

Finally, the \textbf{discrete logarithm problem} (DLP) over an elliptic curve is defined in a way analogous to the $\Z_p^\times$ case, although now we use additive notation: given $a$, $b\in E$ such that $xa=b$ for some integer $x$; our task is to find the least nonnegative integer $s$ with that property. The elliptic-curve DLP is believed to be intractable for classical computers and can be used to define cryptosystems  analog to Diffe-Hellman's \cite{childs_vandam_10_qu_algorithms_algebraic_problems}.

\subsubsection*{Finding discrete logarithms over elliptic curves with normalizer circuits	}

In this section we review Proos-Zalka's quantum approach to solve the DLP problem over an elliptic curve \cite{ProosZalka03_Shors_DiscreteLog_Elliptic_Curves}; their quantum algorithm is, essentially, a modification of Shor's algorithm to solve the DLP over $\Z_p^\times$, which we covered in detail in section \ref{sect:Discrete Log}. 

Our \textbf{main contribution} in this section is that Proos-Zalka's algorithm can be implemented with normalizer circuits over the group $\Z\times \Z \times E$. The proof reduces to combining ideas from sections \ref{sect:Discrete Log} and \ref{sect:Factoring} and will be sketched in less detail.
\begin{algorithm}[\textbf{Proos-Zalka's \cite{ProosZalka03_Shors_DiscreteLog_Elliptic_Curves}}]\label{alg:ProosZalka} $\quad $
\begin{alg_in}
An elliptic curve with associated group $E$ (the group operation is defined as per (\ref{eq:Elliptic curve group operation})), and two points $a,b \in E$. It is promised that $sa = b$ for some nonnegative integer $s$.
\end{alg_in}
\begin{alg_out}
Find the least nonnegative integer $s$ such that $sa=b$.
\end{alg_out}
\begin{enumerate}
\item We use a register $\mathcal{H}_E$, where $E$ is the group associated with the elliptic curve (\ref{eq:Elliptic Curve}), and two ancillary registers $\mathcal{H}$ of dimension $N=2^n$, associated with the group $A=\Z_N\times\Z_N$. The computation begins in the state $\ket{0,0,O}$, where $(0,0)\in A$ and $O\in E$.
\item Fourier transforms are applied to the ancillas to create the superposition $\sum_{(x,y)\in A}\ket{x,y,O}$. 
\item The following transformation is applied unitarily:
\begin{equation}
\sum_{(x,y)\in A}\ket{x,y,O}\quad  \xrightarrow{\quad c\mbox{\,-}U\quad } \quad \sum_{(x,y)\in A}\ket{x,y,xa + yb}.
\end{equation}
\item Fourier transforms are applied again over the ancillas and then measured, obtaining an outcome of the form $(x',y')$. These outcomes contain enough information to extract the number $s$, with similar post-processing techniques to those used in Shor's DLP algorithm.
\end{enumerate}
\end{algorithm}
Algorithm \ref{alg:ProosZalka} is not a normalizer circuit over $\Z_N\times \Z_N \times E$. Similarly to  the factoring case, the algorithm would become a normalizer circuit  if the classical transformation in step 3 was an automorphism gate; however, for this to occur, $N$ needs to be  a common multiple of the orders of $a$ and $b$ (the validity of these claims follows with similar arguments to those in section \ref{sect:Factoring}). In view of our results in sections \ref{sect:Discrete Log} and \ref{sect:Factoring}, one can easily come up with two approaches to implement algorithm \ref{alg:Discrete Log} using normalizer gates.
\begin{itemize}
\item[(a)] The first approach would be to use our normalizer version of Shor's algorithm (theorem \ref{thm:Order Finding}) to find the orders of the elements $a$ and $b$: normalizer gates over $\Z\times E$ would be used in this step. Then, the number $N$ in algorithm \ref{alg:ProosZalka} can be set so that all the gates involved become normalizer gates over $\Z_N\times \Z_N \times E$.
\item[(b)] Alternatively, one can choose not to compute the orders by making the ancillas infinite dimensional, just as we did in algorithm \ref{alg:Order Finding infinite precision}. The algorithm becomes a normalizer circuit over $\Z\times \Z \times E$: as in algorithm \ref{alg:Order Finding infinite precision},  the ancillas are initialized to the zero Fourier basis state, and the discrete Fourier transforms are replaced by QFTs over $\T$ (in step 2) and $\Z$ (in step 4). A finite precision version of the algorithm can be obtained in the same fashion as we derived algorithm \ref{alg:Order Finding infinite precision}. Proos-Zalka's original algorithm could, again, be interpreted as a discretization of the resulting normalizer circuit.
\end{itemize}

\subsection{The hidden subgroup problem}\label{sect:Abelian HSPs}

All problems we have considered this far---finding discrete logarithms and orders of Abelian group elements---fit inside  a general class of problems known as hidden subgroup problems over Abelian groups \cite{Brassard_Hoyer97_Exact_Quantum_Algorithm_Simons_Problem,Hoyer99Conjugated_operators,MoscaEkert98_The_HSP_and_Eigenvalue_Estimation,Damgard_QIP_note_HSP_algorithm}. Most quantum algorithms discovered in the early days of quantum computation solve problems that can be recasted as Abelian HSPs, including Deutsch's problem \cite{Deutsch85quantumtheory}, Simon's \cite{Simon94onthe}, order finding and discrete logarithms \cite{Shor}, finding hidden linear functions \cite{Boneh95QCryptanalysis}, testing shift-equivalence of polynomials \cite{Grigoriev97_testing_shift_equivalence_polynomials}, and Kitaev's Abelian stabilizer problem \cite{kitaev_phase_estimation,Kitaev97_QCs:_algorithms_error_correction}.

In view of our previous results, it is natural to ask how many of these problems can be solved within the normalizer framework. In this section we show that a well-known quantum algorithm  that solves the Abelian HSPs (in full generality) can be modeled as a normalizer circuit over an Abelian group $\mathcal{O}$. Unlike previous cases, the group involved in this computation cannot be regarded as a black-box group, as it will not be clear how to perform group multiplications of its elements. This fact reflects the presence of oracular functions with unknown structure are present in the algorithm, to which the group $\mathcal{O}$ is associated; thus,  we  call $\mathcal{O}$ an \emph{oracular group}. We will discuss, however, that this latter difference does not seem to be very substantial, and that the Abelian HSP algorithm can be naturally regarded as a normalizer computation.

\subsubsection*{The quantum algorithm for the Abelian HSP}

In the \emph{Abelian hidden subgroup} problem we are given a function $f:G\rightarrow X$ from an Abelian finite\footnote{In this section we assume $G$ to be finite for simplicity. For a case where $G$ is infinite, we refer the reader back to section \ref{sect:Factoring}, where we studied the order finding problem (which is a HSP over $\Z$).} group $G$ to a finite set $X$. The function $f$ is constant on cosets of the form $g+H$, where $H$ is a subgroup ``hidden'' by the function; moreover, $f$ is different between different cosets. Given $f$ as a black-box, our task is to find such a subgroup $H$.

The Abelian HSP is a hard problem for classical computers, which need to query the oracle $f$ a superpolynomial amount of times in order to identify $H$ \cite{childs_vandam_10_qu_algorithms_algebraic_problems}. In contrast, a quantum computer can determine $H$ in polynomial time $O(\polylog{|G|})$, and using the same amount of queries to the oracle. We describe next a celebrated quantum algorithm for this task \cite{Brassard_Hoyer97_Exact_Quantum_Algorithm_Simons_Problem,Hoyer99Conjugated_operators,mosca_phd}. The algorithm is efficient given that the group $G$ is explicitly given\footnote{If the group $G$ is not given in a factorized form, the Abelian HSP may still be solved by  applying  Cheung-Mosca's algorithm to decompose $G$ (see next section).} in the form $G=\DProd{d}{m}$ \cite{mosca_phd,cheung_mosca_01_decomp_abelian_groups,Damgard_QIP_note_HSP_algorithm}.
\begin{algorithm}[\textbf{Abelian HSP}]\label{alg:HSP} $ $
\begin{alg_in}
An explicitly decomposed finite abelian group $G=\DProd{d}{m}$, and oracular access to a function $f:\:G\rightarrow X$ for some set $X$. $f$ satisfies the promise that $f(g_1) = f(g_2)$ iff $g_1 = g_2+h$ for some $h \in H$, where $H\subseteq G$ is some fixed but unknown subgroup of $G$.
\end{alg_in}
\begin{alg_out}
A generating set for $H$.
\end{alg_out}
\begin{enumerate}
\item Apply the QFT over the group $G$ to an initial state $\ket{0}$ in order to obtain a uniform superposition over the elements of the group $\sum_{g\in G}\ket{g}$.
\item Query the oracle $f$ in an ancilla register, creating the state
\begin{equation}
\frac{1}{\sqrt{|G|}}\sum_{g\in G}\ket{g,f(g)}
\end{equation}
\item The QFT over $G$ is applied to the first register, which is then measured.
\item After repeating 1-3 polynomially many times, the obtained outcomes can be postprocessed classically to obtain a generating set of $H$ with exponentially high probability (we refer the reader to \cite{lomont_HSP_review} for details on this classical part).
\end{enumerate}
\end{algorithm}
We now claim that the quantum part of algorithm \ref{alg:HSP} is a \emph{normalizer circuit}, of a slightly more general kind than the ones we have already studied. The normalizer structure of the HSP-solving quantum circuit is, however, remarkably well-hidden compared to the other quantum algorithms that we have already studied. It is indeed a   surprising fact that there is \emph{any} normalizer structure in the circuit, due to the presence of an oracular function, whose inner structure appears to be completely unknown to us!
\begin{theorem}[\textbf{The Abelian HSP algorithm is a normalizer circuit.}]\label{thm:HSP} In any Abelian hidden subgroup problem, the subgroup-hiding property of the oracle function $f$ induces a group structure $\mathcal{O}$ in the set $X$. With respect to this hidden ``linear structure'', the function $f$ becomes a group homomorphism, and the HSP-solving quantum circuit becomes a normalizer circuit over $G\times \mathcal{O}$.
\end{theorem}
The proof is the content of the next two sections.

\subsubsection*{Unweaving the hidden-subgroup oracle}

The key ingredient in the proof of the theorem (which is the content of the next section) is to realize that the oracle $f$ cannot fulfill the subgroup-hiding property without having a hidden homomorphism structure, which is also present in the quantum algorithm. 

First, we show that $f$ induces a \emph{\textbf{group structure}} on $X$. Without loss of generality, we  assume that the function $f$ is surjective, so that $\mathrm{im}f = X$. (If this is not true, we can redefine $X$ to be the image of $f$.) Thus, for every element  $x\in X$, the preimage $f^{-1}(x)$ is contained in $G$, and is a coset of the form  $f^{-1}(x)= g_x+H$, where $H$ is the hidden subgroup and $f(g_x)=x$. With these observations in mind, we can define a group operation in $X$  as follows:
\begin{equation}\label{eq:Oracular Group Operation}
x\cdot y = \tilde{f}\left(f^{-1}(x)+f^{-1}(y)\right).
\end{equation}
In (\ref{eq:Oracular Group Operation}) we denote by $\tilde{f}$ the function $\tilde{f}(x+H)=f(x)$ that sends cosets $x+H$ to elements of $X$. The subgroup-hiding property guarantees that this function is well-defined; moreover, $f$ and $\tilde{f}$ are related via $f(x)=\tilde{f}(x+H)$. The addition operation on cosets $f^{-1}(x)=g_x+H$ and  $f^{-1}(y)= g_y+H$ is just the usual group operation of the quotient group $G/H$ \cite{Humphrey96_Course_GroupTheory}:
\begin{equation}\label{eq:Factor Group DEfinition}
f^{-1}(x)+f^{-1}(y)=\left(g_x+H\right)+\left(g_y+H\right)=(g_x+g_y)+ H.
\end{equation}
By combining the two expressions, we get an explicit formula for the group multiplication in terms of coset representatives: $x\cdot y = f(g_x+g_y)$. It is routine to check that this operation is associative and invertible, turning $X$ into a group, which we denote by $\mathcal{O}$. The neutral element of the group is the string $e$ in $X$ such that $e=f(0)=f(H)$, which we show explicitly:
\begin{equation}
x\cdot e = e\cdot x= \tilde{f}\left(f^{-1}(x)+f^{-1}(e)\right) = \tilde{f}\left(f^{-1}(x)+ H \right) = x
\end{equation}
The group $\mathcal{O}$ is manifestly finite and Abelian---the latter property is due to the fact that  the addition  (\ref{eq:Factor Group DEfinition}) is commutative. 

Lastly, it is straightforward to check that the oracle $f$  \textbf{\emph{is a group homomorphism}} from $G$ to $\mathcal{O}$: for any $g$, $h\in G$ let $x:=f(g)$ and $y:=f(h)$, we have 
\begin{align}\label{eq:HSP oracle is group homomorphism}
f(g+h)&=\tilde{f} \left(g+h+H\right)=\tilde{f} \left(\left(g+H\right)+\left(h+H\right)\right)=\tilde{f} \left(f^{-1}\left(x\right)+f^{-1}\left(y\right)\right)\\
&=x\cdot y = f(g)\cdot f(h).
\end{align}
It follows from the first isomorphism theorem in group theory \cite{Humphrey96_Course_GroupTheory} that $\mathcal{O}$ is isomorphic to the quotient group  $G/H$  via the map $\tilde{f}$.

\subsubsection*{The HSP quantum algorithm is a normalizer circuit}

We will now analyze the role of the different quantum gates used in algorithm \ref{alg:HSP} and see that they are examples of normalizer gates over the group $G\times \mathcal{O}$, where $\mathcal{O}$ is the oracular group that we have just introduced.  

The Hilbert space underlying the computation can be written as $\mathcal{H}_G\otimes \mathcal{H}_\mathcal{O}$ with the standard basis  $\left\{\ket{g,x}:g\in G,\,x\in\mathcal{O}\right\}$. associated with this group. We will initialize the ancillary  registers to the state $\ket{e}$, where $e=f(0)$ is the neutral element of the group; the total state at step 1 will be $\ket{0,e}$. The Fourier transforms in steps 1 and 3 are just partial QFTs over the group $G$, which are normalizer gates. The quantum state at the end of step 1 is $\sum \ket{g,e}$. 

Next, we look now at step 2 of the computation: 
\begin{equation}
\sum \ket{g,e} \quad \longrightarrow \quad  \frac{1}{\sqrt{|G|}}\sum_{g\in G}\ket{g,f(g)}.
\end{equation}
This step can be implemented by a normalizer automorphism gate defined as follows. Let $\alpha:G\times \mathcal{O}\rightarrow G\times \mathcal{O}$ be the function $\alpha(g, x)= (g, f(g)\cdot x)$. Using the fact that $f:G\rightarrow \mathcal{O}$ is a group homomorphism (\ref{eq:HSP oracle is group homomorphism}), it is easy to check that $\alpha$  is a group automorphism of $G\times \mathcal{O}$. Then the evolution at step 2 corresponds to the action of the automorphism gate  $U_\alpha$:
\begin{equation}
U_\alpha \sum_g \ket{g,e}= \sum_g \ket{\alpha(g,e)} = \sum_g \ket{g,f(g)\cdot e} = \sum_g \ket{g,f(g)}.
\end{equation}
Of course, our choice to begin the computation in the state $\ket{0,e}$ and to apply $U_\alpha$ in step 3 is only one possible way to implement the first three steps of the algorithm. We could have alternatively initialized the computation on some $\ket{0,0}$ state and used a slightly different gate $U_\textrm{add}=\sum \ket{g,x+f(g)}\bra{g,x}$ in step 3. The latter sequence of gates can however be regarded as an exact gate-by-gate simulation of the former, so that it is perfectly licit to call the algorithm a normalizer computation---at least up to steps 3 and 4.

Finally, note that in the last step of the algorithm  we measure the  $\mathcal{H}_G$ in the standard basis like a normalizer computation. Therefore, every step in the quantum algorithm \ref{alg:HSP} corresponds to one of the allowed operations in a normalizer circuit over $G\times \mathcal{O}$. This finishes the proof of theorem \ref{thm:HSP}.

\subsubsection*{The oracular group $\mathcal{O}$ is not a black-box group (but almost)}\label{sect:Oracular Group is not Black Box}

We ought to stress, at this point, that although theorem \ref{thm:HSP} shows that the Abelian HSP quantum algorithm is a normalizer computation over an Abelian group $G\times \mathcal{O}$, the oracular group $\mathcal{O}$ is not a black-box group (as defined in section \ref{sect:Black Box Groups}), since it is not clear how to compute the group operation (\ref{eq:Oracular Group Operation}), due to our lack of knowledge about the oracular function which defines the multiplication rule. Yet, even in the absence of an efficiently computable group operation, we regard it natural to call the Abelian HSP quantum algorithm a normalizer circuit over $G\times \mathcal{O}$. Our reasons are multi-fold. 

First,  there is a manifest strong similarity between the quantum circuit in algorithm \ref{alg:HSP} and the other normalizer circuits that we have studied in previous sections, which suggests that normalizer operations naturally capture the logic of the Abelian HSP quantum algorithm. 

Second, it is in fact possible to argue that, although $\mathcal{O}$ is not a black-box group, it behaves \emph{effectively} as a black-box group in the quantum algorithm. Observe that, although it is true that one cannot generally compute $x\cdot y$ for arbitrary $x,\,y\in \mathcal{O}$, it is indeed always possible to multiply any element $x$  by the neutral element $e$, since the computation is trivial in this case: $x\cdot e = e\cdot x = x$. Similarly, in the previous section, it is not clear at all how to implement the unitary transformation $U_\alpha\ket{g,x}=\ket{g,f(g)\cdot x}$ for arbitrary inputs. However, for the restricted set of inputs that we need in the quantum algorithm (which is just the state $\ket{e}$), it is trivial to implement the unitary, for in this case $U_\alpha\ket{g,e}=\ket{g,f(g)}$; since quantum queries to the oracle function are allowed (as in step 2 of the algorithm), the unitary can be simulated by such process, regardless of how it is implemented. Consequently, the circuit \emph{effectively} behaves as a normalizer circuit over a black-box group.

Third,  although the oracular model in the black-box normalizer circuit setting is slightly different from the one used in the Abelian HSP they are still \emph{remarkably close} to each other. To see this, let $x_i$ be the elements of $X$ defined as $x_i:=f(e_i)$ where $e_i$ is the bit string containing a 1 in the $i$th position and zeroes elsewhere. Since the $e_i$s form a generating set of $G$, the $x_i$s generate the group $\mathcal{O}$. Moreover, the value of the function $f$ evaluated on  an element $g=\sum g(i)e_i$ is $f(g)=x_1^{g(1)} x_2^{g(2)}\cdots  x_m^{g(m)}$, since $f$ is a group homomorphism. It follows from this expression that the group homomorphism is \emph{implicitly multiplying} elements of the group $\mathcal{O}$. We cannot use this property to multiply elements of $\mathcal{O}$ ourselves, since everything happens at the hidden level. However, this observation shows that the assuming that $f$ is computable is \emph{tightly related} to the assumption that we can multiply in $\mathcal{O}$, although slightly weaker. (See also the next section.)

Finally, we mention that this very last feature can be exploited to extend several of our main results, which we derive in the black-box setting, to the more-general ``HSP oracular group setting'' (although proofs become more technical). For details, we refer the reader to sections \ref{sect:Simulation}-\ref{sect:Complete Problems} and appendix  \ref{app:Extending}.

\subsubsection*{A connection to a result by Mosca and Ekert}

Prior to our work,  it was observed by Mosca and Ekert \cite{MoscaEkert98_The_HSP_and_Eigenvalue_Estimation,mosca_phd} that $f$ must have a hidden homomorphism structure, i.e. that $f$ can be decomposed as $\mathcal{E}\circ\alpha$ where $\alpha$ is a group homomorphism between $G$ and another Abelian group $Q\cong G/H$, and $\mathcal{E}$ is a one-to-one hiding function from $Q$ to the set $X$. In this decomposition, $\mathcal{E}$ hides the homomorphism structure of the oracle.

Our result differs from Mosca-Ekert's in that we show that $X$ \emph{itself} can always be viewed as a group, with a group operation that is induced by the oracle, with no need to know the decomposition $\mathcal{E}\circ\alpha$. 

It is possible to relate both results as follows.  Since both $Q$ and $\mathcal{O}$  are isomorphic to  $G/H$,  they are also mutually isomorphic. Explicitly, if  $\beta$ is an isomorphism from $Q$ to $G/H$ (this map depends on the particular decomposition  $f=\mathcal{E}\circ\alpha$), then  $Q$ and $\mathcal{O}$ are isomorphic via the map $\tilde{f}\circ \beta$.

\subsection{Decomposing finite Abelian groups}\label{sect:Group Decomposition}

As mentioned earlier, there is a quantum algorithm for decomposing Abelian groups, due to Cheung and Mosca \cite{mosca_phd,cheung_mosca_01_decomp_abelian_groups}. In this section, we will introduce this problem, and present a quantum algorithm that solves it, which uses only  black-box normalizer circuits supplemented with classical computation.  The algorithm we give is based on Cheung-Mosca's, but it reveals some additional  information about the structure of the black-box group. We will refer to it as the \emph{\textbf{extended Cheung-Mosca's algorithm}}.

\subsubsection*{The group decomposition problem}

In this work, we define the    \textbf{group decomposition} problem as follows. The input of the problem is a list  of generators $\alpha=( \alpha_1,\cdots,\alpha_k)$ of some Abelian black-box group $\mathbf{B}$. Our task is to return a \emph{group-decomposition table} for $\mathbf{B}$. A group-decomposition table is a tuple $(\alpha, \beta, A, B, c)$ consisting of the original string $\alpha$ and four additional elements:
\begin{enumerate}
\item[(a)] A new generating set $\beta={\beta_1,\ldots,\beta_\ell}$ with the property $\mathbf{B}=\langle \beta_1\rangle\oplus\cdots\oplus\langle \beta_\ell\rangle$. We will say that these new generators are \emph{linearly independent}.
\item[(b)] An integer vector $c$ containing the orders of the linearly independent generators $\beta_i$.
\item[(c)]  Two integer matrices $A$, $B$ that relate the old and new generators as follows:
\begin{equation}\label{eq:Matrix of Relationships}
\begin{pmatrix}
 \beta_1,\ldots,\beta_\ell
\end{pmatrix}=
\begin{pmatrix}
\alpha_1, 
\ldots
\alpha_k
\end{pmatrix} A,\qquad\quad \begin{pmatrix}
\alpha_1, 
\ldots
\alpha_k
\end{pmatrix}= \begin{pmatrix}
 \beta_1,\ldots,\beta_\ell
\end{pmatrix}B.
\end{equation}
This last equation should be read in multiplicative notation (as in e.g.\ \cite{Cohen00_Advanced_Topics_ComputationalNumber_Theory}), where ``vectors'' of group elements are right-multiplied by matrices as follows: given  the $i$th column $a_i$  of $A$ (for the left hand case), we have $\beta_{i}=(\alpha_1,\ldots,\alpha_k) a_i=\alpha_1^{a_i(1)}\cdots \alpha_k^{a_i(k)}$.
\end{enumerate}
Our definition of the group decomposition is more general than the one given in \cite{mosca_phd,cheung_mosca_01_decomp_abelian_groups}. In Cheung and Mosca's formulation, the task is to find just $\beta$ and $c$. The algorithm they give also computes the matrix $A$ in order to find the generators $\beta_i$ (cf.\ the next section). What is completely new in our formulation is that we ask in addition for the matrix $B$.

Note that a {\textbf{group-decomposition table}} $(\alpha,\beta, A, B, c)$ contains a lot of information about the group structure of  $\mathbf{B}$. First of all, the tuple elements (a-b)  tell us  that  $\mathbf{B}$ is isomorphic to a decomposed group $G=\DProd{c}{k}$. In addition, the matrices $A$ and $B$ provide us with an efficient method to re-write linear combinations of the original generators $\alpha_i$ as linear combinations of the new generators $\beta_j$ (and vice-versa). Indeed, equation (\ref{eq:Matrix of Relationships})  implies
\begin{align}\alpha_{1}^{x_1}\cdots \alpha_{k}^{x_k}&=  \begin{pmatrix}
\alpha_1, 
\ldots
\alpha_k
\end{pmatrix}x= \begin{pmatrix}
 \beta_1,\ldots,\beta_\ell
\end{pmatrix}(Bx) ,\quad\textnormal{  for any  $x\in\Z^k$,}\notag\\
\beta_{1}^{y_1}\cdots \beta_{1}^{y_\ell} &= \begin{pmatrix}
 \beta_1,\ldots,\beta_\ell
\end{pmatrix}y=
\begin{pmatrix}
\alpha_1, 
\ldots
\alpha_k
\end{pmatrix} (Ay),\quad\textnormal{  for any $y\in\Z^\ell$}.\notag
\end{align}
It follows that, for any given $x$, the integer string $y=Bx$ (which can be efficiently computed classically)  fulfills the condition $\alpha_{1}^{x_1}\cdots \alpha_{1}^{x_k}=\beta_{1}^{y_1}\cdots \beta_{1}^{y_\ell}$. (A symmetric argument proves the opposite direction.)

As we  discussed earlier in the introduction, the group decomposition problem is \emph{provably hard} for classical computers within the black-box setting, and it is at least \emph{as hard as} Factoring (or Order Finding) for matrix groups of the form $\Z_N^\times$ (the latter being polynomial-time reducible to group decomposition).\ It can be also shown that group decomposition is also at least as hard as computing discrete logarithms, a fact that we will use in the proof of theorems \ref{thm:Simulation}, \ref{thm:No Go Theorem}:
\begin{lemma}[\textbf{Multivariate discrete logarithms}]\label{lemma:Multivariate Discrete Log} Let $\beta_1,\ldots,\beta_\ell$ be  generators of some Abelian black-box group $\mathbf{B}$ with the  property
$\mathbf{B}=\langle \beta_1\rangle\oplus\cdots\oplus\langle \beta_\ell\rangle$. Then, the following generalized version of the discrete-logarithm problem is polynomial time reducible to group decomposition: for a given $\beta\in\mathbf{B}$, find an integer string $x$ such that $\beta_{1}^{x_1}\cdots \beta_{\ell}^{x_\ell}=\beta$. 
\end{lemma}
\begin{proof}
Define a new set of generators  for $\mathbf{B}$  by  adding the element $\beta_{\ell+1}=\beta$ to the given set $\{\beta_i\}$. The array $\alpha':=(\beta_1,\ldots,\beta_\ell+1)$ defines an instance of Group Decomposition. Assume that a group decomposition table $(\alpha', (\beta'_1,\ldots,\beta_m'), A', B', c')$ for this instance of the problem is given to us. We can now use the columns $b_i'$ of the matrix $B'$ to re-write the previous generators $\beta_i$ in terms of the new ones:
\begin{equation}\label{inproof:MultiVariate Discrete Log}
\beta_i = (\beta_1,\ldots,\beta_{\ell+1}) e_i = \begin{pmatrix}
 \beta_1',\ldots,\beta_m'
\end{pmatrix} (B'e_i) =  \begin{pmatrix}
 \beta_1',\ldots,\beta_m'
\end{pmatrix} b_i'=
 \beta_1'^{b_i'(1)}\cdots\beta_m'^{b_i'(m)}.
\end{equation}
Here, $e_i$ denotes the integer vector with $e(i)=1$ and $e(j)=0$ elsewhere. Conditions (a-b) imply that the columns $b_i'$ can be treated as elements of the group $G=\DProd{c'}{m}$. Using this identification, the original discrete logarithm problem reduces to finding an integer string $x\in G$ such that $b_{\ell+1}'=(b_1',\ldots,b_\ell')x=\sum x(i)b_i'$ (now in additive notation). The existence of such an $x$ can be easily proven using  that the elements $\beta_1,\ldots,\beta_\ell$ generate $\mathbf{B}$: the latter guarantees the existence of an $x$ such that
\begin{equation}
\beta_{\ell+1}=(\beta_1,\ldots,\beta_{\ell}) x = \begin{pmatrix}
 \beta_1',\ldots,\beta_m'
\end{pmatrix} (b_1',\ldots,b_\ell') x =   \begin{pmatrix}
 \beta_1',\ldots,\beta_m'
\end{pmatrix} b_{\ell+1}',
\end{equation}
which implies $(b_1',\ldots,b_\ell') x \equiv   \begin{pmatrix}
 \beta_1',\ldots,\beta_m'
\end{pmatrix} b_{\ell+1}' \bmod{(c_1',\ldots,c_m')}$.
By finding such an  $x$, we can  solve the multivariate discrete problem, since $\beta_{1}^{x_1}\cdots \beta_{\ell}^{x_\ell}=\beta_1'^{b_{\ell+1}'(1)}\cdots\beta_m'^{b_{\ell+1}'(m)}=\beta_{\ell+1}=\beta$, due to    (\ref{inproof:MultiVariate Discrete Log}). Finally, note that we can find $x$  efficiently with existing deterministic classical algorithms for Group Membership in finite Abelian groups (cf.\ lemma 3 in \cite{BermejoVega_12_GKTheorem}).
\end{proof}
We highlight that, in order for the latter result to hold, it seems critical to use our formulation of group decomposition instead of Cheung-Mosca's.  Consider again the discrete-log problem over the group $\Z_p^\times$ (recall section \ref{sect:Discrete Log}). This group  $\Z_p^\times$ is cyclic of order $p-1$ and a generating element $a$ is given to us as part of the input of the discrete-log problem. Although  it is not known how to solve this problem efficiently, Cheung-Mosca's group decomposition problem (find some linearly independent generators and their orders) can be solved effortlessly in this case, by simply returning $a$ and $p-1$, since  $\langle a \rangle = \Z_p^\times \cong \Z_{p-1}$. The crucial difference is that Cheung-Mosca's algorithm returns a factorization $ \DProd{c}{\ell}$ of $\mathbf{B}$, but it cannot be used to convert elements between the two representations efficiently (one direction is easy; the other requires computing discrete logarithms). In our formulation, the matrices $A$, $B$ provide such a method.

\subsubsection*{Quantum algorithm for group decomposition}

We now present a quantum algorithm that solves the group decomposition problem.

\begin{algorithm}[\textbf{Extended Cheung-Mosca's algorithm}]
\label{alg:group decomposition}

\begin{alg_in} A list of generators $\alpha=(\alpha_1,\ldots,\alpha_k)$ of an Abelian black-box group $\mathbf{B}$\footnote{Cheung and Mosca's original presentation first used Shor's algorithm to decompose $B$ into Sylow $p$-subgroups, and then performed the group decomposition on these subgroups. As they commented, that step was not strictly necessary, although it would reduce the computational resources required. Also, their algorithm does not receive a list of generators as an input, but this is not a big difference since such a list can always be computed in probabilistic classical polynomial-time for a uniquely encoded black-box group.}.

\end{alg_in}

\begin{alg_out} A group decomposition table $(\alpha, \beta, A, B, c)$.

\end{alg_out}

\begin{enumerate}

\item  Use the order finding algorithm (comprising normalizer circuit over $\Z\times \mathbf{B}$ and classical postprocessing) to obtain the orders $d_i$ of the generators $\alpha_i$. Then, compute (classically) and store their least common multiplier $d=\mathrm{lcm}(d_1,\ldots,d_k)$.

\item  Define the function $f : \Z_{d}^k \rightarrow \mathbf{B}$ as $f(x) = \alpha_1^{x(1)} \cdots \alpha_k^{x(k)}$, which is a group homomorphism and hides the subgroup $\ker f$ (its own kernel). Apply the Abelian HSP algorithm to compute a set  of generators $h_1,\ldots,h_m$ of $\ker f$.  This round uses normalizer circuits over $\Z_d^k\times \mathbf{B}$ and classical post-processing (cf.\ section \ref{sect:Abelian HSPs}).

\item Given the generators $h_i$ of $\ker f$ one can classically compute a $k\times \ell$ matrix $A$ (for some $\ell$) such that $(\beta_1,\ldots,\beta_\ell)=(\alpha_1,\ldots,\alpha_k) A$ is a system of linearly independent generators   \cite[theorem 7]{cheung_mosca_01_decomp_abelian_groups}.  $\beta$, $A$ and the orders $c_i$ of the $\beta_i$s (computed again via an order-finding subroutine) will form part of the output.
\item Finally, we show how to classically compute a valid relationship matrix $B$. (This step is not part of Cheung-Mosca's original algorithm.) The problem reduces to finding a $k\times \ell $ integer matrix $X$ with two properties: 
\begin{itemize}
\item[(a)] $X$ is a solution to the equation $(\alpha_1,\ldots,\alpha_k)X=(\alpha_1,\ldots,\alpha_k)$. Equivalently, every column $x_i$ of $X$ is equal (modulo $d$) to some element of the coset $e_i+\ker{f}\subset\Z_d^k$.
\item[(b)] Every column $x_i$ is an element of the image of the matrix $A$.
\end{itemize}
It is easy to see  that a matrix $X$ fulfilling (a-b) always exists, since for any $\alpha_i$, there exists some $y_i$ such that $\alpha_i=(\beta_1,\ldots,\beta_\ell)y_i$ (because the $\beta_i$s generate the group). It follows that $\alpha_i=(\alpha_1,\ldots,\alpha_k) (Ay_i)$. Then, the matrix with columns $x_i=Ay_i$ has the desired properties.

$\quad$ Our existence proof for $X$ is constructive, and tells us that  $X$ can be computed in quantum polynomial time by solving a multivariate discrete logarithm problem (lemma \ref{lemma:Multivariate Discrete Log}). However,  we will use a more subtle efficient \emph{classical} approach to obtain $X$, by reducing the problem to a {\textbf{system of linear equations over Abelian groups}} \cite{BermejoVega_12_GKTheorem,BermejoLinVdN13_Infinite_Normalizers}. Let $H$ be a matrix formed column-wise by the generators $h_i$ of $\ker f$. By construction, the image of the map $H:\Z_d^m \rightarrow \Z_d^k$ fulfills $\mathrm{im} H =\ker{f}$. Properties (a-b) imply that the $i$th column $x_i$ of $X$ must be a particular solution to the equations $x_i= A y_i$ with $y_i\in \Z^\ell$ and $x_i = e_i + Hz_i \bmod{d}$, with $z_i\in \Z_d^m$. These equations can be equivalently written as a system of linear equations over $\Z^{m+\ell}$:
\begin{equation}
\begin{pmatrix}
A & -H
\end{pmatrix}\begin{pmatrix}
y_i \\ z_i
\end{pmatrix} = e_i\mod{d},\qquad (y_i,z_i)\in \Z_d^{m}\times \Z^\ell,
\end{equation}
which can be  solved in classical polynomial time using e.g.\ algorithms from \cite{BermejoLinVdN13_Infinite_Normalizers}. Then, the matrix $X$ can be constructed column wise taking $x_i=Ay_i$.

$\quad$Finally, given such an $X$, it is easy to find a valid  $B$ by computing a Hurt-Waid integral pseudo-inverse $A^\#$ of $A$ \cite{HurtWaid70_Integral_Generalized_Inverse,BowmanBurget74_systems-Mixed-Integer_Linear_equations}:
\begin{equation}
\alpha= \alpha X = \alpha (A A^\#) X = (\alpha A ) (A^\#X) =(\beta_1,\ldots,\beta_\ell) (A^\# X).
\end{equation}
In the third step, we used that $A^\#$ acts as the inverse of $A$ on inputs $x\in \Z^k$ that live in the image of $A$ \cite{HurtWaid70_Integral_Generalized_Inverse}. Since integral pseudo-inverses can be computed efficiently using the Smith normal form (see e.g.\ our dicussion in \cite[appendix~D]{BermejoLinVdN13_Infinite_Normalizers}), we finally set $B:=A^\# X$. 
\end{enumerate}
\end{algorithm}

\section{Simulation of  black-box normalizer circuits}\label{sect:Simulation}

Our results so far show that the computational power of normalizer circuits over black-box groups (supplemented with classical pre- and post- processing) is \emph{strikingly high}: they can solve several problems believed to be classically intractable and render the RSA, Diffie-Hellman, and elliptic curve public-key cryptosystems vulnerable. In contrast,  standard normalizer circuits, which are associated with Abelian groups that are \emph{explicitly decomposed}, can be efficiently simulated classically, by exploiting a generalized stabilizer formalism \cite{VDNest_12_QFTs,BermejoVega_12_GKTheorem,BermejoLinVdN13_Infinite_Normalizers} over Abelian groups.

It is natural to wonder at this point where the computational  power of black-box normalizer circuits originates. In this section, we will argue that the hardness of simulating black-box normalizer circuits resides \emph{precisely} in the hardness of decomposing black-box Abelian groups. An equivalence is suggested by the fact that we can use these circuits to solve the group decomposition problem and, in turn, when the group is decomposed, the techniques in \cite{VDNest_12_QFTs,BermejoVega_12_GKTheorem,BermejoLinVdN13_Infinite_Normalizers} render these circuits classically simulable. In this sense, then, the \emph{quantum speedup }of such circuits appears to be completely encapsulated in the group decomposition algorithm. This intuition can be made precise and be stated as a theorem.
\begin{theorem}[\textbf{Simulation of black-box normalizer circuits}]\label{thm:Simulation} Black-box normalizer circuits can be efficiently simulated classically using the stabilizer formalism over Abelian groups  \cite{VDNest_12_QFTs,BermejoVega_12_GKTheorem,BermejoLinVdN13_Infinite_Normalizers} if a subroutine for solving the group-decomposition problem is provided as an oracle.
\end{theorem}
The proof of this theorem is the subject of section \ref{app:Simulation proof} in the appendix.

Since normalizer circuits can solve the group decomposition problem (section \ref{sect:Group Decomposition}), we obtain that this problem is complete for the associated normalizer-circuit complexity class, which we now define.
\begin{definition}[\textbf{Black-Box Normalizer}] The complexity class \textbf{Black-Box Normalizer} is the set of oracle problems that can be solved with bounded error by at most polynomially many rounds of efficient black-box normalizer circuits (as defined in section \ref{sect:circuit model}), with polynomial-sized classical computation interspersed between. In other words, if $N$ is an oracle that given an efficient (poly-sized) black-box normalizer circuit as input, samples from its output distribution, then
\begin{equation}
\text{\textbf{Black-Box Normalizer}} = BPP^{N}.
\end{equation}
\end{definition}

\begin{corollary}[\textbf{Completeness of group decomposition}]\label{corollary:Group Decomposition is complete}
Group decomposition is a complete problem for the complexity class \textnormal{\textbf{Black-Box Normalizer}} under classical polynomial-time Turing reductions.
\end{corollary}
We stress that theorem \ref{thm:Simulation}  tells us even more than the completeness of group decomposition. As we discussed in the introduction, an oracle for group decomposition gives us an efficient classical algorithm to simulate Shor's factoring and discrete-log algorithm (and all the others) \emph{step-by-step} with a stabilizer-picture approach ``à la Gottesman-Knill''.

We also highlight that theorem \ref{thm:Simulation} can be restated as a \emph{no-go theorem} for finding new quantum algorithms based on black-box normalizer circuits.
\begin{theorem}[\textbf{No-go theorem for new quantum algorithms}]\label{thm:No Go Theorem} It is not possible to find ``fundamentally new'' quantum algorithms within the class of black-box normalizer circuits studied in this work, in the sense that any new algorithm would be efficiently simulable  using the extended Cheung-Mosca algorithm and classical post-processing.
\end{theorem}
This theorem tells us that  black-box normalizer circuits cannot give exponential speedups over classical circuits that are not already covered by the extended Cheung-Mosca algorithm; the theorem may thus have applications to algorithm design. 

Note, however, that this no-go theorem says nothing about other possible \emph{polynomial} speed-ups for black-box normalizer circuits; there may well be other normalizer circuits that are polynomially faster, conceptually simpler, or easier to implement than the extended Cheung-Mosca algorithm. Our theorem neither denies that investigating black-box normalizer could be of pedagogical or practical value if, e.g., this led  to new interesting complete problems for the class \textbf{Black-Box Normalizer}.

Finally, we note that theorem \ref{thm:Simulation}  can be extended to  the general Abelian hidden subgroup problem to show that the quantum algorithm for the Abelian HSP becomes efficiently classically simulable if an algorithm for decomposing the oracular group $\mathcal{O}$ is given to us (cf.\ section \ref{sect:Abelian HSPs} and refer to appendix \ref{app:Extending} for a proof). We discuss some implications of this fact in the next sections.

\section{Universality of short quantum circuits}\label{sect:Universality}

Since all problems in \textbf{Black-Box Normalizer} are solvable by the extended Cheung-Mosca quantum algorithm (supplemented with classical processing), the structure of said quantum algorithm allows us to state the following:
\begin{theorem}[\textbf{Universality of short normalizer circuits}]\label{thm:Universality Short Circuits}
Any problem in the class \textbf{\emph{Black-Box Normalizer}} can be solved by a quantum algorithm composed of polynomially-many rounds of \emph{short} normalizer circuits, each with at most a \textbf{constant} number of normalizer gates, and additional classical computation. More precisely, in every round, normalizer circuits containing two quantum Fourier transforms and one automorphism gate (and no quadratic phase gate) are already sufficient.
\end{theorem}
\begin{proof}
This result follows immediately form the fact that group decomposition is complete for this class (theorem \ref{corollary:Group Decomposition is complete}) and from the structure of the extended Cheung-Mosca quantum algorithm with this problem, which has precisely this structure. 
\end{proof}
Similarly to theorem \ref{thm:Simulation}, theorem \ref{thm:Universality Short Circuits} can be extended to the general Abelian HSP  setting. For details, we refer the reader to appendix \ref{app:Extending}.

We find the latter result is insightful, in that it actually explains a somewhat intriguing feature present in Deutsch's, Simon's, Shor's and virtually all known quantum algorithms for solving Abelian hidden subgroup problems: they all contain at most two quantum Fourier transforms! Clearly, it follows from this theorem than no more than two are enough. 

Also, theorem \ref{thm:Universality Short Circuits} tells us that it is actually pretty \emph{useless} to use logarithmically or polynomially long  sequences of quantum Fourier transforms for solving Abelian hidden subgroup problems, since just two of them suffice\footnote{This last comment does not imply that building up sequences of Fourier transforms is useless in general. On the contrary, this can be actually be useful, e.g., in QMA amplification  \cite{Nagaj2009_Fast_Amplification_QMA}.}. In this sense, the Abelian HSP quantum algorithm uses an \emph{asymptotically optimal} number of quantum Fourier transforms. Furthermore, the normalizer-gate depth of this algorithm is optimal in general.
 
\section{Other Complete problems}\label{sect:Complete Problems}

We end this paper by giving two other complete problems for the complexity class \textbf{Black Box Normalizer}.
\begin{theorem}[\textbf{Hidden kernel problem is complete}]\label{thm:Hidden Kernel Problem} Let the Abelian hidden kernel problem (Abelian HKP) be the subcase of the hidden subgroup problem where the oracle function $f$ is a group homomorphism from a group of the form  $G=\Z^a\times\DProd{N}{b}$ into a black-box group $\mathbf{B}$. This problem is complete for \textbf{\emph{Black Box Normalizer}} under polynomial-time Turing reductions.
\end{theorem}
\begin{proof}
Clearly group decomposition reduces to this problem, since the quantum steps of the extended Cheung-Mosca algorithm algorithm (steps 1 and 3) are solving instances of the Abelian kernel problem. Therefore, the Abelian HKP problem is hard for \textbf{Black Box Normalizer}.

Moreover, Abelian HKP can be solved with the general Abelian HSP quantum algorithm, which manifestly becomes a black-box normalizer circuit for oracle functions $f$ that are group homomorphisms onto black-box groups. This implies that Abelian HKP is inside \textbf{Black Box Normalizer}, and therefore, it is complete.

\textbf{Note.} Although we originally stated the Abelian HSP for finite groups, one can first apply the order-finding algorithm to compute a multiple $d$ of the orders of the elements $f(e_i)$, where $e_i$ are the canonical generators of $G$. This can be used to reduce the original HKP problem to a simplified HKP over the group $\Z_{d}^a\times \DProd{N}{b}$
\end{proof}
The latter result can be extended to the HSP setting to show that the Abelian hidden subgroup problem is polynomial-time equivalent to decomposing groups of the form $\mathcal{O}$ (cf.\ appendix \ref{app:Extending}).
\begin{theorem}[\textbf{System of linear equations over groups}] Let $\alpha$ be a group homomorphism from a group $G=\Z^a\times\DProd{N}{b}$ onto a black-box group $\mathbf{B}$. An instance of a linear system of equations over $G$ and $\mathbf{B}$ \footnote{These systems where studied extensively in our previous works \cite{BermejoVega_12_GKTheorem,BermejoLinVdN13_Infinite_Normalizers} in the decomposed-group setting.}  is given by $\alpha$ and an element $\mathbf{b}\in \mathbf{B}$. Our task is to find a general $(x_0, K)$ solution to the equation
\begin{equation}\nonumber
\alpha(x)=\mathbf{b},\quad x\in G,
\end{equation}
where $x_0$ is any particular solution and  $K$ is a generating set of the kernel of $\alpha$. This problem is  complete for \textbf{\emph{Black Box Normalizer}} under polynomial-time Turing reductions.
\end{theorem}
\begin{proof}
Clearly, this problem is hard for our complexity class, since the Abelian hidden kernel problem reduces to finding $K$.

Moreover, this problem can be solved with black-box normalizer circuits and classical computation, proving its completeness. First, we find a decomposition $\mathbf{B}=\bigoplus\langle \beta_i \rangle \cong H=\DProd{c}{\ell}$ with black-box normalizer circuits. Second, we recycle the de-black-boxing idea from the proof of theorem \ref{thm:Simulation} to compute a matrix representation of $\alpha$, and   solve the multivariate discrete logarithm problem $\mathbf{b}=\beta_1^{b(1)}\cdots\beta_\ell^{b(\ell)}$, $b\in H$, either with black-box normalizer circuits or classically (recall section \ref{sect:Group Decomposition}). The original system of equations can now be equivalently written as $A x = b \pmod{H}$. A general solution of this system can be computed with classical algorithms given in \cite{BermejoLinVdN13_Infinite_Normalizers}.
\end{proof}

\section{Acknowledgments}

We are grateful to Mykhaylo (Mischa) Panchenko, Uri Vool, Pawel Wocjan,  Raúl García-Patrón, Geza Giedke, Mari Carmen Bañuls, Liang Jiang, Steven M.\ Girvin, Barbara M.\ Terhal, Hussain Anwar,  Tobias J.\ Osborne and Kevin C.\ Zatloukal for encouraging discussions; and to Kevin C.\ Zatloukal for comments on the manuscript.  JBV thanks the MIT Center for Theoretical Physics, where part of this work was conducted, for their warm and generous hospitality in Spring 2013. 

CYL acknowledges support from the ARO grant W911NF-12-0486 (Quantum Information Group), and the Natural Sciences and Engineering Research Council of Canada. JBV acknowledges  funding from the Elite Network of Bavaria QCCC Program. This paper is preprint number MIT-CTP/4584.

\small 
\phantomsection
\addcontentsline{toc}{section}{References}
\bibliographystyle{utphys}
\bibliography{database}

\providecommand{\href}[2]{#2}\begingroup\raggedright\begin{thebibliography}{100}

\bibitem{VDNest_12_QFTs}
M.~Van~den Nest, ``Efficient classical simulations of quantum {F}ourier
  transforms and normalizer circuits over {A}belian groups,'' {\em Quantum
  Information and Computation} {\bfseries 0} no.~1, (2012) ,
  \href{http://arxiv.org/abs/1201.4867v1}{{\ttfamily arXiv:1201.4867v1
  [quant-ph]}}.

\bibitem{BermejoVega_12_GKTheorem}
J.~Bermejo-Vega and M.~Van Den~Nest, ``Classical simulations of {A}belian-group
  normalizer circuits with intermediate measurements,'' {\em Quantum Info.
  Comput.} {\bfseries 14} no.~3-4, (2014) ,
  \href{http://arxiv.org/abs/1210.3637}{{\ttfamily arXiv:1210.3637
  [quant-ph]}}.

\bibitem{BermejoLinVdN13_Infinite_Normalizers}
J.~Bermejo-Vega, C.~Y.-Y. Lin, and M.~Van~den Nest, ``Normalizer circuits and a
  {G}ottesman-{K}nill theorem for infinite-dimensional systems,''
  \href{http://arxiv.org/abs/1409.3208}{{\ttfamily arXiv:1409.3208
  [quant-ph]}}.

\bibitem{BabaiSzmeredi_Complexity_MatrixGroup_Problems_I}
L.~Babai and E.~Szemer\'{e}di,
  \href{http://dx.doi.org/10.1109/SFCS.1984.715919}{``On the complexity of
  matrix group problems {I},''} in {\em Proceedings of the 25th Annual
  Symposium onFoundations of Computer Science, 1984}, SFCS '84.
\newblock IEEE Computer Society, 1984.

\bibitem{Mosca09_Quantum_Algorithms_REVIEW}
M.~Mosca, \href{http://dx.doi.org/10.1007/978-0-387-30440-3_423}{``Quantum
  algorithms,''} in {\em Encyclopedia of Complexity and Systems Science}.
\newblock Springer, 2009.
\newblock \href{http://arxiv.org/abs/0808.0369}{{\ttfamily arXiv:0808.0369
  [quant-ph]}}.

\bibitem{childs_vandam_10_qu_algorithms_algebraic_problems}
A.~M. Childs and W.~van Dam, ``Quantum algorithms for algebraic problems,''
  \href{http://dx.doi.org/10.1103/RevModPhys.82.1}{{\em Rev. Mod. Phys.}
  {\bfseries 82} (2010) }, \href{http://arxiv.org/abs/0812.0380v1}{{\ttfamily
  arXiv:0812.0380v1 [quant-ph]}}.

\bibitem{Bacon10_Recent_Progres_Quantum_Algorithms}
D.~Bacon and W.~van Dam, ``Recent progress in quantum algorithms,''
  \href{http://dx.doi.org/10.1145/1646353.1646375}{{\em Commun. ACM} {\bfseries
  53} no.~2, (2010) }.

\bibitem{VanDamSasaki12_Q_algorithms_number_theory_REVIEW}
W.~van Dam and Y.~Sasaki, {\em {Q}uantum algorithms for problems in number
  theory, algebraic geometry, and group theory}.
\newblock World Scientific, 2012.
\newblock \href{http://arxiv.org/abs/1206.6126}{{\ttfamily arXiv:1206.6126
  [quant-ph]}}.

\bibitem{MoscaSmith12_Algorithms_Quantum_Computers_REVIEW}
J.~Smith and M.~Mosca,
  \href{http://dx.doi.org/10.1007/978-3-540-92910-9_43}{``Algorithms for
  quantum computers,''} in {\em Handbook of Natural Computing}.
\newblock Springer, 2012.

\bibitem{Jordan_Quantum_Algorithm_Zoo}
S.~Jordan, ``{Q}uantum {A}lgorithm {Z}oo.''
\newblock \url{http://math.nist.gov/quantum/zoo/}.

\bibitem{Shor_04_Progress_Quantum_Algorithms}
P.~W. Shor, ``Progress in quantum algorithms,''
  \href{http://dx.doi.org/10.1007/s11128-004-3878-2}{{\em Quantum Information
  Processing} {\bfseries 3} no.~1-5, (2004) }.

\bibitem{Gottesman_PhD_Thesis}
D.~Gottesman, {\em Stabilizer Codes and Quantum Error Correction}.
\newblock PhD thesis, California Institute of Technology, 1997.
\newblock \href{http://arxiv.org/abs/quant-ph/9705052v1}{{\ttfamily
  quant-ph/9705052v1}}.

\bibitem{Gottesman99_HeisenbergRepresentation_of_Q_Computers}
D.~Gottesman, ``The {H}eisenberg representation of quantum computers,'' in {\em
  Group22: Proceedings of the XXII International Colloquium on Group
  Theoretical Methods in Physics}.
\newblock International Press, 1999.
\newblock \href{http://arxiv.org/abs/quant-ph/9807006v1}{{\ttfamily
  quant-ph/9807006v1}}.

\bibitem{Knill96non-binaryunitary}
E.~Knill, ``Non-binary unitary error bases and quantum codes,'' tech. rep., Los
  Alamos National Laboratory, 1996.
\newblock \href{http://arxiv.org/abs/quant-ph/9608048}{{\ttfamily
  quant-ph/9608048}}.

\bibitem{Gottesman98Fault_Tolerant_QC_HigherDimensions}
D.~Gottesman, ``Fault-tolerant quantum computation with higher-dimensional
  systems,'' in {\em Selected papers from the First NASA International
  Conference on Quantum Computing and Quantum Communications}.
\newblock Springer, 1998.
\newblock \href{http://arxiv.org/abs/quant-ph/9802007v1}{{\ttfamily
  quant-ph/9802007v1}}.

\bibitem{Valiant02_matchgates}
L.~G. Valiant, ``Quantum circuits that can be simulated classically in
  polynomial time,'' \href{http://dx.doi.org/10.1137/S0097539700377025}{{\em
  SIAM J. Comput.} {\bfseries 31} no.~4, (2002) }.

\bibitem{Knill01_Fermionic_Linear_Optics}
E.~Knill, ``{F}ermionic {L}inear {O}ptics and {M}atchgates,'' 2001.
\newblock arXiv:quant-ph/0108033.

\bibitem{Terhal02_Simulation_noninteracting_fermion_circuits}
B.~M. Terhal and D.~P. DiVincenzo, ``Classical simulation of
  noninteracting-fermion quantum circuits,''
  \href{http://dx.doi.org/10.1103/PhysRevA.65.032325}{{\em Phys. Rev. A}
  {\bfseries 65} (2002) },
  \href{http://arxiv.org/abs/quant-ph/0108010}{{\ttfamily quant-ph/0108010}}.

\bibitem{Jozsa08_Matchgates_classical_simulation}
R.~Jozsa and A.~Miyake, ``Matchgates and classical simulation of quantum
  circuits,'' \href{http://dx.doi.org/10.1098/rspa.2008.0189}{{\em Proceedings
  of the Royal Society A: Mathematical, Physical and Engineering Science}
  {\bfseries 464} no.~2100, (2008) },
  \href{http://arxiv.org/abs/0804.4050}{{\ttfamily arXiv:0804.4050
  [quant-ph]}}.

\bibitem{LloydBraunstein99_QC_over_CVs}
S.~Lloyd and S.~L. Braunstein, ``Quantum computation over continuous
  variables,'' \href{http://dx.doi.org/10.1103/PhysRevLett.82.1784}{{\em Phys.
  Rev. Lett.} {\bfseries 82} (1999) },
  \href{http://arxiv.org/abs/quant-ph/9810082}{{\ttfamily quant-ph/9810082}}.

\bibitem{Bartlett02Continuous-Variable-GK-Theorem}
S.~D. Bartlett, B.~C. Sanders, S.~L. Braunstein, and K.~Nemoto, ``Efficient
  classical simulation of continuous variable quantum information processes,''
  \href{http://dx.doi.org/10.1103/PhysRevLett.88.097904}{{\em Phys. Rev. Lett.}
  {\bfseries 88} (2002) },
  \href{http://arxiv.org/abs/quant-ph/0109047}{{\ttfamily quant-ph/0109047}}.

\bibitem{BartlettSanders02Simulations_Optical_QI_Circuits}
S.~D. Bartlett and B.~C. Sanders, ``Efficient classical simulation of optical
  quantum information circuits,''
  \href{http://dx.doi.org/10.1103/PhysRevLett.89.207903}{{\em Phys. Rev. Lett.}
  {\bfseries 89} (2002) },
  \href{http://arxiv.org/abs/quant-ph/0204065}{{\ttfamily quant-ph/0204065}}.

\bibitem{KnillLaflamme98_DQC1}
E.~Knill and R.~Laflamme, ``Power of one bit of quantum information,''
  \href{http://dx.doi.org/10.1103/PhysRevLett.81.5672}{{\em Phys. Rev. Lett.}
  {\bfseries 81} (1998) },
  \href{http://arxiv.org/abs/quant-ph/9802037}{{\ttfamily quant-ph/9802037}}.

\bibitem{Shepherd10_PhD_thesis}
D.~J. Shepherd, {\em {Q}uantum {C}omplexity: {r}estrictions on {a}lgorithms and
  {a}rchitectures}.
\newblock PhD thesis, 2010.
\newblock \href{http://arxiv.org/abs/1005.1425}{{\ttfamily 1005.1425
  [quant-ph]}}.

\bibitem{ShepherdBremner09_Temporally_Unstructured_QC}
D.~Shepherd and M.~J. Bremner, ``Temporally unstructured quantum computation,''
  \href{http://dx.doi.org/10.1098/rspa.2008.0443}{{\em Proceedings of the Royal
  Society A: Mathematical, Physical and Engineering Science} (2009) },
  \href{http://arxiv.org/abs/0809.0847}{{\ttfamily 0809.0847 [quant-ph]}}.

\bibitem{BremnerJozsaShepherd08}
M.~J. Bremner, R.~Jozsa, and D.~J. Shepherd, ``Classical simulation of
  commuting quantum computations implies collapse of the polynomial
  hierarchy,'' \href{http://dx.doi.org/10.1098/rspa.2010.0301}{{\em Proceedings
  of the Royal Society A: Mathematical, Physical and Engineering Science}
  {\bfseries 467} (2011) }, \href{http://arxiv.org/abs/1005.1407}{{\ttfamily
  arXiv:1005.1407 [quant-ph]}}.

\bibitem{Ni13Commuting_Circuits}
X.~Ni and M.~Van~den Nest, ``Commuting quantum circuits: efficiently classical
  simulations versus hardness results,'' {\em Quantum Information and
  Computation} {\bfseries 13} no.~1\&2, (2013) ,
  \href{http://arxiv.org/abs/1204.4570}{{\ttfamily arXiv:1204.4570
  [quant-ph]}}.

\bibitem{Shor}
P.~W. Shor, ``Polynomial-time algorithms for prime factorization and discrete
  logarithms on a quantum computer,'' {\em SIAM J. Sci. Statist. Comput. 26}
  (1997) .

\bibitem{Humphrey96_Course_GroupTheory}
J.~F. Humphreys, {\em A course in group theory}.
\newblock Oxford University Press, 1996.

\bibitem{Babai98apolynomial_time_theory_of_Black_Box_Groups}
L.~Babai and R.~Beals, ``A polynomial-time theory of black-box groups {I},'' in
  {\em Groups St Andrews 1997 in Bath}, vol.~I of {\em London Mathematical
  Society Lecture Note Series}.
\newblock Cambridge University Press, 1999.

\bibitem{Shoup08_A_Computational_Introducttion_to_Number_Theory_and_Algebra}
V.~Shoup, {\em A Computational Introduction to Number Theory and Algebra}.
\newblock Cambridge University Press, 2nd~ed., 2008.

\bibitem{ProosZalka03_Shors_DiscreteLog_Elliptic_Curves}
J.~Proos and C.~Zalka, ``Shor's discrete logarithm quantum algorithm for
  elliptic curves,'' {\em Quantum Info. Comput.} {\bfseries 3} no.~4, (2003) ,
  \href{http://arxiv.org/abs/quant-ph/0301141}{{\ttfamily quant-ph/0301141}}.

\bibitem{Kaye05_optimized_Quantum_Elliptic_Curve}
P.~Kaye, ``Optimized quantum implementation of elliptic curve arithmetic over
  binary fields,'' {\em Quantum Info. Comput.} {\bfseries 5} no.~6, (2005) ,
  \href{http://arxiv.org/abs/quant-ph/0407095}{{\ttfamily quant-ph/0407095}}.

\bibitem{CheungMaslovMathew08_Design_QuantumAttack_Elliptic_CC}
D.~Cheung, D.~Maslov, J.~Mathew, and D.~Pradhan,
  \href{http://dx.doi.org/10.1007/978-3-540-89304-2_9}{``On the design and
  optimization of a quantum polynomial-time attack on elliptic curve
  cryptography,''} in {\em Theory of Quantum Computation, Communication, and
  Cryptography}, vol.~5106 of {\em Lecture Notes in Computer Science}.
\newblock Springer, 2008.
\newblock \href{http://arxiv.org/abs/0710.1093}{{\ttfamily arXiv:0710.1093
  [quant-ph]}}.

\bibitem{mosca_phd}
M.~Mosca, {\em Quantum computer algorithms}.
\newblock PhD thesis, University of Oxford, 1999.

\bibitem{cheung_mosca_01_decomp_abelian_groups}
K.~K.~H. Cheung and M.~Mosca, ``Decomposing finite {A}belian groups,'' {\em
  Quantum Info. Comput.} {\bfseries 1} no.~3, (2001) ,
  \href{http://arxiv.org/abs/cs/0101004}{{\ttfamily cs/0101004}}.

\bibitem{Deutsch85quantumtheory}
D.~Deutsch, ``Quantum theory, the {C}hurch-{T}uring principle and the universal
  quantum computer,'' \href{http://dx.doi.org/10.1098/rspa.1985.0070}{{\em
  Proceedings of the Royal Society of London. A. Mathematical and Physical
  Sciences} {\bfseries 400} no.~1818, (1985) }.

\bibitem{Simon94onthe}
D.~R. Simon, ``On the power of quantum computation,'' {\em SIAM Journal on
  Computing} {\bfseries 26} (1994) .

\bibitem{Boneh95QCryptanalysis}
D.~Boneh and R.~Lipton,
  \href{http://dx.doi.org/10.1007/3-540-44750-4_34}{``Quantum cryptanalysis of
  hidden linear functions,''} in {\em Advances in Cryptology --- CRYPT0' 95},
  D.~Coppersmith, ed., vol.~963 of {\em Lecture Notes in Computer Science}.
\newblock Springer, 1995.

\bibitem{Grigoriev97_testing_shift_equivalence_polynomials}
D.~Grigoriev, ``Testing shift-equivalence of polynomials by deterministic,
  probabilistic and quantum machines,''
  \href{http://dx.doi.org/10.1016/S0304-3975(96)00188-0}{{\em Theor. Comput.
  Sci.} {\bfseries 180} no.~1-2, (1997) }.

\bibitem{kitaev_phase_estimation}
A.~Y. Kitaev, ``Quantum measurements and the {A}belian stabilizer problem,''
  1995.
\newblock arXiv:quant-ph/9511026v1.

\bibitem{Kitaev97_QCs:_algorithms_error_correction}
A.~Y. Kitaev, ``Quantum computations: algorithms and error correction,'' {\em
  Russian Mathematical Surveys} {\bfseries 52} no.~6, (1997) .

\bibitem{Brassard_Hoyer97_Exact_Quantum_Algorithm_Simons_Problem}
G.~Brassard and P.~H\o{}yer, ``An exact quantum polynomial-time algorithm for
  {S}imon's problem,'' in {\em Proceedings of the Fifth Israel Symposium on the
  Theory of Computing Systems (ISTCS '97)}, ISTCS '97.
\newblock IEEE Computer Society, 1997.
\newblock \href{http://arxiv.org/abs/quant-ph/9704027}{{\ttfamily
  quant-ph/9704027}}.

\bibitem{Hoyer99Conjugated_operators}
P.~H\o{}yer, ``Conjugated operators in quantum algorithms,''
  \href{http://dx.doi.org/10.1103/PhysRevA.59.3280}{{\em Phys. Rev. A}
  {\bfseries 59} (1999) }.

\bibitem{MoscaEkert98_The_HSP_and_Eigenvalue_Estimation}
M.~Mosca and A.~Ekert, ``The hidden subgroup problem and eigenvalue estimation
  on a quantum computer,'' in {\em Selected Papers from the First NASA
  International Conference on Quantum Computing and Quantum Communications},
  QCQC '98.
\newblock Springer, 1998.
\newblock \href{http://arxiv.org/abs/quant-ph/9903071}{{\ttfamily
  quant-ph/9903071}}.

\bibitem{Damgard_QIP_note_HSP_algorithm}
I.~Damg\aa{}rd, ``{QIP} note: on the quantum {F}ourier transform and
  applications.''
\newblock \url{http://www.daimi.au.dk/~ivan/fourier.ps}.

\bibitem{DiffieHellman}
W.~Diffie and M.~Hellman, ``New directions in cryptography,''
  \href{http://dx.doi.org/10.1109/TIT.1976.1055638}{{\em Information Theory,
  IEEE Transactions on} {\bfseries 22} no.~6, (1976) }.

\bibitem{RSA}
R.~Rivest, A.~Shamir, and L.~Adleman, ``A method for obtaining digital
  signatures and public-key cryptosystems,'' {\em Communications of the ACM}
  {\bfseries 21} (1978) .

\bibitem{Menezes96_cryptography_book}
A.~J. Menezes, S.~A. Vanstone, and P.~C.~V. Oorschot, {\em Handbook of Applied
  Cryptography}.
\newblock CRC Press, 1st~ed., 1996.

\bibitem{Buchmann00_cryptography_book}
J.~A. Buchmann, {\em Introduction to Cryptography}.
\newblock Springer, 1st~ed., 2000.

\bibitem{dehaene_demoor_coefficients}
J.~Dehaene and B.~De~Moor, ``{C}lifford group, stabilizer states, and linear
  and quadratic operations over {GF}(2),''
  \href{http://dx.doi.org/10.1103/PhysRevA.68.042318}{{\em Phys. Rev. A}
  {\bfseries 68} (2003) },
  \href{http://arxiv.org/abs/quant-ph/0304125v1}{{\ttfamily
  quant-ph/0304125v1}}.

\bibitem{AaronsonGottesman04_Improved_Simul_stabilizer}
S.~Aaronson and D.~Gottesman, ``Improved simulation of stabilizer circuits,''
  \href{http://dx.doi.org/10.1103/PhysRevA.70.052328}{{\em Phys. Rev. A}
  {\bfseries 70} (2004) },
  \href{http://arxiv.org/abs/quant-ph/0406196}{{\ttfamily quant-ph/0406196}}.

\bibitem{dehaene_demoor_hostens}
E.~Hostens, J.~Dehaene, and B.~De~Moor, ``Stabilizer states and {C}lifford
  operations for systems of arbitrary dimensions and modular arithmetic,''
  \href{http://dx.doi.org/10.1103/PhysRevA.71.042315}{{\em Phys. Rev. A}
  {\bfseries 71} (2005) },
  \href{http://arxiv.org/abs/quant-ph/0408190v2}{{\ttfamily
  quant-ph/0408190v2}}.

\bibitem{AndersBriegel06_Simulation_Stabilizer_GraphStates}
S.~Anders and H.~J. Briegel, ``Fast simulation of stabilizer circuits using a
  graph-state representation,''
  \href{http://dx.doi.org/10.1103/PhysRevA.73.022334}{{\em Phys. Rev. A}
  {\bfseries 73} (2006) },
  \href{http://arxiv.org/abs/quant-ph/0504117}{{\ttfamily quant-ph/0504117}}.

\bibitem{VdNest10_Classical_Simulation_GKT_SlightlyBeyond}
M.~Van~den Nest, ``Classical simulation of quantum computation, the
  {G}ottesman-{K}nill theorem, and slightly beyond,'' {\em Quantum Info.
  Comput.} {\bfseries 10} no.~3, (2010) ,
  \href{http://arxiv.org/abs/0811.0898}{{\ttfamily arXiv:0811.0898
  [quant-ph]}}.

\bibitem{deBeaudrap12_linearised_stabiliser_formalism}
N.~de~Beaudrap, ``A linearized stabilizer formalism for systems of finite
  dimension,'' {\em Quantum Info. Comput.} {\bfseries 13} no.~1-2, (2013) ,
  \href{http://arxiv.org/abs/1102.3354v3}{{\ttfamily arXiv:1102.3354v3
  [quant-ph]}}.

\bibitem{JozsaVdNest14_Classical_Simulation_Extended_Clifford_Circuits}
R.~Jozsa and M.~Van Den~Nest, ``Classical simulation complexity of extended
  {C}lifford circuits,'' {\em Quantum Info. Comput.} {\bfseries 14} no.~7\&{}8,
  (2014) , \href{http://arxiv.org/abs/1305.6190}{{\ttfamily arXiv:1305.6190
  [quant-ph]}}.

\bibitem{Morris77_Pontryagin_Duality_and_LCA_groups}
S.~A. Morris, {\em {P}ontryagin Duality and the Structure of Locally Compact
  {A}belian Groups}.
\newblock Cambridge University Press, 1977.

\bibitem{Stroppel06_Locally_Compact_Groups}
M.~Stroppel, {\em Locally Compact Groups}.
\newblock EMS Textbooks in Mathematics. European Mathematical Society, 2006.

\bibitem{Dikranjan11_IntroTopologicalGroups}
D.~Dikranjan, ``Introduction to topological groups.'' 2010.

\bibitem{rudin62_Fourier_Analysis_on_groups}
W.~Rudin, {\em Fourier analysis on groups}.
\newblock No.~12 in Interscience Tracts in Pure and Applied Mathematics.
  Interscience Publishers, 1962.

\bibitem{HofmannMorris06The_Structure_of_Compact_Groups}
K.~H. Hofmann and S.~A. Morris, {\em The Structure of Compact Groups}.
\newblock No.~25 in de Gruyter Studies in Mathematics. Walter de Gruyter.

\bibitem{Armacost81_Structure_LCA_Groups}
D.~L. Armacost, {\em The structure of locally compact abelian groups}.
\newblock M. Dekker New York, 1981.

\bibitem{Baez08LCA_groups_Blog_Post}
J.~Baez, ``{T}he n-{C}ategory {C}af{\'e}: {L}ocally {C}ompact {H}ausdorff
  {A}belian {G}roups,'' 2008.
\newblock
  \url{http://golem.ph.utexas.edu/category/2008/11/locally_compact_hausdorff_abel.html}.

\bibitem{BravyiKitaev05MagicStateDistillation}
S.~Bravyi and A.~Kitaev, ``Universal quantum computation with ideal clifford
  gates and noisy ancillas,''
  \href{http://dx.doi.org/10.1103/PhysRevA.71.022316}{{\em Phys. Rev. A}
  {\bfseries 71} (2005) },
  \href{http://arxiv.org/abs/quant-ph/0403025}{{\ttfamily quant-ph/0403025}}.
  \url{http://link.aps.org/doi/10.1103/PhysRevA.71.022316}.

\bibitem{ClarkJozsaLinden08Generalized_Clifford_Groups}
S.~Clark, R.~Jozsa, and N.~Linden, ``Generalized {C}lifford groups and
  simulation of associated quantum circuits,'' {\em Quantum Info. Comput.}
  {\bfseries 8} no.~1, (2008) ,
  \href{http://arxiv.org/abs/quant-ph/0701103}{{\ttfamily quant-ph/0701103}}.

\bibitem{EttingerHoyerKnill2004_Hidden_Subgroup}
M.~Ettinger, P.~Hoyer, and E.~Knill, ``The quantum query complexity of the
  hidden subgroup problem is polynomial,'' {\em Information Processing Letters}
  {\bfseries 91} no.~1, (2004) ,
  \href{http://arxiv.org/abs/quant-ph/0401083}{{\ttfamily
  arXiv:quant-ph/0401083}}.

\bibitem{HallgrenRusselTaShma2003_Normal_Subgroup}
S.~Hallgren, A.~Russell, and A.~Ta-Shma, ``Normal subgroup reconstruction and
  quantum computation using group representations,'' {\em SIAM Journal on
  Computing} {\bfseries 32} no.~4, (2003) .

\bibitem{Kuperberg2005_Dihedral_Hidden_Subgroup}
G.~Kuperberg, ``A subexponential-time quantum algorithm for the dihedral hidden
  subgroup problem,'' {\em SIAM Journal on Computing} {\bfseries 35} no.~1,
  (2005) , \href{http://arxiv.org/abs/quant-ph/0302112}{{\ttfamily
  arXiv:quant-ph/0302112}}.

\bibitem{Regev2004_Dihedral_Hidden_Subgroup}
O.~Regev, ``A subexponential time algorithm for the dihedral hidden subgroup
  problem with polynomial space,''
  \href{http://arxiv.org/abs/quant-ph/0406151}{{\ttfamily
  arXiv:quant-ph/0406151}}.

\bibitem{Kuperberg2013_Hidden_Subgroup}
G.~Kuperberg, ``Another subexponential-time quantum algorithm for the dihedral
  hidden subgroup problem,'' in {\em Proceedings of TQC13}.
\newblock 2013.
\newblock \href{http://arxiv.org/abs/1112.3333}{{\ttfamily arXiv:1112.3333
  [quant-ph]}}.

\bibitem{RoettelerBeth1998_Hidden_Subgroup}
M.~Roetteler and T.~Beth, ``Polynomial-time solution to the hidden subgroup
  problem for a class of non-abelian groups,''
  \href{http://arxiv.org/abs/quant-ph/9812070}{{\ttfamily
  arXiv:quant-ph/9812070}}.

\bibitem{IvanyosMagniezSantha2001_Hidden_Subgroup}
G.~Ivanyos, F.~Magniez, and M.~Santha, ``Efficient quantum algorithms for some
  instances of the non-abelian hidden subgroup problem,''
  \href{http://arxiv.org/abs/quant-ph/0102014}{{\ttfamily
  arXiv:quant-ph/0102014}}.

\bibitem{MooreRockmoreRussellSchulman2004}
C.~Moore, D.~Rockmore, A.~Russell, and L.~Schulman, ``The power of basis
  selection in fourier sampling: the hidden subgroup problem in affine
  groups.,'' in {\em Proceedings of the 15th ACM-SIAM Symposium on Discrete
  Algorithms}.
\newblock 2004.
\newblock \href{http://arxiv.org/abs/quant-ph/0211124}{{\ttfamily
  arXiv:quant-ph/0211124}}.

\bibitem{InuiLeGall2007_Hidden_Subgroup}
Y.~Inui and F.~Le~Gall, ``Efficient quantum algorithms for the hidden subgroup
  problem over a class of semi-direct product groups,'' {\em Quantum
  Information and Computation} {\bfseries 7} no.~5/6, (2007) ,
  \href{http://arxiv.org/abs/0412033}{{\ttfamily arXiv:0412033 [quant-ph]}}.

\bibitem{BaconChildsVDam2005_Hidden_Subgroup}
D.~Bacon, A.~M. Childs, and W.~van Dam, ``From optimal measurement to efficient
  quantum algorithms for the hidden subgroup problem over semidirect product
  groups,'' in {\em Proceedings of the 46th IEEE Symposium on Foundations of
  Computer Science}.
\newblock 2005.
\newblock \href{http://arxiv.org/abs/0504083}{{\ttfamily arXiv:0504083
  [quant-ph]}}.

\bibitem{ChiKimLee2006_Hidden_Subgroup}
D.~P. Chi, J.~S. Kim, and S.~Lee, ``Notes on the hidden subgroup problem on
  some semi-direct product groups,'' {\em Physical Letters A} {\bfseries 359}
  no.~2, (2006) , \href{http://arxiv.org/abs/quant-ph/0604172}{{\ttfamily
  arXiv:quant-ph/0604172}}.

\bibitem{IvanyosSanselmeSantha2007_Hidden_Subgroup}
G.~Ivanyos, L.~Sanselme, and M.~Santha, ``An efficient quantum algorithm for
  the hidden subgroup problem in extraspecial groups,'' in {\em Proceedings of
  the 24th Symposium on Theoretical Aspects of Computer Science}.
\newblock 2007.
\newblock \href{http://arxiv.org/abs/quant-ph/0701235}{{\ttfamily
  arXiv:quant-ph/0701235}}.

\bibitem{MagnoCosmePortugal2007_Hidden_Subgroup}
C.~Magno, M.~Cosme, and R.~Portugal, ``Quantum algorithm for the hidden
  subgroup problem on a class of semidirect product groups,''
  \href{http://arxiv.org/abs/quant-ph/0703223}{{\ttfamily
  arXiv:quant-ph/0703223}}.

\bibitem{IvanyosSanselmeSantha2007_Nil2_Groups}
G.~Ivanyos, L.~Sanselme, and M.~Santha, ``An efficient quantum algorithm for
  the hidden subgroup problem in nil-2 groups,'' in {\em LATIN 2008:
  Theoretical Informatics}, vol.~4957 of {\em LNCS}.
\newblock Springer, 2008.
\newblock \href{http://arxiv.org/abs/0707.1260}{{\ttfamily arXiv:0707.1260
  [quant-ph]}}.

\bibitem{FriedlIvanyosMagniezSanthaSen2003_Hidden_Translation}
K.~Friedl, G.~Ivanyos, F.~Magniez, M.~Santha, and P.~Sen, ``Hidden translation
  and translating coset in quantum computing,'' in {\em Proceedings of the 35th
  ACM Symposium on Theory of Computing}.
\newblock 2003.
\newblock \href{http://arxiv.org/abs/quant-ph/0211091}{{\ttfamily
  arXiv:quant-ph/0211091}}.

\bibitem{Gavinsky2004_Hidden_Subgroup}
D.~Gavinsky, ``Quantum solution to the hidden subgroup problem for
  poly-near-hamiltonian-groups,'' {\em Quantum Information and Computation}
  {\bfseries 4} (2004) .

\bibitem{ChildsVDam2007_Hidden_Shift}
A.~M. Childs and W.~van Dam, ``Quantum algorithm for a generalized hidden shift
  problem,'' in {\em Proceedings of the 18th ACM-SIAM Symposium on Discrete
  Algorithms}.
\newblock 2007.
\newblock \href{http://arxiv.org/abs/quant-ph/0507190}{{\ttfamily
  arXiv:quant-ph/0507190}}.

\bibitem{DenneyMooreRussel2010_Conjugate_Stabilizer_Subgroups}
A.~Denney, C.~Moore, and A.~Russell, ``Finding conjugate stabilizer subgroups
  in psl(2;q) and related groups,'' {\em Quantum Information and Computation}
  {\bfseries 10} no.~3, (2010) ,
  \href{http://arxiv.org/abs/0809.2445}{{\ttfamily arXiv:0809.2445
  [quant-ph]}}.

\bibitem{Wallach2013_Hidden_Subgroup}
N.~Wallach, ``A quantum polylog algorithm for non-normal maximal cyclic hidden
  subgroups in the affine group of a finite field,''
  \href{http://arxiv.org/abs/1308.1415}{{\ttfamily arXiv:1308.1415
  [quant-ph]}}.

\bibitem{lomont_HSP_review}
C.~Lomont, ``The hidden subgroup problem - review and open problems,'' {\em
  arXiv} (2004) ,
  \href{http://arxiv.org/abs/arXiv:quant-ph/0411037v1}{{\ttfamily
  arXiv:quant-ph/0411037v1}}.

\bibitem{childs_lecture_8}
A.~Childs, {\em Lecture Notes on Quantum Algorithms.}
\newblock University of Waterloo, 2011.
\newblock Published online.

\bibitem{Arvind97solvableblack-box}
V.~Arvind and N.~Vinodchandran, ``Solvable black-box group problems are low for
  pp,'' {\em Theoretical Computer Science} {\bfseries 180} (1997) .

\bibitem{Babai1991_Vertex_Transive_Graphs_Random_Generation_Finite_Groups}
L.~Babai, \href{http://dx.doi.org/10.1145/103418.103440}{``Local expansion of
  vertex-transitive graphs and random generation in finite groups,''} in {\em
  Proceedings of the Twenty-third Annual ACM Symposium on Theory of Computing},
  STOC '91.
\newblock ACM, 1991.

\bibitem{Babai1992_Bounded_Round_Interactive_Proofs_Finite_Groups}
L.~Babai, ``Bounded round interactive proofs in finite groups,''
  \href{http://dx.doi.org/10.1137/0405008}{{\em SIAM J. Discret. Math.}
  {\bfseries 5} no.~1, (Feb., 1992) }.

\bibitem{Babai97_Randomization_group_algorithms}
L.~Babai, ``Randomization in group algorithms: conceptual questions,'' in {\em
  Groups and Computation II}, vol.~28 of {\em DIMACS Ser. Discrete Math.
  Theoret. Comput. Sci.}
\newblock 1997.

\bibitem{Zhang11Decomposing}
Y.~Zhang, ``Quantum algorithm for decomposing black-box finite abelian
  groups,'' in {\em Proceedings of the 7th Annual International Conference on
  Foundations of Computer Science}.
\newblock 2011.

\bibitem{watrous00_quantumAlgorithms_solvableGroups}
J.~Watrous, ``Quantum algorithms for solvable groups,'' in {\em Proceedings of
  the 33rd ACM Symposium on Theory of Computing}.
\newblock 2001.
\newblock \href{http://arxiv.org/abs/quant-ph/0011023}{{\ttfamily
  quant-ph/0011023}}.

\bibitem{MagniezNayak2005_Group_Commutativity}
F.~Magniez and A.~Nayak, ``Quantum complexity of testing group commutativity,''
  in {\em Proceedings of 32nd International Colloquium on Automata, Languages
  and Programming}, vol.~3580 of {\em LNCS}.
\newblock 2005.
\newblock \href{http://arxiv.org/abs/quant-ph/0506265}{{\ttfamily
  arXiv:quant-ph/0506265}}.

\bibitem{Fenner05_QAlg_Group_Theoretic_Problems}
S.~A. Fenner and Y.~Zhang,
  \href{http://dx.doi.org/10.1007/11560586_18}{``Quantum algorithms for a set
  of group theoretic problems,''} in {\em Proceedings of the 9th Italian
  Conference on Theoretical Computer Science}, ICTCS'05.
\newblock Springer, 2005.
\newblock \url{http://dx.doi.org/10.1007/11560586_18}.

\bibitem{LeGall2010_Group_Isomorphism}
F.~Le~Gall, ``An efficient quantum algorithm for some instances of the group
  isomorphism problem,'' in {\em Proceedings of STACS}.
\newblock 2010.
\newblock \href{http://arxiv.org/abs/1001.0608}{{\ttfamily arXiv:1001.0608
  [quant-ph]}}.

\bibitem{Zatloukal2013_Equivalent_Group_Extensions}
K.~C. Zatloukal, ``Classical and quantum algorithms for testing equivalence of
  group extensions,'' \href{http://arxiv.org/abs/1305.1327}{{\ttfamily
  arXiv:1305.1327 [quant-ph]}}.

\bibitem{Bravyi05_Lagrangian_Rep_Fermionic_Linear_Optics}
S.~Bravyi, ``Lagrangian representation for fermionic linear optics,'' {\em
  Quantum Info. Comput.} {\bfseries 5} no.~3, (2005) ,
  \href{http://arxiv.org/abs/quant-ph/0404180}{{\ttfamily quant-ph/0404180}}.

\bibitem{JozsaKrausMiyakeWatrous10Matchgates}
R.~Jozsa, B.~Kraus, A.~Miyake, and J.~Watrous, ``Matchgate and space-bounded
  quantum computations are equivalent,''
  \href{http://dx.doi.org/10.1098/rspa.2009.0433}{{\em Proceedings of the Royal
  Society A: Mathematical, Physical and Engineering Science} {\bfseries 466}
  no.~2115, (2010) }.

\bibitem{VdNest11_Matchgates}
M.~Van~den Nest, ``Quantum matchgate computations and linear threshold gates,''
  \href{http://dx.doi.org/10.1098/rspa.2010.0332}{{\em Proceedings of the Royal
  Society A: Mathematical, Physical and Engineering Science} {\bfseries 467}
  no.~2127, (2011) }, \href{http://arxiv.org/abs/1005.1143}{{\ttfamily
  arXiv:1005.1143 [quant-ph]}}.

\bibitem{BravyiKoenig12_Simulation_Dissipative_Fermionic_Linear_Optics}
S.~Bravyi and R.~K\"{o}nig, ``Classical simulation of dissipative fermionic
  linear optics,'' {\em Quantum Info. Comput.} {\bfseries 12} no.~11-12, (2012)
  , \href{http://arxiv.org/abs/1112.2184}{{\ttfamily arXiv:1112.2184
  [quant-ph]}}.

\bibitem{deMeloCwiklinskiTerhal13_Noisy_Fermionic_Quantum_Computation}
F.~de~Melo, P.~{\'C}wikli{\'n}ski, and B.~M. Terhal, ``The power of noisy
  fermionic quantum computation,'' {\em New Journal of Physics} {\bfseries 15}
  no.~1, (2013) , \href{http://arxiv.org/abs/1208.5334}{{\ttfamily
  arXiv:1208.5334 [quant-ph]}}.

\bibitem{AmbainisShulmanVazirani06_Computing_highly_mixed_states}
A.~Ambainis, L.~J. Schulman, and U.~Vazirani, ``Computing with highly mixed
  states,'' \href{http://dx.doi.org/10.1145/1147954.1147962}{{\em J. ACM}
  {\bfseries 53} no.~3, (2006) },
  \href{http://arxiv.org/abs/quant-ph/0003136}{{\ttfamily quant-ph/0003136}}.

\bibitem{Pouline03_Integrability_DQC1}
D.~Poulin, R.~Laflamme, G.~J. Milburn, and J.~P. Paz, ``Testing integrability
  with a single bit of quantum information,''
  \href{http://dx.doi.org/10.1103/PhysRevA.68.022302}{{\em Phys. Rev. A}
  {\bfseries 68} (2003) },
  \href{http://arxiv.org/abs/quant-ph/0303042}{{\ttfamily quant-ph/0303042}}.

\bibitem{Poulin04_Fidelity_Decay_DQC1}
D.~Poulin, R.~Blume-Kohout, R.~Laflamme, and H.~Ollivier, ``Exponential speedup
  with a single bit of quantum information: Measuring the average fidelity
  decay,'' \href{http://dx.doi.org/10.1103/PhysRevLett.92.177906}{{\em Phys.
  Rev. Lett.} {\bfseries 92} (2004) },
  \href{http://arxiv.org/abs/quant-ph/0310038}{{\ttfamily quant-ph/0310038}}.

\bibitem{Shepherd06_DQC1}
D.~Shepherd, ``{C}omputation with {U}nitaries and {O}ne {P}ure {Q}ubit,'' 2006.
\newblock arXiv:quant-ph/0608132v2.

\bibitem{ShorJordan08_jones_polynomial_Complete_DQC1}
P.~W. Shor and S.~P. Jordan, ``Estimating {J}ones polynomials is a complete
  problem for one clean qubit,'' {\em Quantum Info. Comput.} {\bfseries 8}
  no.~8, (2008) , \href{http://arxiv.org/abs/0707.2831}{{\ttfamily
  arXiv:0707.2831 [quant-ph]}}.

\bibitem{JordanWocjan09_DQC1_Jones_Homfly_polynomials}
S.~P. Jordan and P.~Wocjan, ``Estimating {J}ones and {H}omfly polynomials with
  one clean qubit,'' {\em Quantum Info. Comput.} {\bfseries 9} no.~3, (2009) ,
  \href{http://arxiv.org/abs/0807.4688}{{\ttfamily arXiv:0807.4688
  [quant-ph]}}.

\bibitem{jordan2014approximating_Turaev_Viro_DQC1}
S.~P. Jordan and G.~Alagic, ``Approximating the turaev-viro invariant of
  mapping tori is complete for one clean qubit,'' in {\em Theory of Quantum
  Computation, Communication, and Cryptography}.
\newblock Springer, 2014.
\newblock \href{http://arxiv.org/abs/1105.5100}{{\ttfamily arXiv:1105.5100
  [quant-ph]}}.

\bibitem{MorimaeFujiiFitzsimons14_Hardness_Simulating_DQC1}
T.~Morimae, K.~Fujii, and J.~F. Fitzsimons, ``Hardness of classically
  simulating the one-clean-qubit model,''
  \href{http://dx.doi.org/10.1103/PhysRevLett.112.130502}{{\em Phys. Rev.
  Lett.} {\bfseries 112} (2014) },
  \href{http://arxiv.org/abs/1312.2496}{{\ttfamily arXiv:1312.2496
  [quant-ph]}}.

\bibitem{Aaronson11_Computational_Complexity_Linear_Optics}
S.~Aaronson and A.~Arkhipov,
  \href{http://dx.doi.org/10.1145/1993636.1993682}{``The computational
  complexity of linear optics,''} in {\em Proceedings of the Forty-third Annual
  ACM Symposium on Theory of Computing}, STOC '11.
\newblock ACM, 2011.
\newblock \href{http://arxiv.org/abs/1011.3245}{{\ttfamily arXiv:1011.3245
  [quant-ph]}}.

\bibitem{Veitch12_Negative_QuasiProbability_Resource_QC}
V.~Veitch, C.~Ferrie, D.~Gross, and J.~Emerson, ``Negative quasi-probability as
  a resource for quantum computation,'' {\em New Journal of Physics} {\bfseries
  14} no.~11, (2012) , \href{http://arxiv.org/abs/1201.1256}{{\ttfamily
  arXiv:1201.1256 [quant-ph]}}.

\bibitem{MariEisert12_Positive_Wigner_Functions_Quantum_Computation}
A.~Mari and J.~Eisert, ``Positive {W}igner functions render classical
  simulation of quantum computation efficient,''
  \href{http://dx.doi.org/10.1103/PhysRevLett.109.230503}{{\em Phys. Rev.
  Lett.} {\bfseries 109} (2012) },
  \href{http://arxiv.org/abs/1208.3660}{{\ttfamily arXiv:1208.3660
  [quant-ph]}}.

\bibitem{VeitchWiebeFerrieEmerson13_Simulation_scheme_large_class_quantum_optics_experiments}
V.~Veitch, N.~Wiebe, C.~Ferrie, and J.~Emerson, ``Efficient simulation scheme
  for a class of quantum optics experiments with non-negative wigner
  representation,'' {\em New Journal of Physics} {\bfseries 15} no.~1, (2013) ,
  \href{http://arxiv.org/abs/1210.1783}{{\ttfamily arXiv:1210.1783
  [quant-ph]}}.

\bibitem{VDNest2012_Little_Entanglement}
M.~Van Den~Nest, ``Universal quantum computation with little entanglement,''
  \href{http://dx.doi.org/10.1103/PhysRevLett.110.060504}{{\em Physical Review
  Letters} {\bfseries 110} (2012) },
  \href{http://arxiv.org/abs/1204.3107}{{\ttfamily arXiv:1204.3107
  [quant-ph]}}.

\bibitem{Jozsa03_Role_Entanglement_Quantum_Computational_SpeedUP}
R.~Jozsa and N.~Linden, ``On the role of entanglement in quantum-computational
  speed-up,'' {\em Proceedings of the Royal Society of London. Series A.
  Mathematical, Physical and Engineering Sciences} {\bfseries 459} (2003) ,
  \href{http://arxiv.org/abs/quant-ph/0201143}{{\ttfamily quant-ph/0201143}}.

\bibitem{Vidal_03_Efficient_Simulation_Sligtly_Entangled}
G.~Vidal, ``Efficient classical simulation of slightly entangled quantum
  computations,'' \href{http://dx.doi.org/10.1103/PhysRevLett.91.147902}{{\em
  Phys. Rev. Lett.} {\bfseries 91} (2003) }.

\bibitem{TerhalDiVincenzo02Adaptive_ConstantDepth_Quantum_Comp}
B.~M. Terhal and D.~P. DiVincenzo, ``{A}daptive {Q}uantum {C}omputation,
  {C}onstant {D}epth {Q}uantum {C}ircuits and {A}rthur-{M}erlin {G}ames,'' {\em
  Quantum Info. Comput.} {\bfseries 4} no.~2, (2014) ,
  \href{http://arxiv.org/abs/quant-ph/0205133}{{\ttfamily quant-ph/0205133}}.

\bibitem{MarkovShi08_Simulating_QuantumComp_TensorNetwork}
I.~Markov and Y.~Shi, ``Simulating quantum computation by contracting tensor
  networks,'' \href{http://dx.doi.org/10.1137/050644756}{{\em SIAM Journal on
  Computing} {\bfseries 38} no.~3, (2008) },
  \href{http://arxiv.org/abs/quant-ph/0511069}{{\ttfamily quant-ph/0511069}}.

\bibitem{aharonov_AQFT}
D.~Aharonov, Z.~Landau, and J.~Makowsky, ``The quantum {FFT} can be classically
  simulated,'' {\em arXiv} (2006) ,
  \href{http://arxiv.org/abs/quant-ph/0611156v2}{{\ttfamily
  quant-ph/0611156v2}}.

\bibitem{yoran_short_QFT}
N.~Yoran and A.~J. Short, ``Efficient classical simulation of the approximate
  quantum {F}ourier transform,''
  \href{http://dx.doi.org/10.1103/PhysRevA.76.042321}{{\em Phys. Rev. A}
  {\bfseries 76} (2007) },
  \href{http://arxiv.org/abs/quant-ph/0611241v1}{{\ttfamily
  quant-ph/0611241v1}}.

\bibitem{browne_QFT}
D.~E. Browne, ``Efficient classical simulation of the quantum fourier
  transform,'' {\em New Journal of Physics} {\bfseries 9} no.~5, (2007) ,
  \href{http://arxiv.org/abs/quant-ph/0612021}{{\ttfamily quant-ph/0612021}}.

\bibitem{Yoran08_Contractable_circults_little_entanglement}
N.~Yoran, ``{E}fficiently contractable quantum circuits cannot produce much
  entanglement,'' \href{http://arxiv.org/abs/0802.1156}{{\ttfamily
  arXiv:0802.1156 [quant-ph]}}.

\bibitem{nest_weak_simulations}
M.~Van~den Nest, ``Simulating quantum computers with probabilistic methods,''
  {\em Quantum Info. Comput.} {\bfseries 11} no.~9-10, (2011) ,
  \href{http://arxiv.org/abs/0911.1624v3}{{\ttfamily arXiv:0911.1624v3
  [quant-ph]}}.

\bibitem{Stahlke14_Interference_resource_speedup}
D.~Stahlke, ``Quantum interference as a resource for quantum speedup,''
  \href{http://dx.doi.org/10.1103/PhysRevA.90.022302}{{\em Phys. Rev. A}
  {\bfseries 90} (2014) }, \href{http://arxiv.org/abs/1305.2186}{{\ttfamily
  arXiv:1305.2186 [quant-ph]}}.

\bibitem{Jordan10_Permutational_Quantum_Computing}
S.~P. Jordan, ``Permutational quantum computing,'' {\em Quantum Info. Comput.}
  {\bfseries 10} no.~5, (2010) ,
  \href{http://arxiv.org/abs/0906.2508}{{\ttfamily arXiv:0906.2508
  [quant-ph]}}.

\bibitem{schwarz2013simulating}
M.~Schwarz and M.~V.~d. Nest, ``Simulating quantum circuits with sparse output
  distributions,'' {\em Electronic Colloquium on Computational Complexity}
  (2013) , \href{http://arxiv.org/abs/1310.6749}{{\ttfamily arXiv:1310.6749
  [quant-ph]}}.

\bibitem{bermejo2011classical}
J.~Bermejo-Vega, ``Classical simulations of non-abelian quantum {F}ourier
  transforms,'' Master's thesis, Technical University of Munich, 2011.

\bibitem{Sarvepalli14_1D_infrastructures}
P.~Sarvepalli and P.~Wocjan, ``Quantum algorithms for one-dimensional
  infrastructures,'' {\em Quantum Info. Comput.} {\bfseries 14} no.~1-2, (2014)
  , \href{http://arxiv.org/abs/1106.6347}{{\ttfamily arXiv:1106.6347
  [quant-ph]}}.

\bibitem{FonteinWocjan11_Q_Alg_Period_Lattice_Infrastructure}
F.~Fontein and P.~Wocjan, ``{Q}uantum {A}lgorithm for {C}omputing the {P}eriod
  {L}attice of an {I}nfrastructure,'' {\em arXiv} (2011) ,
  \href{http://arxiv.org/abs/1111.1348}{{\ttfamily arXiv:1111.1348
  [quant-ph]}}.

\bibitem{Hallgren07_Pells_equation}
S.~Hallgren, ``Polynomial-time quantum algorithms for {P}ell's equation and the
  principal ideal problem,''
  \href{http://dx.doi.org/10.1145/1206035.1206039}{{\em J. ACM} {\bfseries 54}
  no.~1, (2007) }.

\bibitem{Jozsa03_Hallgrens_Algorithm}
R.~Jozsa, ``Quantum computation in algebraic number theory: {H}allgren's
  efficient quantum algorithm for solving {P}ell's equation,''
  \href{http://dx.doi.org/http://dx.doi.org/10.1016/S0003-4916(03)00067-8}{{\em
  Annals of Physics} {\bfseries 306} no.~2, (2003) },
  \href{http://arxiv.org/abs/quant-ph/0302134}{{\ttfamily quant-ph/0302134}}.

\bibitem{Schmidt05_Q_Algorithm_Computation_Unit_Group}
A.~Schmidt and U.~Vollmer,
  \href{http://dx.doi.org/10.1145/1060590.1060661}{``Polynomial time quantum
  algorithm for the computation of the unit group of a number field,''} in {\em
  Proceedings of the Thirty-seventh Annual ACM Symposium on Theory of
  Computing}, STOC '05.
\newblock ACM, 2005.

\bibitem{Hallgren2005_Unit_Group_Class_Group}
S.~Hallgren, \href{http://dx.doi.org/10.1145/1060590.1060660}{``Fast quantum
  algorithms for computing the unit group and class group of a number field,''}
  in {\em Proceedings of the Thirty-seventh Annual ACM Symposium on Theory of
  Computing}, STOC '05.
\newblock ACM, 2005.

\bibitem{Childs14_Discrete_Log_Semigroups}
A.~M. Childs and G.~Ivanyos, ``Quantum computation of discrete logarithms in
  semigroups,'' \href{http://arxiv.org/abs/1310.6238}{{\ttfamily
  arXiv:1310.6238 [quant-ph]}}.

\bibitem{VanDamSeroussi02_Gauss_Sums_QALG}
U.~B. Wim~van Dam~(HP, MSRI and G.~S. (HP), ``{E}fficient {Q}uantum
  {A}lgorithms for {E}stimating {G}auss {S}ums,'' {\em arXiv} (2002) ,
  \href{http://arxiv.org/abs/quant-ph/0207131}{{\ttfamily quant-ph/0207131}}.

\bibitem{raussen_briegel_01_Cluster_State}
H.~J. Briegel and R.~Raussendorf, ``Persistent entanglement in arrays of
  interacting particles,''
  \href{http://dx.doi.org/10.1103/PhysRevLett.86.910}{{\em Phys. Rev. Lett.}
  {\bfseries 86} (2001) },
  \href{http://arxiv.org/abs/quant-ph/0004051v2}{{\ttfamily
  quant-ph/0004051v2}}.

\bibitem{raussen_briegel_onewayQC}
R.~Raussendorf and H.~J. Briegel, ``A one-way quantum computer,''
  \href{http://dx.doi.org/10.1103/PhysRevLett.86.5188}{{\em Phys. Rev. Lett.}
  {\bfseries 86} (2001) }.

\bibitem{brent_zimmerman10CompArithmetic}
R.~P. Brent and P.~Zimmermann, {\em Modern Computer Arithmetic}.
\newblock Cambridge University Press, 2010.

\bibitem{Fontein2014_Probability_Generating_Lattice}
F.~Fontein and P.~Wocjan, ``On the probability of generating a lattice,''
  \href{http://dx.doi.org/http://dx.doi.org/10.1016/j.jsc.2013.12.002}{{\em
  Journal of Symbolic Computation} {\bfseries 64} no.~0, (2014) 3 -- 15},
  \href{http://arxiv.org/abs/1211.6246}{{\ttfamily arXiv:1211.6246
  [quant-ph]}}.

\bibitem{Bruhat61_Schwatz-Bruhat-functions}
F.~{B}ruhat, ``{D}istributions sur un groupe localement compact et applications
  {\`a} l etude des repr{\'e}sentations des groupes $p$-adiques,'' {\em {B}ull.
  {S}oc. {M}ath. {F}rance} {\bfseries 89} (1961) .

\bibitem{Osborne75_Schwartz_Bruhat}
M.~S. {O}sborne, ``{O}n the {S}chwartz-{B}ruhat space and the {P}aley-{W}iener
  theorem for locally compact {A}belian groups,'' {\em {J}. {F}unct. {A}nal.}
  {\bfseries 19} (1975) .

\bibitem{delaMadrid05_roleofthe_riggedHilbert}
R.~de~la Madrid, ``The role of the rigged {H}ilbert space in quantum
  mechanics,'' {\em European Journal of Physics} {\bfseries 26} no.~2, (2005) ,
  \href{http://arxiv.org/abs/quant-ph/0502053}{{\ttfamily quant-ph/0502053}}.

\bibitem{Antoine98_QM_beyond_Hilber_space}
J.~P. Antoine, ``Quantum mechanics beyond {H}ilbert space,'' in {\em
  Irreversibility and Causality Semigroups and Rigged {H}ilbert Spaces},
  vol.~504 of {\em Lecture Notes in Physics}.
\newblock Springer, 1998.

\bibitem{Gadella02_unified_Dirac_formalism}
M.~Gadella and F.~G\'{o}mez, ``A unified mathematical formalism for the {D}irac
  formulation of quantum mechanics,'' {\em Foundations of Physics} {\bfseries
  32} no.~6, (2002) .

\bibitem{gadella12_Riggings_LCA_Groups}
M.~Gadella, F.~G\'omez, and S.~Wickramasekara, ``{Riggings of locally compact
  {A}belian groups.},'' {\em J. Geom. Symmetry Phys.} {\bfseries 11} (2008) .

\bibitem{nielsen_chuang}
M.~A. Nielsen and I.~L. Chuang, {\em Quantum Computation and Quantum
  Information}.
\newblock Cambridge University Press, 2000.

\bibitem{Cohen:1995Course_Computational_Algebraic_Number_Theory}
H.~Cohen, {\em A Course in Computational Algebraic Number Theory}.
\newblock Springer, 1993.

\bibitem{oppenheim_Signals_and_Systems}
A.~V. Oppenheim, A.~S. Willsky, and S.~H. Nawab, {\em Signals \&Amp; Systems
  (2Nd Ed.)}.
\newblock Prentice-Hall, 1996.

\bibitem{KLM_QC_07}
P.~R. Kaye, R.~Laflamme, and M.~Mosca, {\em An Introduction to Quantum
  Computing}.
\newblock Oxford University Press, 2007.

\bibitem{Knill95onshors}
E.~Knill, ``On shor's quantum factor finding algorithm: Increasing the
  probability of success and tradeoffs involving the fourier transform
  modulus,'' 1995.

\bibitem{lorenzini1997invitation}
D.~Lorenzini, {\em An Invitation to Arithmetic Geometry}.
\newblock Graduate studies in mathematics. American Mathematical Society, 1997.

\bibitem{Cohen00_Advanced_Topics_ComputationalNumber_Theory}
H.~Cohen, \href{http://dx.doi.org/10.1007/978-1-4419-8489-0}{{\em Advanced
  Topics in Computational Number Theory}}.
\newblock Graduate Texts in Mathematics. Springer.

\bibitem{HurtWaid70_Integral_Generalized_Inverse}
M.~F. Hurt and C.~Waid, ``A generalized inverse which gives all the integral
  solutions to a system of linear equations,''.

\bibitem{BowmanBurget74_systems-Mixed-Integer_Linear_equations}
V.~J. Bowman and C.-A. Burdet, ``On the general solution to systems of
  mixed-integer linear equations,'' {\em SIAM Journal on Applied Mathematics}
  {\bfseries 26} no.~1, (1974) .

\bibitem{Nagaj2009_Fast_Amplification_QMA}
D.~Nagaj, P.~Wocjan, and Y.~Zhang, ``Fast amplification of {QMA},'' {\em
  Quantum Info. Comput.} {\bfseries 9} no.~11, (2009) ,
  \href{http://arxiv.org/abs/0904.1549}{{\ttfamily arXiv:0904.1549
  [quant-ph]}}.

\end{thebibliography}\endgroup
\normalsize

\appendix

\section{Proof of theorem \ref{thm:ModExp requires Z}}\label{app:ModExp requires Z}

To prove the result we can assume that we know a group isomorphism $\varphi:\Z_N^\times\rightarrow G$ that decomposes the black-box group as a product of cyclic factors $G= \DProd{N}{d}$. Let $U_{\varphi}:\mathcal{H}_{\mathbf{B}}\rightarrow \mathcal{H}_G$ be the unitary that implements the isomorphism  $U_\varphi\ket{b}=\ket{\varphi(b)}$ for any $b\in\Z_N^\times$. It is easy to check that $\mathcal{C}$ is a normalizer circuit over $\Z_M\times G$ if and only if $(I\otimes U_\varphi) \mathcal{C} (I\otimes U_\varphi)^\dagger $ is a normalizer circuit over $\Z_M\times \Z_N^\times$: automorphism (resp. quadratic phase) gates get mapped to automorphism (resp. quadratic phase) gates and vice-versa; isomorphic groups have isomorphic character groups \cite{Morris77_Pontryagin_Duality_and_LCA_groups}, and therefore Fourier transforms get mapped to Fourier transforms.
  
  As a result, it is enough to prove the result in the basis labeled by elements of $\Z_M^\times\times G$. The advantage now is that we can use results from  \cite{VDNest_12_QFTs, BermejoVega_12_GKTheorem}. In fact, the rest of the proof will be similar to the proof of theorem 2 in \cite{VDNest_12_QFTs}.
  
  We define $\alpha:=\varphi(a)$. The action of $U_\mathrm{me}$ in the $G$-basis reads  $U_\mathrm{me}\ket{m,g}=\ket{m,m\alpha+g}$, in additive notation. Define a fuction $F(m,g)=(m,m\alpha+g)$.  We now assume that $M$ is not divisible by $|a|$ and that there exists a normalizer circuit $\mathcal{C}$ such that $\|\mathcal{C}-U_\mathrm{me}\|\leq \delta$ with $\delta=1-1/\sqrt{2}$ and try to arrive to a contradiction. This property implies that $\|\mathcal{C}\ket{m,g}-U_\mathrm{me}\ket{m,g}\|\leq \delta$ for any standard basis state, and consequently
  \begin{equation}\label{inproof:Ume is no Clifford}
  |\bra{F(m,g)}\mathcal{C}\ket{m,g}|\geq 1- \delta=\tfrac{1}{\sqrt{2}}
  \end{equation}
  It was shown in  \cite{VDNest_12_QFTs} that $\mathcal{C}\ket{m,g}$ is a uniform superposition over some subset $x+H$ of $G$, being $H$ a subgroup. If $H$ has more than two elements, then $\mathcal{C}\ket{m,g}$ is a uniform superposition over more than two computational basis states. It follows that $\bra{n,h}\mathcal{C}\ket{m,g}\leq\frac{1}{\sqrt{2}}$ for any basis state $\ket{n,h}$ in contradiction with \ref{inproof:Ume is no Clifford}, so that we can assume $H=\{0\}$ and  that $\mathcal{C}\ket{m,g}$ is a standard basis state. Then (\ref{inproof:Ume is no Clifford}) implies that $\ket{F(m,g)}$ and $\mathcal{C}\ket{m,g}$ must coincide for every $(m,g)\in \Z_M\times G$, so that $\mathcal{C}$ must perfectly realize the transformation $\ket{(m,g)}\rightarrow\ket{F(m,g)}$; however, the only classical functions that can be implemented by normalizer circuits of this form are affine maps \cite{VDNest_12_QFTs}, meaning that $F(m,g)=f(m,g)+b$ for some group automorphism $f:\Z_M\times G\rightarrow\Z_M\times G$ and some $b\in \Z_M\times G$. 
  
  Finally, we arrive to a contradiction showing that $F(m,g)$ is  not affine unless $M$ is a multiple of $|a|$. First, by evaluating $F(m,g)=f(m,g)+b=(m,m\alpha+g)$ at $(0,0)$,$(1,0)$ and elements of the form $(0,g)$, we can  compute $b$ and a matrix representation $A$ of the automorphism $f$ \cite{VDNest_12_QFTs}: we obtain $b=0$, so that $F(m,g)$ must be an automorphism, and $A=\begin{pmatrix}
    1 & 0\\
    \alpha & 1 
    \end{pmatrix}$.
 However, for the matrix  $A$ to be a matrix representation of a group  automorphism, the first column $a_1$ needs to fulfill the equation:  $Ma_1\pmod{\Z_M\times G}$ \cite[lemma 2]{BermejoVega_12_GKTheorem}. Expanding this equation, we finally get  that $M\alpha = 0\pmod{G}$, which means that $M$ needs to be a multiple of the order of $\alpha$.
 
\section{Proof of theorem \ref{thm:Simulation}} \label{app:Simulation proof}
In this section we will prove theorem \ref{thm:Simulation}. The proof uses results of Section \ref{sect:Group Decomposition}; the reader may wish to review that section before proceeding with this proof.

We first state the simulation result of \cite{BermejoLinVdN13_Infinite_Normalizers}, summarized below:

\begin{theorem}[{\cite[theorem~1]{BermejoLinVdN13_Infinite_Normalizers}}] Let $G = \Z^a  \times \DProd{N}{c}\times \T^b$ be an explicitly decomposed elementary abelian group. Suppose we are given a normalizer circuit over $G$, where each gate is specified as follows (assume each gate is such that all entries of the matrices and vectors below are rational):
\begin{itemize}
\item A (partial) quantum Fourier transform is specified by the elementary subgroups it acts on.
\item A group automorphism is specified as a matrix $A$, as in the normal form that we will introduce in Thm \ref{thm:Normal form of a matrix representation}.
\item A quadratic phase gate is specified as $(M,v)$, where $M$ is a matrix and $v$ is a matrix, as in the normal form that we will introduce in Thm \ref{thm:Normal form of a quadratic function}.
\end{itemize}
Then the output distribution of this normalizer circuit can be sampled classically efficiently.
\end{theorem}

We will describe precisely what we mean by the \emph{normal forms} of normalizer gates in the following sections.

Given a black-box normalizer circuit acting on a black-box group $G= \Z^a \times \T^b \times \DProd{N}{c} \times \mathbf{B}$, there are two things we need to do to de-blackbox it, so that the circuit can be classically simulated:

\begin{enumerate}
\item Decompose the black-box portion of $G$, $\mathbf{B}$, into cyclic subgroups: $\mathbf{B} = \Z_{N_{c+1}} \times \cdots \times \Z_{N_{c+d}}$.
\item Calculate \emph{normal forms} for each of the normalizer gates in the computation.
\end{enumerate}

Assuming we have access to an oracle for Group Decomposition, Task 1 can already be done. In this proof we will concentrate on tackling task 2, for group automorphisms and quadratic phase gates (a quantum Fourier transform is easily specified by the subgroup it acts on).

\subsection{Group automorphisms}
Suppose we have an abelian group $G= \Z^a \times \DProd{N}{c}\times  \T^b $; we can represent each element $g \in G$ as an $a+b+c$-tuple of real numbers $g=(g_1,\cdots,g_m)$, where each of the $g_i$'s are only defined modulo the characteristic $\text{char}(G_i)$ of the group $G_i$, multiplied by some integer. (We have $\text{char}(\Z)=0$, $\text{char}(\T)=1$, and $\text{char}(\Z_N)=N$.) Using this matrix representation, it turns out that there exist \emph{normal forms} for the group automorphisms and quadratic phase functions:

\begin{theorem}[\textbf{Normal form of a matrix representation} {\cite[lemmas~7, 8]{BermejoLinVdN13_Infinite_Normalizers}}]\label{thm:Normal form of a matrix representation}
Let $G = G_1\times\dots\times G_m$ be an elementary Abelian group. A real $n\times m$ matrix $A$ is a valid matrix representation of some group automorphism $\alpha: G\rightarrow G$ iff  $A$ is of the form 
\begin{equation}\label{eq:block-structure of matrix representations}
A:= \begin{pmatrix}
      A_{\Z\Z} & 0 & 0 \\
      A_{F\Z} & A_{FF} & 0 \\
      A_{\T\Z} & A_{\T F} & A_{\T\T}
    \end{pmatrix}
\end{equation}
with the following restrictions:
  \begin{enumerate}
  \item $A_{\Z\Z}$ and $A_{\T\T}$ are arbitrary integer matrices.
  \item $A_{F\Z}$, $A_{FF}$ are integer matrices: the first can be arbitrary, while the coefficients of the second must be of the form \begin{equation}\label{eq:coefficients of Matrix Rep for nonzero characteristic groups}
  A(i,j)= \alpha_{i,j}\, \frac{N_i}{\gcd{(N_i, N_j)}}
  \end{equation}
  where  $\alpha_{i,j}$ can be arbitrary integers, and $N_i$ is the order of the $i$-th cyclic subgroup of $F$, $\Z_{N_i}$. The coefficients of the $i$-th rows of these matrices can be chosen w.l.o.g. to lie in the range $[0,N_i)$ (by taking remainders).
  \item $A_{\T\Z}$ and $A_{\T F}$ are real matrices: the former is arbitrary, while the coefficients of the latter are of the form  $A(i,j)= \alpha_{i,j}/N_j$, where $\alpha_{i,j}$ can be arbitrary integers, and $N_i$ is the order of the $i$-th cyclic subgroup of $F$, $\Z_{N_i}$. (Due to the periodicity of the torus, the coefficients of $A_{\T\Z}$, $A_{\T F}$ can be chosen to lie in the range  $[0,1)$.)
  \end{enumerate}
\end{theorem}

Recall the underlying group is $G= \Z^a \times \T^b \times \DProd{N}{c} \times \mathbf{B}$ with $\mathbf{B} \cong \Z_{N_{c+1}} \times \cdots \times \Z_{N_{c+d}} $. Assume we are given black box access to a group automorphism $\alpha:\: G \rightarrow G$ implemented as a classical function (a uniformly generated circuit family, say). We wish to find a matrix representation for $A$ for $f$. We will assume (for the efficiency of this algorithm) that the size and precision of the coefficients are upper bounded by some known constant $D$, i.e. each element of $M$ can be written as $A_{i,j} = \alpha_{i,j}/\beta_{i,j}$ for integers $\alpha_{i,j},\beta_{i,j}$ with absolute value no more than $D$.\footnote{Note that $D$ can be inferred from the precision bound (see section \ref{sect:circuit model}) of the automorphism gate. If the automorphism gate increases the input size by at most $n$ bits, then it follows that the size of the denominator or numerator of every matrix element can increase by at most $D=2^n$. A similar argument will hold for quadratic phase gates.}

We will show how to find the matrix representation $A$ in two steps:

\begin{enumerate}
\item Given access to $\alpha$, we show how to switch the input and output of $\alpha$ from the black-box encoding (where the group action is implemented as a black-box circuit) to the decomposed group encoding (where elements of the group are given as a list of numbers, and the group action is simply addition of vectors), and vice versa.
\item Once we have achieved this, we can implement classically a rational function $f:\: \mathbb{Q}^{\ell+m+k} \rightarrow \mathbb{Q}^{\ell+m+k}$ that implements $\alpha$ in decomposed form. We will show how to obtain the matrix representation $A$ from this $f$.
\end{enumerate}

\subsubsection{Switching from black-box encoding to decomposed group encoding.}\label{sect:Encodings}

We need to be able to convert elements back and forth from the black-box encoding and the decomposed group encoding. We assume our black box group $\mathbf{B}$ has already been decomposed, i.e. we have found linearly independent generators $b_{1},\cdots,b_{k'}$ of $\mathbf{B}$ such that $\mathbf{B} =\langle \beta_1\rangle\oplus\cdots\oplus\langle \beta_\ell\rangle$; moreover, we know the order $N'_{i} = N_{c+i}$ of $\beta_i$. Define the explicitly decomposed group $\Z_{\mathbf{B}} = \DProd{N'}{d}$; then we need to be able to perform the following tasks:

\begin{enumerate*}
\item[(a)] Our first task is to map an element from the decomposed group $\Z_{\mathbf{B}}$ to the black box group $\mathbf{B}$. In otherwords, we need to be able to compute the following group homomorphism $\varphi$:
\begin{equation*}
\varphi:\Z_{\mathbf{B}}\rightarrow\mathbf{B},\qquad\varphi(g)=b_1^{g(1)}\cdots b_{d}^{g(d)},\quad\textnormal{for any $g\in \Z_{\mathbf{B}}$.}
\end{equation*}
\item[(b)]Our second task is to convert elements from the original black-box group encoding to the new encoding defined by $\Z_{\mathbf{B}}$. In other words, given an arbitrary $\mathbf{b}\in \mathbf{B}$, we need to be able to compute $\varphi^{-1}{(\mathbf{b})}$.
\end{enumerate*}
Note that it is always possible to compute $\varphi(g)=b_1^{g(1)}\cdots b_d^{g(d)}$ for any $g\in \Z_\mathbf{B}$, since this can be done using a polynomial number of queries to the black-box group oracle (using repeated squaring if necessary for the exponentiation). Task (a) is therefore immediate.

As for Task (b), we note that computing $\varphi^{-1}(\mathbf{b})$ for an element $\mathbf{b}\in \mathbf{B}$ is equivalent to finding a list of integers $(g(1),\cdots,g(d))$ such that $b_1^{g(1)}\cdots b_{d}^{g(d)} = \mathbf{b}$. This is a special case of the multivariate discrete logarithm problem, defined in lemma \ref{lemma:Multivariate Discrete Log}; from lemma \ref{lemma:Multivariate Discrete Log} we see that Task (b) can be solved efficiently with a polynomial number of calls to the Group Decomposition oracle.

\subsubsection{Finding the matrix representation $A$}
Now by converting the input and output of $\alpha$ to the decomposed group representation, we may assume that we instead have a classical rational function $f:\: \mathbb{Q}^{a+b+c+d} \rightarrow \mathbb{Q}^{a+b+c+d}$ such that $f$ can be treated as a group automorphism on $G$:
\be
f(x) \equiv f(x') \mbox{ mod } G  \quad \text{if } x \equiv x' \mbox{ mod } G.
\ee
(Here we say two vectors are equal modulo $G$ if each pair of corresponding entries are equal modulo $\text{char}(G_i)$.)
We wish to find a matrix representation $A$ for $f$. For most entries of $A$ this is trivial: note that we have
\be
A_{i,j} \equiv f(e_j)_i \mbox{ mod } c_i
\ee
where $c_i = \text{char}(G_i)$. Hence by evaluating $f$ on the unit vectors $e_i$, we can determine $A_{i,j}$ modulo $c_i$. Thus we can evalute $A_{\T F}$ exactly, the coefficients of the $i$-th rows of $A_{F\Z}$ and $A_{FF}$ modulo $\Z_{N_i}$, and the coefficients of $A_{\T\Z}$ and $A_{\T F}$ modulo $1$. This is sufficient for the cases listed above; the only case we still need to treat is $A_{\T\T}$, whose entries are arbitrary integers (and $c_i=\text{char}(\T)=1$ in this case). We can instead evaluate $f(e_j / \alpha)$ for some large integer $\alpha$:

\be
A_{i,j}/\alpha \equiv f(e_j/\alpha)_i \mbox{ mod } c_i
\ee
which allows us to determine $A_{i,j}$ modulo $\alpha c_i$ for our choice of $\alpha$. Choosing $\alpha > 2D$ then allows us to determine $A_{i,j}$ exactly for the case of $A_{\T\T}$.

\subsection{Quadratic phase functions}
We will continue to use the matrix representation referenced in the last section.

\begin{theorem}[\textbf{Normal form of a quadratic function}{\cite[lemmas~8,~11,~and~theorem~3]{BermejoLinVdN13_Infinite_Normalizers}}] \label{thm:Normal form of a quadratic function} \ \\
Let $G= \Z^a  \times \DProd{N}{c} \times \T^b$ be an elementary Abelian group. Define $\Z^\bullet_N = \{0, 1/N, \cdots, (N-1)/N\}$ to be a group under addition modulo 1, and let $G^\bullet = \T^a  \times \Z_{N_1}^\bullet \times \cdots \times \Z_{N_c}^\bullet\times \Z^b$. Then a function $\xi: G\rightarrow U(1)$ is quadratic if and only if
\begin{equation}
\xi(g)=\euler^{\pii \,\left(g^{\transpose} M g \: +  \: C^{\transpose} g  \: +  \: 2v^\transpose g\right)}
\end{equation}
where $C$, $v$, $M$ are, respectively, two vectors and a matrix that satisfy the following:
\begin{itemize}
\item $v$ is an element of $G^\bullet$;
\item $M$ is the matrix representation of a group homomorphism from $G$ to $G^\bullet$, which necessarily has the form
\begin{equation}
M:= \begin{pmatrix}
      M_{\T\Z} & M_{\T F} & M_{\T\T} \\
      M_{F^{\bullet}\Z} & M_{F^{\bullet}F} & 0 \\
      M_{\Z\Z} & 0 & 0
    \end{pmatrix}
\end{equation}
with the following restrictions:
\begin{itemize}
\item $M_{\Z\Z}$ and $M_{\T\T}$ are arbitrary integer matrices.
\item $M_{F^\bullet\Z}$ and $M_{\T F}$ are rational matrices, the former with the form $M(i,j) = \alpha_{i,j}/N_i$ and the latter with the form $M(i,j) = \alpha_{i,j}/N_j$, where $\alpha_{i,j}$ are arbitrary integers, and $N_i$ is the order of the $i$-th cyclic subgroup $\Z_{N_i}$.
  \item $M_{F^\bullet F}$ is a rational matrix with coefficients of the form \begin{equation}
  M(i,j)= \frac{\alpha_{i,j}}{\gcd{(N_i, N_j)}}
  \end{equation}
  where $\alpha_{i,j}$ are arbitrary integers, and $N_i$ is the order of the $i$-th cyclic subgroup $\Z_{N_i}$.
  \item $M_{\T\Z}$ is an arbitrary real matrix.
\end{itemize}
The entries of $M_{F^\bullet\Z}$, $M_{\T F}$, $M_{F^\bullet F}$, and $M_{\T\Z}$ can be assumed to lie in the interval $[0,1)$.
Moreover, $M$ can be assumed to be symmetric, i.e. $M_{\Z\Z}^\transpose = M_{\T\T}$, $M_{F^\bullet\Z}^\transpose = M_{\T F}$, $M_{F^\bullet F}^\transpose = M_{F^\bullet F}$, and $M_{\T\Z}^\transpose = M_{\T\Z}$.
\item $C$ is an integer vector dependent on $M$, defined component-wise as $C(i)=M(i,i)c_i$, where $c_i$ is the characteristic of the group $G_i$. (Recall that  $\textnormal{char}(\Z)=0$, $\textnormal{char}(\T)=1$, and $\textnormal{char}(\Z_N)=N$.)
\end{itemize}
\end{theorem}
Recall the underlying group is $G= \Z^a \times \T^b \times \DProd{N}{c} \times B$ with $B \cong \Z_{N_{c+1}} \times \cdots \times \Z_{N_{c+d}} $. Assume we are given a quadratic phase gate $\xi$, implemented as a classical circuit family $q:\: G \rightarrow \mathbb{Q}$ such that
\be
\xi(g) = \euler^{2\pii q(g)} \quad \forall g \in G.
\ee
Since we can switch between the black-box group and decomposed group encodings (see section \ref{sect:Encodings}), we can assume the elements of $G$ are treated as a vector in $\mathbb{Q}^{a+b+c+d}$. 

We wish to write the quadratic function $\xi(g)$ in the normal form given by theorem \ref{thm:Normal form of a quadratic function}, i.e. find $M,c,v$ as in theorem \ref{thm:Normal form of a quadratic function} such that
\begin{equation}
\xi(g)=\euler^{\pii \,\left(g^{\transpose} M g \: +  \: C^{\transpose} g  \: +  \: 2v^\transpose g\right)}.
\end{equation}
$q$, and hence $M$, $c$, and $v$, are rational by assumption. We will furthermore assume, as before, that the size and precision of the coefficients are upper bounded by some known constant $D$, i.e. each element of $M$ can be written as $M_{i,j} = \alpha_{i,j}/\beta_{i,j}$ for integers $\alpha_{i,j},\beta_{i,j}$ with absolute value no more than $D$.

To do this, let us first determine the entries of $M$. This can be done in the following manner: it should be straightforward to verify that
\begin{equation}
\xi(x+y)=\xi(x)\xi(y)\euler^{2\pii \,x^{\transpose} M y}
\end{equation}
for any $x,y \in G$, and therefore
\begin{equation}
x^{\transpose} M y \equiv q(x+y) - q(x) - q(y) \mbox{ mod } \Z.
\end{equation}

We can use this method to determine nearly all the entries of $M$ exactly, by taking $x$ and $y$ to be unit vectors $e_i$ and $e_j$; this would determine $M_{ij}$ up to an integer, i.e.
\begin{equation}
M_{i,j} = e_i^{\transpose} M e_j \equiv q(e_i+e_j) - q(e_i) - q(e_j) \mbox{ mod } \Z.
\end{equation}

This determines all entries of $M$ except for those in $M_{\Z\Z}$ and $M_{\T\T}$ (the other entries can be assumed to lie in $[0,1)$). To deal with $M_{\Z\Z}$ we take $x = \alpha^{-1}e_i$, and $y = e_j$, such that the coefficient $M(i,j)$ is in the submatrix $M_{\Z\Z}$ and $1/\alpha$ is an element of the circle group with $\alpha<2D$, where $D$ is the precision bound. We obtain an analogous equation
\begin{equation}
\left(\frac{e_i^{\transpose}}{\alpha} M e_j\right) \equiv \frac{M_{i,j}}{\alpha}  \equiv  q(\alpha^{-1}e_i+e_j) - q(\alpha^{-1}e_i) - q(e_j) \mbox{ mod } \Z,
\end{equation}
which allows us to determine $M_{i,j}$: since the number $M_{i,j}/\alpha$ is smaller than $1/2$ in absolute value, the coefficient is not truncated modulo 1. One can apply the same argument to obtain the coefficients of $M_{\T\T}$, choosing  $x = e_i$, and $y = \alpha^{-1} e_j$.

Once we determine all the entries of $M$ in this manner, we get immediately the vector $C$ as well (since $C(i) = c_iM(i,i)$). It is then straightforward to calculate the vector $\tilde{v}$. Thus we can efficiently find the normal form of $\xi(g)$ through polynomially many uses of the classical function $q$.

\section{Extending theorem \ref{thm:Simulation}  to the Abelian HSP setting}\label{app:Extending}

In this appendix, we briefly discuss that theorem \ref{thm:Simulation} (and some of the results that follow from this theorem) can be re-proven in the general hidden subgroup problem oracular setting that we studied in section \ref{sect:Abelian HSPs}.  This fact supports our view (discussed in the main text) that the oracle models in the HSP and in the black-box setting are  very close to each other.

Recall that the main result  in this section (theorem \ref{thm:HSP}) states that the quantum algorithm Abelian HSP is a normalizer circuits over a group of the form $\DProd{d}{m}\times \mathcal{O}$, where $\mathcal{O}$ is a group associated with the Abelian HSP oracle $f$ via the formula (\ref{eq:Oracular Group Operation}). The group $\mathcal{O}$ is not a black-box group, because no oracle to multiply in $\mathcal{O}$ was provided. However, we discussed at the end of section \ref{sect:Abelian HSPs} that one can use the hidden subgroup problem oracle to perform certain multiplications implicitly.

We show next that theorem \ref{thm:Simulation} can be re-casted in the HSP setting as ``\emph{the ability to decompose the oracular group $\mathcal{O}$  renders normalizer circuits over $\DProd{d}{m}\times \mathcal{O}$ efficiently classically simulable}''. To see this, assume a group decomposition table $(\alpha, \beta, A, B, c)$ is given. Then we know $\mathcal{O}\cong\Z_{c_1}\times \cdots \times \Z_{c_m}$. Let us now view the function  $\alpha(g)=(g,f(g))$ used in the HSP quantum algorithm as a  group automorphism of $G\times \mathcal \Z_{c_1}\times \cdots \times \Z_{c_m}$, where we decompose $\mathcal{O}$. Then, it is easy to check that $\begin{pmatrix}
1 & 0 \\
B & 1 \\
\end{pmatrix}$ is  a matrix representation of this map. It follows that the group decomposition table can be used to de-black-box the HSP oracle, and this fact allows us to adapt the proof of theorem \ref{thm:Simulation} step-by-step to this case.

We point out further that the extended Cheung-Mosca algorithm can be adapted to the HSP setting, showing that normalizer circuits over $G\times \mathcal{O}$ can be used to decompose $\mathcal{O}$. This follows from the fact that the function $f$ that we need to query to decompose $\mathbf{B}$ using the extended Cheung-Mosca algorithm (algorithm \ref{alg:group decomposition}) has precisely the same form as the HSP oracle. Using the HSP oracle as a subroutine in algorithm \ref{alg:group decomposition} (which we can query \emph{by promise}), the algorithm computes a group decomposition tuple for $\mathcal{O}$.

Finally, we can combine these last observations with theorem \ref{thm:Hidden Kernel Problem} and conclude that the problem of decomposing groups of the form $\mathcal{O}$ is classically polynomial-time equivalent to the Abelian hidden subgroup problem. The proof is analogous to that of theorem \ref{thm:Hidden Kernel Problem}.

\end{document}